\documentclass{article}
\usepackage{amssymb,amsmath,amsthm}
\usepackage{mathrsfs}
\usepackage{bbm}   
\usepackage{graphicx}
\usepackage[left=2.6cm,right=2.6cm]{geometry}
\usepackage{tikz}
\usepackage{subfigure}
\usepackage{multirow}

\definecolor{myurlcolor}{rgb}{0,0,0.7}
\definecolor{myrefcolor}{rgb}{0.8,0,0}
\usepackage{hyperref}
\hypersetup{colorlinks,
linkcolor=myrefcolor,
citecolor=myurlcolor,
urlcolor=myurlcolor}

\usepackage[all]{xy}
\input xy
\xyoption{all}
\xyoption{arc}
\xyoption{line}

\newcommand{\circledentry}[1]{%
  \tikz \node[anchor=south west, draw,circle, inner sep=1pt, minimum size=4mm,
    text height=2mm]{#1} ;}
\newcommand{\squaredentry}[1]{%
  \tikz \node[anchor=south west, draw, inner sep=3pt, minimum size=4mm,
    text height=2mm]{#1} ;}
\newcommand{\uncircledentry}[1]{%
  \tikz \node[anchor=south west, circle, inner sep=1pt, minimum size=4mm,
    text height=2mm]{#1} ;}

\newcommand{\beq}{\begin{displaymath}}
\newcommand{\eeq}{\end{displaymath}}
\newcommand{\beqn}{\begin{equation}}
\newcommand{\eeqn}{\end{equation}}
\newcommand{\beqa}{\begin{eqnarray*}}
\newcommand{\eeqa}{\end{eqnarray*}}
\newcommand{\beqna}{\begin{eqnarray}}
\newcommand{\eeqna}{\end{eqnarray}}

\newcommand{\eq}[1]{~(\ref{#1})}
\newcommand{\eps}{\varepsilon}
\newcommand{\N}{\mathbb{N}}
\newcommand{\Z}{\mathbb{Z}}
\newcommand{\R}{\mathbb{R}}
\newcommand{\C}{\mathbb{C}}
\newcommand{\F}{\mathbb{F}}

\renewcommand{\H}{\mathcal{H}}
\newcommand{\B}{\mathcal{B}}
\newcommand{\id}{\mathrm{id}}
\newcommand{\ra}{\rightarrow}

\newcommand{\lra}{\longrightarrow}

\newcommand{\tr}{\mathrm{tr}}

\newtheorem{prop}{Proposition}[section]
\newtheorem{thm}[prop]{Theorem}
\newtheorem{defn}[prop]{Definition}
\newtheorem{cor}[prop]{Corollary}
\newtheorem{lem}[prop]{Lemma}

\theoremstyle{definition} 
\newtheorem{ex}[prop]{Example}
\newtheorem{rem}[prop]{Remark}
\newtheorem{notn}[prop]{Notation}

\numberwithin{equation}{section}

\title{Tsirelson's Problem and Kirchberg's Conjecture}
\author{Tobias Fritz\\[.3cm]
\small{ICFO--Institut de Ciencies Fotoniques, Mediterranean Technology Park, 08860 Castelldefels (Barcelona), Spain}\\[.3cm]
\small{\texttt{tobias.fritz@icfo.es}}}

\begin{document}

\maketitle

\begin{abstract}
Tsirelson's problem asks whether the set of nonlocal quantum correlations with a tensor product structure for the Hilbert space coincides with the one where only commutativity between observables located at different sites is assumed. Here it is shown that Kirchberg's QWEP conjecture on tensor products of $C^*$-algebras would imply a positive answer to this question for all bipartite scenarios. This remains true also if one considers not only spatial correlations, but also spatiotemporal correlations, where each party is allowed to apply their measurements in temporal succession; we provide an example of a state together with observables such that ordinary spatial correlations are local, while the spatiotemporal correlations reveal nonlocality. Moreover, we find an extended version of Tsirelson's problem which, for each nontrivial Bell scenario, is equivalent to the QWEP conjecture. This extended version can be conveniently formulated in terms of steering the system of a third party. Finally, a comprehensive mathematical appendix offers background material on complete positivity, tensor products of $C^*$-algebras, group $C^*$-algebras, and some simple reformulations of the QWEP conjecture.
\end{abstract}

\setcounter{tocdepth}{1}
\tableofcontents

\section{Introduction}

\paragraph{Quantum correlations in composite quantum systems.} A composite physical system is a physical system which we think of as made out of parts. These parts may correspond to different degrees of freedom---like polarization and frequency of a photon---or be explicitly realized by spatially separated and physically distinct components. In any case, standard quantum theory posits that the Hilbert space of states modeling the total system is given by the tensor product $\H_A\otimes\H_B$, when a system is made out of component systems modeled by Hilbert spaces $\H_A$ and $\H_B$. 

This applies in particular to the study of quantum correlations, by which we mean the study of Bell inequalities and their quantum violations. In the case of a two-component system, the correlations take on the form of a conditional probability distribution
\beqn
\label{corrtensor}
P(a,b|x,y)=\langle\psi,(A_x^a\otimes B_y^b)\psi\rangle.
\eeqn
where $|\psi\rangle\in\H_A\otimes\H_B$ is the initial state of the total system and $A^x_a$ and $B^y_b$ are positive operators forming a POVM with outcome index $a$ (resp. $b$) for measurement setting $x$ (resp. $y$). See section~\ref{qcdefs} for more detail on this.

How is the postulate that the total Hilbert space is $\H_A\otimes\H_B$ justified from physical principles? Can we really be sure that this tensor product assumption is appropriate? In order to get an intuition for this, one should keep in mind that the quantum-mechanical description of atoms, photons and many other systems is an \emph{effective description}: it should in principle be derivable from the quantum field theory constituting the Standard Model\footnote{Of course, the Standard Model itself may ultimately be derivable from an underlying even more fundamental theory, but this is of little relevance here.}. Although at the present time, there is no known mathematically rigorous formulation of the Standard Model, we have at least some rigorous candidate frameworks for doing so in terms of the Wightman Axioms~\cite[II.1.2]{Haag} and the related Haag-Kastler axioms~\cite[III]{Haag} of Algebraic Quantum Field Theory (AQFT). With both axiom sets, causality is incorporated by postulating that observables localized on spacelike separated spacetime regions commute. However, a priori there is only a single large Hilbert space on which all field operators act, so no tensor product assumption seems to exist at this stage. As it turns out, for observable algebras localized in spacelike separated regions the existence of a tensor product splitting (more precisely, the ``split property'') can be \emph{derived} under certain reasonable conditions; see~\cite{RS,Yng} and the references therein. So as far as AQFT is concerned, the tensor product assumption is justified; but then again, since the Standard Model is not known to fit into this paradigm, let alone any potential theory of quantum gravity, deriving the tensor product assumption from the AQFT axioms may be a vain endevaour unrelated to actual physics.

Given this slightly dubious status of the tensor product assumption, one may wonder what happens to the set of quantum correlations upon relaxing the tensor product structure. One possibility of doing this is the \emph{commutativity assumption}: now we assume that there is a Hilbert space $\H$ for the total system, and all observables act on this total Hilbert space. The subsystems are defined by specifying observable algebras: these are assumed to be $C^*$-algebras $A\subseteq\B(\H)$ and $B\subseteq\B(\H)$ which mutually commute, $ab=ba$ for all $a\in A$, $b\in B$. In particular, this situation arises $\H=\H_A\otimes\H_B$ with $A=\B(\H_A)\otimes\mathbbm{1}$ and $B=\mathbbm{1}\otimes\B(\H_B)$, so that the tensor product assumption is a special case of the commutativity assumption. However, such a splitting does not always exist, and this is a recurrent theme in the theory of $C^*$-algebra tensor products; see section~\ref{tensor} and in particular example~\ref{minmax}.

Now in analogy to~(\ref{corrtensor}), with the commutativity assumption one can also write down quantum correlations
\beqn
\label{corrcomm}
P(a,b|x,y)=\langle\psi,A_x^a B_y^b\psi\rangle   
\eeqn
in which the commutativity assumption $[A^x_a,B^y_b]=0$ is relevant for ensuring that the imaginary part of this expectation value vanishes. It is easily verified that such correlations satisfy the usual no-signaling requirement, and we will study in the rest of this paper how the set of correlations of the form\eq{corrcomm} relates to those of the form\eq{corrtensor}.

There is another fundamental reason to consider the commutativity assumption as an alternative to the tensor product assumption. It is our point of view that the operation of forming a composite system $\H_A\otimes\H_B$ from its subsystems $\H_A$ and $\H_B$ should not be a fundamental structure in a physical theory. The point is that nature presents us with a huge quantum system which we observe and conduct experiments on, and in some ways this total system behaves as if it were composed of smaller parts. Hence it seems that the correct question would be ``When does a physical system behave like it were composed of smaller parts?'' rather than ``How do physical systems compose to composite systems?''. Note that this is in stark contrast to many other approaches to the foundations of quantum theory, in which the operation of forming a composite system from subsystems is a fundamental structure. This applies for example to categorical quantum mechanics~\cite{Coecke} and to certain approaches of reconstructing quantum mechanics from certain axioms on the probabilistic structure of the theory~\cite{Har},~\cite{Mas2}. So from our point of view, the tensor product operation should not be a fundamental structure of quantum theory, and hence we see the need to consider other structures pertaining to physical systems which potentially make the systems behave like they were composed of subsystems. Both the tensor product assumption and the commutativity assumption may be viewed as candidates for conditions of what it means for physical systems to be composed out of parts, and one of these may eventually be derivable from the other postulates of quantum theory via a physical analysis of what it means to be made out of subsystems.

\paragraph{Tsirelson's problem.} In this paper, we consider only those aspects of the tensor product assumption vs.~commutativity assumption problem which pertain to the study of quantum correlations and quantum violations of Bell inequalities. Upon fixing the number of observables for each party and the number of outcomes of each observable, the two assumptions~\ref{corrtensor} and~\ref{corrcomm} each give rise to a set of quantum correlations as a subset of all no-signaling conditional probability distributions. Calling these sets $\mathcal{Q}_{\otimes}$ and $\mathcal{Q}_{c}$, respectively, and taking $\overline{\mathcal{Q}}_\otimes$ to be the topological closure of $\mathcal{Q}_\otimes$, we arrive at:

\begin{quote}
\textbf{Tsirelson's problem.} Is $\overline{\mathcal{Q}}_{\otimes}=\mathcal{Q}_{c}$ or $\overline{\mathcal{Q}}_{\otimes}\neq\mathcal{Q}_{c}$?
\end{quote}

Of course, the answer to this question may in principle depend on the specific Bell scenario under consideration. We therefore consider the hypothesis:

\begin{quote}
\textbf{TP conjecture}: $\overline{\mathcal{Q}}_{\otimes}=\mathcal{Q}_{c}$ holds in all bipartite Bell scenarios with fixed finite number of observables per party and fixed finite number of outcomes per observable.
\end{quote}

At present, the TP conjecture is wide open, and very little is known besides some relatively simple observations. Firstly, $\mathcal{Q}_{\otimes}\subseteq\mathcal{Q}_{c}$ holds in all scenarios, since observables acting on separate tensor factors automatically commute, so that the tensor product assumption implies the commutativity assumption. Furthermore, correlations of the form\eq{corrcomm} can also be written in the form\eq{corrtensor} provided that the Hilbert space $\mathcal{H}$ is finite-dimensional; see e.g.~\cite{SW} for a short proof. Finally, there is also a positive answer to the TP conjecture in the case of the CHSH scenario (two settings and two outcomes for each party~\cite{CHSH}); see remark~\ref{CHSH} for why the CHSH scenario is ``too simple''.

What would be the implications of an answer to the TP conjecture? Clearly, a positive answer would be a nice justification for assuming quantum correlations to have the form\eq{corrtensor}; even if the analogous question in the multipartite case would still be open. A negative answer in terms of some correlations which are of the form\eq{corrcomm} but not of the form\eq{corrtensor} however would probably have a large impact since it would mean that some of the research done since the inception of quantum information theory, where one usually takes the tensor product splitting for granted, would not be applicable to these quantum correlations; a notable exception is~\cite{BFS}, where the formalism of smooth entropy has been studied under the commutativity assumption. Also, such a negative answer would certainly raise many more questions: could these correlations be physically realistic, despite the split property of AQFT~\cite{RS}? If so, would they also be experimentally accessible? Would they be more useful for quantum communication and computation than those of the form\eq{corrtensor}? Moreover, such a negative result would also provide a physically intuitive context in which infinite-dimensional Hilbert spaces cannot always be approximated by finite-dimensional ones; see also~\cite{WID}.

Finally, besides the fundamental and philosophical considerations described above, the TP conjecture also possesses a high theoretical significance for the characterization of quantum correlations. The reason is as follows: most, if not all, of the well-understood examples of quantum correlations are based on the tensor product assumption. On the other hand, most, if not all, upper bounds on the set of quantum correlations and on quantum violations of Bell operators actually use the commutativity assumption; in particular, this applies to the hierarchy of semidefinite programs characterizing quantum correlations~\cite{NPA}. So, will this dual strategy of bounding the set of quantum correlations from below by the tensor product assumption and bounding it from above by the commutative assumption converge to a unique definite set of quantum correlations? Or will there remain a gap? This is precisely Tsirelson's problem.

\paragraph{A bit of history.} One of the pioneers of the theory of quantum correlations is Boris Tsirelson, whose seminal papers, in particular~\cite{Tsirelson}, have initiated the study of quantum correlations and introduced methods from functional analysis. In that paper, Tsirelson stated that the tensor product assumption and the commutativity assumption were equivalent. While working on the hierarchy of semidefinite programs characterizing quantum correlations~\cite{NPA}, Navascu\'es, Pironio and Ac\'in noticed that this purported equivalence had not been proven by Tsirelson, so that they contacted him and requested a proof. This was how Tsirelson noticed that he could prove the equivalence only for systems with finite-dimensional Hilbert space, but not in the infinite-dimensional case. Therefore, he subsequently issued the question to the (discontinued) website on open problems in quantum information theory which was hosted at the Institute for Mathematical Physics at the University of Braunschweig\footnote{His problem statement is presently (March 2011) still retrievable from Tsirelson's website at \url{http://www.tau.ac.il/~tsirel/download/bellopalg.pdf}.}. Since then, Scholz and Werner have written a paper~\cite{SW} reformulating the problem in the language of operator systems and relating it to finite-dimensional approximability.

\paragraph{Kirchberg's QWEP conjecture.} The dichotomy between the tensor product assumption and the commutativity assumption also prevails in the theory of tensor products of $C^*$-algebras (see e.g.~\cite{KR2}). Given $C^*$-algebras $A$ and $B$, there are in general many different $C^*$-algebras which can be regarded as a $C^*$-algebraic tensor product of $A$ and $B$; as explained in more detail in appendix~\ref{tensor}, there is a ``minimal'' way and a ``maximal'' way to take the $C^*$-algebraic tensor product, which results in $C^*$-algebras respectively denoted by
\beqn
\label{minmaxposs}
A\otimes_{\min}B\quad\textrm{and}\quad A\otimes_{\max}B,
\eeqn
In general, these two tensor products are different, and there can be many others lying ``in between''. In certain cases, in particular for sufficiently small (``nuclear'') $C^*$-algebras, the two tensor products $A\otimes_{\min}B$ and $A\otimes_{\max}B$ turn out to be identical, which makes the $C^*$-algebraic tensor product unique for the pair $(A,B)$. Determining whether such a uniqueness occurs for a particular pair of $C^*$-algebras is often a very difficult problem if neither of them is nuclear. Kirchberg~\cite{Kir} (see~\cite{Oza} for a more recent review) has proposed the following as an open problem:

\begin{quote}
\textbf{QWEP conjecture:} $\quad C^*(\F_2)\otimes_{\min}C^*(\F_2)=C^*(\F_2)\otimes_{\max}C^*(\F_2)$.
\end{quote}

Here, $\F_2$ stands for the free group on two generators, while $C^*(\F_2)$ is the corresponding maximal group $C^*$-algebra~\cite{KR2}; see appendix~\ref{maxgroup} for some background material on maximal group $C^*$-algebras.

The QWEP conjecture derives its name from another formulation of the same question, also due to Kirchberg~\cite{Kir}. This different formulation asks whether every $C^*$-algebra is a \emph{Q}uotient of one having the \emph{W}eak \emph{E}xpectation \emph{P}roperty (a property which will not discuss in this paper; see~\cite[3.19]{Oza}). The importance of the QWEP conjecture manifests itself in the large number of open problems known to be equivalent to QWEP; just to mention a few, the list of equivalent open questions contains Connes' embedding problem~\cite{Con,Cap}, a problem on positivity of noncommutative polynomials~\cite{KS} or on tensor products of operator systems~\cite{Kavruk}. More recently, also connections to quantum information theory have been found~\cite{HM}. In this paper, we will also study the connection between the QWEP conjecture and a question in quantum information theory: Tsirelson's problem.

\paragraph{Quantum correlations and group $C^*$-algebras.} It is shown in this paper---and independently in~\cite{Pal}---that the following implication holds:
\beqn
\label{main}
\boxed{\textrm{QWEP conjecture $\Longrightarrow$ TP conjecture}}
\eeqn
We do not know whether the converse implication is also true; but we will formulate a variant of Tsirelson's problem fully equivalent to the QWEP conjecture. Results of this type are important in that they provide a physical interpretation of the QWEP conjecture. Thereby it becomes possible to try and attack this purely mathematical problem using physical intuition and physical principles: for example, one might try to look for a counterexample to the TP conjecture in terms of correlations of the form\eq{corrcomm} which provably violate some physical principle like Information Causality~\cite{IC} or another one known to hold for all quantum correlations with the tensor product assumption. By\eq{main}, this would then automatically yield a disproof of the QWEP conjecture.

So how does the correspondence\eq{main} come about? The basic idea is very simple and consists in replacing a projective $m$-outcome observable, given in terms of
\beq
\textrm{projections}\quad P_1,\ldots,P_m\quad\textrm{such that}\quad\sum_jP_j=\mathbbm{1},
\eeq
by the operator
\beq
U\equiv\sum_j e^{\frac{2\pi i j}{m}}P_j,
\eeq
which is a unitary of order $m$; basically, this replacement relabels the outcomes $1,\ldots,m$ by the $m$-th roots of unity $e^{\frac{2\pi i j}{m}}$. In this way, the specification of a projective observable with $m$ outcomes is equivalent to the specification of a unitary representation of the cyclic group\footnote{The group $\Z_m$ is defined to be the integers $\{0,\ldots,m-1\}$ with addition modulo $m$ as the group operation.} $\Z_m$ . It then follows that a specification of $k$ such observables is equivalent to the specification of a unitary representation of the free product group
\beqn
\label{freeprod}
\Gamma\equiv\underbrace{\Z_m\ast\ldots\ast\Z_m}_{k\textrm{ factors}}.
\eeqn
And then by the very definition of maximal group $C^*$-algebras, such a unitary representation is nothing but a representation of the maximal group $C^*$-algebra
\beq
C^*(\Z_m\ast\ldots\ast\Z_m).
\eeq
Now it should be plausible that the conjectures TP and QWEP are intimately related.

Moreover, we believe that these group $C^*$-algebras generally provide a useful and relevant framework for the classification of quantum correlations. When studying problems like whether a given set of nonlocal correlations admits a quantum-mechanical model, or determining the maximal quantum violation of a Bell inequality, one naturally has to deal with all theoretically possible quantum systems at once. What is required is a universal quantification over all Hilbert spaces, all (entangled) states on these, and all viable observable specifications for the parties involved. Such a universal quantification over all possible instances of mathematical structures seems like a tremendous task to deal with. However, as our theorems~\ref{qcminmaxpovm} and~\ref{bominmax} will show, by formulating these problems in the language of group $C^*$-algebras, they turn into questions about a \emph{single mathematical entity}. And although this reformulation in terms of a single mathematical entity only shifts the universal quantification involved into the definition of this entity, we nevertheless believe that this approach is very useful since it allows us to profitably apply the well-developed methods and results both from the general theory of $C^*$-algebras as well as from the theory of discrete groups and their unitary representations. For example, the hierarchy of semidefinite programs characterizing the set of quantum correlations~\cite{NPA} is essentially based upon theorem~\ref{qcminmaxpovm} (see remark~\ref{hierarchy}). Also, a suitable choice of language is always crucial for gaining deeper understanding of a problem. So besides presenting and proving our results, we hope to convince the reader that the language of $C^*$-algebra tensor products is a suitable framework for Tsirelson's problem, and for the study of quantum correlations in general\footnote{Compare~\cite{2D} for an application of the same ideas to the classification of temporal quantum correlations.}. For related approaches based on the languages of \emph{operator systems} and \emph{operator spaces}, see~\cite{SW},~\cite{Pal} and~\cite{OST}.

\paragraph{Variants of Tsirelson's problem.} Besides the sets of nonlocal quantum correlations, there are many other things one can study in order to understand both the power of the quantum-mechanical formalism and its limitations. We do so by defining two extensions of the concept of quantum correlations, both motivated by our $C^*$-algebraic picture, and formulate Tisrelson's problem for these.

The first extensions of the concept of quantum correlations is the notion of \emph{spatiotemporal quantum correlations}. Here, it is assumed that the measurements of both parties are projective and do not destroy the system, so that they can be applied in temporal succession. Since any local measurement necessarily decreases the entanglement contained in the shared bipartite state, it may be surprising that spatiotemporal correlations can nevertheless be stronger than ordinary spatial ones, as example~\ref{Wstate} demonstrates. The QWEP conjecture also implies a positive answer to the spatiotemporal variant of Tsirelson's problem.

The second extension of the concept of quantum correlations is defined in terms of \emph{steering}. As originally formulated already by Schr\"odinger~\cite{Schr35}, this is the phenomenon that Alice's measurement changes the state of Bob's system, given that one postselects on one specific outcome of Alice's measurement. Our version of steering considers the case where both Alice and Bob steer the system of a third party; one can view this as replacing, in the definition of quantum correlations, the ordinary classical probabilities by unnormalized density matrices. We formulate a version of Tsirelson's problem also in this case and prove it to be \emph{equivalent} to the QWEP conjecture, \emph{for each Bell scenario separately} (except CHSH).

Figure~\ref{implications} provides an overview of all the conjectures considered in this paper, the implications which we are able to prove between them, and the corresponding references to the main text.

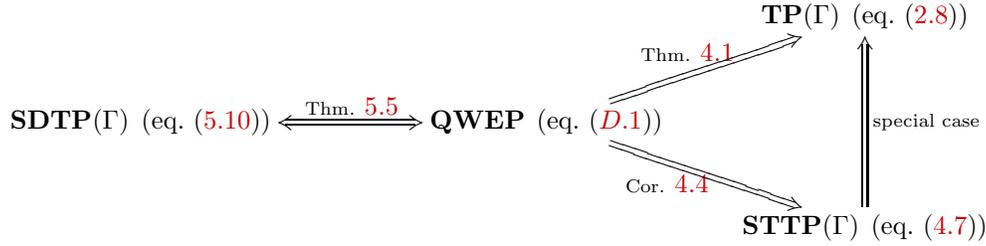
\begin{figure}
\label{implications}
\begin{center}
\beq
\xymatrix{ &&&   \mathbf{TP}(\Gamma)\:\:(\textrm{eq.}\eq{TP}) \\
\mathbf{SDTP}(\Gamma)\:\:(\textrm{eq.}\eq{SDTP})\ar@{<=>}[rr]^(.52){\textrm{Thm. }\ref{QWEPtoSDTP}} && \mathbf{QWEP}\:\:(\textrm{eq.}\eq{qwep})\ar@{=>}[rd]_(.43){\textrm{Cor. }\ref{QWEPtoSTTP}}\ar@{=>}[ru]^{\textrm{Thm. }\ref{QWEPtoTP}}  \\
&&& \mathbf{STTP}(\Gamma)\:\:(\textrm{eq.}\eq{STTP})\ar@{=>}[uu]_{\textrm{special case}} }
\eeq
\end{center}
\caption{The implications between our versions of Tsirelson's problem and the QWEP conjecture. $\Gamma$ stands for the particular Bell scenario (number of measurement settings and the number of their outputs) to be considered. The implication from $\mathbf{SDTP}(\Gamma)$ to $\mathbf{QWEP}$ does not hold when $\Gamma$ is the CHSH scenario. With this sole exception, all implications hold for all $\Gamma$.}
\end{figure}

\paragraph{Structure and summary of this paper.}
This article is structured into a main part pursuing the essential lines of thought and an appendix containing mostly standard mathematical background material. Section~\ref{qcdefs} starts by introducing bipartite Bell scenarios and introduces the two sets of quantum correlations to be considered in the sequel. Tsirelson's problem is stated again in a more formal way. Section~\ref{Tsipovm} then describes the connection between quantum correlations and group $C^*$-algebras; in particular, this provides a new proof of the fact that all extremal quantum correlations in the CHSH scenario can be achieved with two qubits. Section~\ref{mainsec} proceeds to to state our main result\eq{main} as theorem~\ref{QWEPtoTP} and continues by introducing spatiotemporal quantum correlations. The ensuing orollary~\ref{QWEPtoSTTP} states that a positive answer to the QWEP conjecture would also imply a positive answer to the spatiotemporal version of Tsirelson's problem. Example~\ref{Wstate} considers a particular state together with projective measurements in the CHSH scenario such that the resulting spatiotemporal correlations are nonlocal, although the spatial correlations alone are local. In section~\ref{reverseimp}, we consider another variant of Tsirelson's problem and show it to be \emph{equivalent} to the QWEP conjecture; this variant is defined in terms of the ability of Alice and Bob to steer the system of a third party.

We have included the mathematical appendix in order to achieve a reasonable level of self-containedness. This mathematical background material is necessary for proving our results in the main text, where references to the statements in the appendix have been included. Appendix~\ref{ucp} begins by introducing unital completely positive maps between $C^*$-algebras and discusses some of their properties. Appendix~\ref{tensor} addresses the crucial issue of $C^*$-algebra tensor products and provides a detailed exposition of these. Appendix~\ref{maxgroup} treats maximal group $C^*$-algebras, which are, in a way, the main theme of the present paper. In appendix~\ref{appqwep}, we consider the QWEP conjecture and introduce a simple reformulation where we replace the free group $\F_2$ by a free product of finite cyclic groups.

Despite the considerable length of the present paper, there is lots of material which we have not touched upon. A particularly relevant omission is the discussion of questions of approximability by finite-dimensional representations (see~\cite{SW} for this in relation to Tsirelson's problem, and~\cite{Brown} in relation to the QWEP conjecture). In fact, we have nothing sufficiently original to say about this, and hence do not touch upon these issues here except by reemphasizing that Tsirelson's problem is only an issue if one considers quantum systems with infinite-dimensional state spaces.

\paragraph{Notation and conventions.} Unfortunately, the commonly used notation in quantum information theory and quantum nonlocality theory is radically different from the one established in $C^*$-algebra theory. In this paper, we mostly follow the latter, which seems more convenient for our purposes. For example, one reason to restrain from using Dirac notation is the possible ambiguity of expressions like $\langle\psi|A|\psi\rangle$: when $A$ is not self-adjoint, this value depends on whether $A$ acts to the left or to the right. Also, instead of using the slightly clumsy notation $\tr(\rho A)$ for the expectation value of an observable $A$ on the density matrix $\rho$, we rather write $\rho(A)$ for the same quantity, where $\rho$ now is to be interpreted as a state in the $C^*$-algebraic sense, i.e.~as a linear functional on the $C^*$-algebra of observables. For us, ``state'' means ``state in the $C^*$-algebraic sense'', and, following the philosophy of Algebraic Quantum Mechanics~\cite{AQT}, this is what we have tried to reserve the word ``state'' for; concrete physical states are taken to be unit vectors in a Hilbert space, and consequently we use the term \emph{(unit) vector} for those. All our $C^*$-algebras are assumed unital, even when this is not explicitly mentioned. Here is an overview over the notation used throughout the paper:
\begin{itemize}
\item $\H$ is any Hilbert space, no separability assumption is made. $\psi\in\H$ ranges over all unit vectors.
\item In the main part of the paper, $A^x_a$ and $B^y_b$ are the POVM components of Alice and Bob (see section~\ref{qcdefs}).
\item In the appendix, $A$ and $B$ are any unital $C^*$-algebras.
\item For a $C^*$-algebra $A$, we write $\mathscr{S}(A)$ for the space of states on $A$, equipped with the weak $*$-topology. 
\item $\rho$ always stands for a state on a $C^*$-algebra.
\item $M_n(A)$ for a $C^*$-algebra $A$ is the $C^*$-algebra of $n\times n$-matrices over $A$. We also write it as $M_n(\C)\otimes A$ (see example~\ref{matrixA}).
\item The asterisk ``$\ast$'' denotes a free product of groups. For $C^*$-algebras, we only consider unital free products written as ``$\ast_1$''.
\end{itemize}

\paragraph{Relevant literature for background material.} We hope that the present paper is sufficiently self-contained in order to be accessible with some elementary background in $C^*$-algebra theory, up to familiarity with free products of $C^*$-algebras, and some interest in the problem of characterizing the set of nonlocal quantum correlations. But especially on the mathematical side of the story, there is a considerable literature on issues related to the questions discussed here. Hence, let us mention some of the advanced literature, thematically ordered:
\begin{itemize}
\item Tsirelson's problem:~\cite{SW} gives an overview and some first results, albeit using a somewhat different mathematical formalism. This work has been continued in~\cite{Pal}, where results similar to ours have been obtained. Follow-up works include~\cite{NCPV} and~\cite{Pru}.
\item Kirchberg's QWEP conjecture:~\cite{Kir} is the original paper, whereas~\cite{Oza} is a more recent review.~\cite{Pis} discusses the QWEP conjecture from the point of view of operator spaces.
\item Connes' embedding conjecture: From the many mathematical problems equivalent to the QWEP conjecture, this is the most well-known one.~\cite{Con} is the original paper, whereas~\cite{Cap} is a more recent review. The equivalence to the QWEP conjecture is proven in~\cite{Kir} and~\cite{Oza}.
\end{itemize}

\paragraph{Acknowledgments.} Carlos Palazuelos has kindly informed me about~\cite{Pal}, which contains very similar results and was publicized on the same day as an earlier version of this paper. Furthermore, I would like to thank Antonio Ac\'in, Matilde Marcolli, Miguel Navascu\'es and Carlos Palazuelos for interesting discussions and the referee for some constructive criticism. This work was partly carried out within the IMPRS graduate program at the Max Planck Institute for Mathematics and has thereafter been supported by the EU STREP QCS.

\section{Bipartite nonlocality scenarios}
\label{qcdefs}

Suppose that there are two experimenters, commonly dubbed \emph{Alice} and \emph{Bob} and referred to as \emph{parties} or \emph{sites}, located in spatially separated physics labs. From a common source, they each receive a quantum system on which they conduct their measurements. They are both free to measure any one out of a fixed number $k\in\N$ of observables on their system, with each observable having a fixed number $m\in\N$ of possible outcomes. We write $[k]$ for the set $\{1,\ldots,k\}$ which indexes the observables, and $[m]=\{1,\ldots,m\}$ for the set which indexes the possible outcomes. The pair of natural numbers $\Gamma\equiv(k,m)$ specifies a \emph{Bell scenario} $\Gamma$.\footnote{One can also consider scenarios where different observables have different numbers of outcomes, and/or scenarios where the number of choices for Alice differs from the number of choices for Bob. For the ease of notation and readability, we do not consider such scenarios explicitly, but all of our results generalize to these in the obvious way.} We take $\Gamma$ as a shorthand notation for the pair $(k,m)$; this notation will be explained by\eq{gammanot}.

\paragraph{Quantum correlations with the tensor product assumption.} We start by explaining in detail what quantum correlations in the Bell scenario $\Gamma$ with the tensor product assumption are. With the tensor product assumption, Alice's system is described by a Hilbert space $\mathcal{H}_A$, while Bob's system lives on a Hilbert space $\mathcal{H}_B$. The state of the total system is then a vector in the \emph{composite Hilbert space}
\beq
\mathcal{H}_A\otimes\mathcal{H}_B.
\eeq

For each measurement setting $x\in[k]$, Alice has access to a POVM $\{A^x_1,\ldots,A^x_m\}$, where each $A^x_a$ is a positive operator on $\mathcal{H}_A$. Similarly, for each measurement setting $y\in[k]$, Bob has access to a POVM $\{B^y_1,\ldots,B^y_m\}$, where each $B^y_b$ is a positive operator on $\mathcal{H}_B$. Then for some joint initial state
\beq
\psi\in\mathcal{H}_A\otimes\mathcal{H}_B
\eeq
the outcome probabilties for a joint measurement are given by the fundamental formula
\beqn
\label{tensorcorrs}
\boxed{P(a,b|x,y)=\langle\psi,(A^x_a\otimes B^y_b)\psi\rangle}
\eeqn
We use the assumption that the initial state is pure mainly for ease of notation. Since we do not take the Hilbert space dimension to be fixed, a mixed state can always be purified by adding an ancilla either on Alice's side or on Bob's side (or on both sides), so the assumption of purity can be made without loss of generality.

The set of all conditional probability distributions $P(a,b|x,y)$ which have a representation\eq{tensorcorrs} form a subset of $\R^{k^2m^2}$. In the following, we will denote this set by $\mathcal{Q}_{\otimes}(\Gamma)$. One may regard it as the fundamental object of study in the theory of quantum nonlocality for the Bell scenario $\Gamma$. As a first basic observation about $\mathcal{Q}_{\otimes}(\Gamma)$, one can note the following:

\begin{lem}
\label{Qconvex}
$\mathcal{Q}_{\otimes}(\Gamma)$ is convex.
\end{lem}

\begin{proof}
Suppose that $\mathcal{Q}_{\otimes}(\Gamma)$ contains the point
\beq
P_1(a,b|x,y)=\langle\psi_1,(A^x_{a,1}\otimes B^y_{b,1})\psi_1\rangle
\eeq
where the vector $\psi_1$ and the observables live on some tensor product Hilbert space $\H_{A,1}\otimes\H_{B,1}$, as well the point
\beq
P_2(a,b|x,y)=\langle\psi_2,(A^x_{a,2}\otimes B^y_{b,2})\psi_2\rangle
\eeq
where the vector $\psi_2$ and the observables live on some tensor product $\H_{A,2}\otimes\H_{B,2}$. Let $\lambda\in(0,1)$ be some coefficient. Then the claim is that the conditional probability distribution
\beqn
\label{Pconvcomb}
\lambda\cdot P_1+(1-\lambda)\cdot P_2
\eeqn
also lies in $\mathcal{Q}_{\otimes}$, i.e.~that it also can be written in the form\eq{tensorcorrs}. To this end, consider the observables 
\beq
A^x_a\equiv A^x_{a,1}\oplus A^x_{a,2},\qquad B^y_b\equiv B^y_{b,1}\oplus B^y_{b,2}
\eeq
which act on the direct sums $\H_{A,1}\oplus\H_{A,2}$ and $\H_{B,1}\oplus\H_{B,2}$, respectively. The tensor product of these direct sum Hilbert spaces can be decomposed as
\begin{align}
\begin{split}
\label{Hdecomp}
(\H_{A,1}\oplus\H_{A,2})\otimes(\H_{B,1}\oplus\H_{B,2})\:\:=&\phantom{\oplus}\:\:(\H_{A,1}\otimes\H_{B,1})\\&\oplus(\H_{A,1}\otimes\H_{B,2})\\&\oplus(\H_{A,2}\otimes\H_{B,1})\\&\oplus(\H_{A,2}\otimes\H_{B,2})
\end{split}
\end{align}
so that one can consider on this total Hilbert space the unit vector
\beq
\psi\equiv\sqrt{\lambda}\,\psi_1\oplus 0\oplus 0\oplus\sqrt{1-\lambda}\,\psi_2.
\eeq
A short calculation then verifies that this represents\eq{Pconvcomb} in the form\eq{tensorcorrs}.
\end{proof}

\begin{rem}
\label{closure}
It is unclear whether $\mathcal{Q}_{\otimes}(\Gamma)$ is a \emph{closed} convex subset of $\R^{k^nm^2}$. In other words, if $P(a,b|x,y)$ can be approximated arbitrarily well by conditional probability distributions of the form\eq{tensorcorrs}, is it then itself also of the form\eq{tensorcorrs}? Clearly this question is irrelevant for any practical purposes, since in practice $P(a,b|x,y)$ will only be known to limited accuracy anyway, so one may regard the potential non-equality between $\mathcal{Q}_{\otimes}(\Gamma)$ and its closure $\overline{\mathcal{Q}_{\otimes}(\Gamma)}$ as an issue of mathematical pedantry. From the mathematical perspective however, the upcoming theorem~\ref{qcminmaxpovm} can be interpreted as stating that the closure $\overline{\mathcal{Q}_{\otimes}(\Gamma)}$ is ``nicer'' than $\mathcal{Q}_{\otimes}(\Gamma)$ itself, as it allows a succinct characterization in $C^*$-algebraic terms. So in this paper, we always work with $\overline{\mathcal{Q}_{\otimes}(\Gamma)}$, disregarding the question whether this set coincides with $\mathcal{Q}_\otimes(\Gamma)$ itself or not.
\end{rem}

\paragraph{Quantum correlations with the commutativity assumption.} Here, there is only a single Hilbert space $\mathcal{H}$ which contains the initial state $\psi\in\H$ (again taken to be pure without loss of generality) and on which \emph{both} Alice's and Bob's POVMs act. Hence we have operators
\beq
A^x_a,B^y_b\in\B(\H),\qquad A^x_a,B^y_b\geq 0,\qquad\sum_a A^x_a=\mathbbm{1}\quad\forall x,\qquad\sum_b B^y_b=\mathbbm{1}\quad\forall y.
\eeq
The relevant assumption now is commutativity of Alice's observables with Bob's observables,
\beqn
\label{commu}
A^x_aB^y_b=B^y_bA^x_a\qquad\forall x,y\in[k],\:a,b\in[m].
\eeqn
This implies in particular that Alice's and Bob's observables are jointly measurable. The outcome probabilities for such a joint measurement on the unit vector $\psi\in\H$ then take on the form
\beqn
\label{commcorrs}
\boxed{P(a,b|x,y)=\langle\psi,A^x_aB^y_b\psi\rangle}
\eeqn
and the commutativity assumption is relevant for assuring that $A^x_aB^y_b$ is a hermitian operator, so that this joint outcome probability is guaranteed to be a real number. The set of conditional probability distributions $P(a,b|x,y)$ which can be written in this form is another subset of $\R^{k^2m^2}$, which we will denote by $\mathcal{Q}_c$, with the subscript standing for ``$c$ommuting''. We will see in theorem~\ref{qcminmaxpovm} that $\mathcal{Q}_c$ is closed and convex.

\paragraph{Other possible assumptions?} It seems conceivable in principle that there might even be more than these two ways to define sets of quantum correlations. One approach could be to try and replace the commutativity assumption by something even weaker, for example by the assumption of \emph{joint measurability} in the sense of~\cite{HW}. In this case, Alice and Bob are both assumed to have POVMs as in the previous paragraph. But instead of commutativity\eq{commu}, we only assume that for every pair of measurement choices $x,y$, there is a doubly indexed POVM
\beq
\Pi^{x,y}_{a,b}\in\B(\H),\qquad\Pi^{x,y}_{a,b}\geq 0,\qquad\sum_{a,b}\Pi^{x,y}_{a,b}=\mathbbm{1},
\eeq
which reduces to the marginal observables as
\beqn
\label{jointmeasmarg}
A^x_a=\sum_b\Pi^{x,y}_{a,b}\quad\forall x,y,a;\qquad B^y_b=\sum_a\Pi^{x,y}_{a,b}\quad\forall x,y,b.
\eeqn
By this requirement, it is clear that the outcome probabilities for the joint measurements
\beqn
\label{jointmeas}
P(a,b|x,y)=\langle\psi,\Pi^{x,y}_{a,b}\psi\rangle
\eeqn
are automatically no-signaling.

On the other hand, it is quite clear that all no-signaling correlations $P(a,b|x,y)$ can be written in the form\eq{jointmeas}: one can simply take $\mathcal{H}=\C$ and $\Pi^{x,y}_{a,b}=P(a,b|x,y)$. Hence, joint measurability does not lead to a sensible set of quantum correlations. See however~\cite{CSW} for a simple set of conditions on the operators $\Pi^{x,y}_{a,b}$ which guarantee the existence of the marginals\eq{jointmeasmarg} together with their commutativity\eq{commu}.

\paragraph{Statement of Tsirelson's problem.}

Since operators acting on separate tensor factors necessarily commute, it is obvious that quantum correlations with the tensor product assumption can directly also be written as quantum correlations with the commutativity assumption. Hence, we certainly have the inclusion $\mathcal{Q}_\otimes(\Gamma)\subseteq\mathcal{Q}_c(\Gamma)$. And since $\mathcal{Q}_c(\Gamma)$ is closed (see theorem~\ref{qcminmaxpovm}), we also have $\overline{\mathcal{Q}_\otimes(\Gamma)}\subseteq \mathcal{Q}_c(\Gamma)$. So, can $\mathcal{Q}_c(\Gamma)$ be bigger than $\overline{\mathcal{Q}_\otimes(\Gamma)}$? Can there be quantum correlations with commuting observables between the sites which cannot be approximated by quantum correlations where the total Hilbert space is the tensor product of the Hilbert spaces at each site? We take this question as the definition of \emph{Tsirelson's problem} for the Bell scenario $\Gamma$, or $\mathbf{TP}(\Gamma)$:
\beqn
\label{TP}
\boxed{\mathbf{TP}(\Gamma):\quad\overline{\mathcal{Q}_\otimes(\Gamma)}\stackrel{?}{=}\mathcal{Q}_c(\Gamma)}
\eeqn
The original formulation of Tsirelson's problem~\cite{SW} is stated without the closure operation on the left-hand side, i.e.~as the question whether quantum correlations with the commutativity assumption can always be written \emph{exactly} (without approximation) as quantum correlations with the tensor product assumption. However as discussed in remark~\ref{closure}, the distinction between $\mathcal{Q}_\otimes(\Gamma)$ and its closure $\overline{\mathcal{Q}_\otimes(\Gamma)}$ is pure mathematical pedantry and irrelevant for experiments. It seems conceivable that $\overline{\mathcal{Q}_\otimes(\Gamma)}=\mathcal{Q}_c(\Gamma)$ but $\mathcal{Q}_\otimes(\Gamma)\neq\mathcal{Q}_c(\Gamma)$; this would mean that there are some Bell inequalities for which the maximal quantum value can only be achieved with the commutativity assumption, but quantum correlations with the tensor product assumption can get arbitrarily close to this value. We would regard this as an \emph{affirmative} answer to Tsirelson's problem.

Also, it should be pointed out again that any point potentially lying in $\mathcal{Q}_c(\Gamma)\setminus\mathcal{Q}_\otimes(\Gamma)$ would require an infinite-dimensional Hilbert space for its quantum-mechanical realization~\cite{SW}.

Finally, for some comments on the multipartite analogue of this problem, see remark~\ref{multipart}.

\section{From universal $C^*$-algebras to quantum correlations}
\label{Tsipovm}

As mentioned in the introduction, the problem of determining the maximal violation of a Bell inequality, or describing the set of quantum correlations, involves a universal quantification over all Hilbert spaces, all choices of observables on each Hilbert space, and all unit vectors in these Hilbert spaces. It may not come as a surprise that this problem is very difficult in general. In this section, we offer a reformulation of this in terms of universal $C^*$-algebras, where the universal quantification over the Hilbert spaces and over the observables becomes redundant, or rather hidden inside the definition of these universal $C^*$-algebras. This allows for a mathematically elegant reformulation of the problem, albeit at the cost of introducing an additional level of abstraction. As a benefit, the ample results and techniques from $C^*$-algebra theory and group theory become available for the study of quantum correlations.

We believe that the formalism presented here provides the most natural formulation of sets of quantum correlations and quantum values of Bell inequalities.

We use the same notation as in the previous section.

\paragraph{Observable specifications.} We begin with some simple reformulations of the concept of POVM. To this end, we write $e_a$ with $a=1,\ldots,m$ for the standard basis of $\C^m$, and consider $\C^m$ as a commutative $C^*$-algebra with respect to componentwise multiplication and componentwise complex conjugation. ($\C^m$ is canonically isomorphic to the $C^*$-algebra of functions on $m$ isolated points.) Moreover, we write $e_a^x$ for the standard basis vector $e_a$ in the $x$th factor of the $k$-fold unital free product $\C^m\ast_1\ldots\ast_1\C^m$.

For the relevant background material on ucp maps, see appendix~\ref{ucp}.

\begin{prop}
\begin{enumerate}
\item\label{POVMa} For any $m$-outcome POVM $\left\{A_1,\ldots,A_m\right\}$ in $\mathcal{B}(\mathcal{H})$, there is a ucp map 
\beq
\Phi:\C^m\lra\mathcal{B}(\mathcal{H})
\eeq
such that
\beqn
\label{POVMucp}
A_a\equiv\Phi(e_a),\qquad a=1,\ldots,m.
\eeqn
Conversely, this equation defines a POVM for any such ucp map $\Phi$.
\item\label{POVMb} For any $k$-tuple of $m$-outcome POVMs $\{A_1^x,\ldots,A_m^x\}$, $x\in[k]$, on $\mathcal{B}(\mathcal{H})$, there is a ucp map
\beq
\Phi:\underbrace{\C^m\ast_1\ldots\ast_1\C^m}_{k\,\textrm{ factors}}\lra\mathcal{B}(\mathcal{H})
\eeq
such that
\beqn
\label{POVMsucps}
A_a^x\equiv\Phi(e_a^x).
\eeqn
Conversely, this equation defines a $k$-tuple of POVMs for any such ucp map $\Phi$.
\end{enumerate}
\label{POVMchar}
\end{prop}

\begin{proof}
\begin{enumerate}
\item[\ref{POVMa}] Since $e_a\geq 0$ and $\sum_a e_a=\mathbbm{1}$ in $\C^m$, it is clear by positivity and unitality of $\Phi$ that any $A_a$ of the form\eq{POVMucp} are the components of an $m$-outcome POVM. On the other hand, for a given POVM the condition\eq{POVMucp} defines a map $\Phi$ since there is a unique linear extension to all of $\C^m$. This extension is unital since $\Phi(\mathbbm{1})=\Phi(\sum_a e_a)=\sum_a A_a=\mathbbm{1}$. A similar argument shows positivity. And by lemma~\ref{commucp} about positive maps on commutative $C^*$-algebras, this $\Phi$ is then automatically ucp.
\item[\ref{POVMb}] Again it is clear that for given $\Phi$, the assignment\eq{POVMsucps} defines a family of POVMs. Conversely for a given family of POVMs, thanks to part~\ref{POVMa} we already know that each POVM in the family induces a ucp map
\beq
\Phi^x:\C^m\ra\mathcal{B}(\mathcal{H}).
\eeq
The assertion hence follows from a recursive application of corollary~\ref{extfreeprod}.
\end{enumerate}
\end{proof}

\begin{rem}
\label{fourier}
The discrete Fourier transform relates this to maximal group $C^*$-algebras (see appendix~\ref{maxgroup}) as follows. Let $u$ be the generator of the cyclic group $\Z_m$. Then the assignment
\beqn
\label{fourier2}
u\mapsto\sum_{a=1}^m\exp\left(2\pi i\frac{a}{m}\right) e_a
\eeqn
defines an isomorphism $C^*(\Z_m)\cong\C^m$, the \emph{discrete Fourier transform}. The $m$-fold free product of this identification yields an isomorphism
\beqn
\label{Cfreeprod}
C^*(\Z_m)\ast_1\ldots\ast_1 C^*(\Z_m)\cong\C^m\ast_1\ldots\ast_1\C^m.
\eeqn
Since taking maximal group $C^*$-algebras is a left adjoint functor by the universal property (proposition~\ref{maxgroupunivprop}), it preserves coproducts; in other words, it is irrelevant whether one takes the free product on the level of groups or on the level of $C^*$-algebras:
\beq
C^*(\Z_m\ast\ldots\ast\Z_m)\cong C^*(\Z_m)\ast_1\ldots\ast_1 C^*(\Z_m)
\eeq
In total, the discrete Fourier transform\eq{fourier2} implements an isomorphism
\beqn
\label{freefourier}
C^*(\Z_m\ast\ldots\ast\Z_m)\cong\C^m\ast_1\ldots\ast_1\C^m. 
\eeqn
\end{rem}

\begin{notn}
For this free product of cyclic groups, we also use the shorthand notation
\beqn
\label{gammanot}
\Gamma\equiv\underbrace{\Z_m\ast\ldots\ast\Z_m}_{k\textrm{ factors}}.
\eeqn
The overloading of the symbol ``$\Gamma$'' standing both for this group and for the specification of the Bell scenario with $k$ observables having $m$ outcomes is deliberate: in this way, the $C^*$-algebra $C^*(\Gamma)$ can be interpreted as either the maximal $C^*$-algebra of the group $\Gamma$, or, as well will see in the following theorem, the $C^*$-algebra relevant for describing the set of quantum correlations in the Bell scenario $\Gamma$. This identification generally suggests to use group presentations as a notation for specifying Bell scenarios: e.g.~$\Gamma\equiv\Z_2\ast\Z_3$ would correspond to a scenario where each party may choose between one $2$-outcome measurement and one $3$-outcome measurement.
\end{notn}

We can now formulate the characterization of quantum correlations in terms of $C^*$-algebra tensor products. The notation $\mathcal{Q}_\otimes$ and $\mathcal{Q}_c$ is as introduced in the previous section.

\begin{prop}
\label{qcminmaxpovm}
Let $\Gamma$ be any Bell scenario. Then a given conditional probability distribution $P(a,b|x,y)$ for the scenario $\Gamma$ is in the set\ldots
\begin{enumerate}
\item\label{qmma} \ldots $\overline{\mathcal{Q}_\otimes(\Gamma)}$ if and only if there is a $C^*$-algebraic state $\rho\in\mathscr{S}(C^*(\Gamma)\otimes_{\min} C^*(\Gamma))$ such that
\beqn
\label{qrepmin}
P(a,b|x,y)=\rho(e^x_a\otimes e^y_b).
\eeqn
\item\label{qmmb} \ldots $\mathcal{Q}_c(\Gamma)$ if and only if there is a $C^*$-algebraic state $\rho\in\mathscr{S}(C^*(\Gamma)\otimes_{\max} C^*(\Gamma))$ such that
\beqn
\label{qrepmax}
P(a,b|x,y)=\rho(e^x_a\otimes e^y_b).
\eeqn
Furthermore, $\mathcal{Q}_c(\Gamma)$ is a closed and convex set, and coincides with the set of quantum correlations attainable by projective measurements satisfying the commutativity assumption\eq{commu}.
\end{enumerate}
\end{prop}

\begin{proof}
\begin{enumerate}
\item[\ref{qmma}] Suppose that $P(a,b|x,y)$ is quantum with the tensor product assumption, so that there are Hilbert spaces $\H_A$ and $\H_B$, a unit vector $\psi\in\H_A\otimes\H_B$ and observables $A^x_a$ and $B^y_b$ such that
\beqn
\label{Prep}
P(a,b|x,y)=\langle\psi,(A^x_a\otimes B^x_b)\psi\rangle.
\eeqn
Then as outlined before, Alice's POVMs define a ucp map $\Phi_A:C^*(\Gamma)\ra\B(\H_A)$ such that\eq{POVMsucps} holds, and likewise for Bob in terms of $\Phi_B:C^*(\Gamma)\ra\B(\H_B)$. But then by proposition~\ref{minproducp} on minimal tensor products of ucp maps, also
\beq
\Phi_A\otimes_{\min}\Phi_B:C^*(\Gamma)\otimes_{\min} C^*(\Gamma)\lra\B(\H_A\otimes\H_B),\qquad \gamma_A\otimes\gamma_B\mapsto \Phi_A(\gamma_A)\otimes\Phi_B(\gamma_B)
\eeq
is a ucp map, so that
\beq
\rho(\gamma_A\otimes\gamma_B)\equiv\langle\psi,(\Phi_A(\gamma_A)\otimes_{\min}\Phi_B(\gamma_B))\psi\rangle \qquad\forall\gamma_A,\gamma_B\in C^*(\Gamma)
\eeq
defines a state on $C^*(\Gamma)\otimes_{\min}C^*(\Gamma)$. By construction, this state satisfies\eq{qrepmin}.

For the converse implication, fix first any faithful representation $C^*(\Gamma)\subseteq\B(\H)$ for some Hilbert space $\H$. Then, the $e^x_a$ are concrete operators in $\B(\H)$. And by the definition of the minimal tensor product, $C^*(\Gamma)\otimes_{\min}C^*(\Gamma)$ is exactly the $C^*$-algebra generated by the joint observables $e^x_a\otimes e^y_b$ acting on $\H\otimes\H$. Now suppose that $P(a,b|x,y)$ is given in terms of a state $\rho\in\mathscr{S}(C^*(\Gamma)\otimes_{\min} C^*(\Gamma))$ satisfying\eq{qrepmin}. By proposition~\ref{vectorstates} on the density of vector states, one can hence find, for every $\eps>0$, a finite collection of unit vectors $\xi_1,\ldots,\xi_n\in\H\otimes\H$ together with non-negative coefficients $\lambda_i$ summing to $1$ such that
\beq
\bigg|\rho(e^x_a\otimes e^y_b)-\sum_j\lambda_j\langle\xi_j,(e^x_a\otimes e^y_b)\xi_j\rangle\bigg|<\eps\:\qquad\forall x,y,a,b.
\eeq
Hence with $A^x_a\equiv e^x_a$ and $B^y_b\equiv e^y_b$, the given conditional probability distribution $P(x,y|a,b)$ can be approximated arbitrarily well by quantum correlations with the tensor product assumption coming from a mixed state $\sum_j\lambda_j\langle\xi_j,\,\cdot\,\xi_j\rangle$. 

\item[\ref{qmmb}] Suppose that $P(a,b|x,y)$ is quantum with the commutativity assumption, so that there is a Hilbert space $\H$, a unit vector $\psi\in\H$ and observables $A^x_a\in\B(\H)$ and $B^y_b\in\B(\H)$ satisfying the commutativity assumption\eq{commu} such that
\beq
P(a,b|x,y)=\langle\psi,A^x_a B^x_b\psi\rangle.
\eeq
Then again, Alice's POVMs define a ucp map $\Phi_A:C^*(\Gamma)\ra\B(\H)$ such that\eq{POVMsucps} holds, and likewise for Bob in terms of $\Phi_B:C^*(\Gamma)\ra\B(\H)$. Now the assertion follows from corollary~\ref{maxproducp} on the maximal tensor products of ucp maps.

The converse here is simpler than in part~\ref{qmma}: the GNS representation of the given state $\rho$ has all the desired properties. Alice's observables $A^x_a$ are implemented as projection operators $e^x_a\otimes\mathbbm{1}$, and likewise for Bob's $B^y_b=\mathbbm{1}\otimes e^y_b$, so that the measurements are actually projective.

About closedness and convexity, recall that the state space $\mathscr{S}(C^*(\Gamma)\otimes_{\max}C^*(\Gamma))$ is convex and compact in the weak $*$-topology. The projection from $\mathscr{S}(C^*(\Gamma)\otimes_{\max}C^*(\Gamma))$ down to the joint probability space $\R^{k^nm^2}$ is given by evaluation on the elements $e^x_a\otimes e^y_b$, which makes it linear and continuous. Hence its image, which is $\mathcal{Q}_{c}(\Gamma)$, is also convex and compact, and therefore closed.
\end{enumerate}
\end{proof}

Note that in both cases, the proof exhibits the existence of a universal quantum system in the following sense: there is a Hilbert space together with \emph{fixed projective measurements} which can reproduce all quantum correlations as the state varies over all unit vectors in the Hilbert space.

In particular, we have obtained the result that the set of quantum correlations does not depend on whether the parties are allowed to use any POVMs or whether they are restricted to projective measurements. While this is clear with the tensor product assumption by adjoining one ancilla for Alice and one ancilla for Bob, it is physically less intuitive with the commutativity assumption.

From our point of view, proposition~\ref{qcminmaxpovm} provides the most natural framework for the study of quantum correlations. This is apparent not only from the connection between Tsirelson's Problem and the QWEP conjecture which we will describe in section~\ref{mainsec}, but also from the following two remarks which sketch how two other important results on quantum correlations naturally fit into our formalism.

\begin{rem}[The semidefinite hierarchy]
\label{hierarchy}
Building on work of Wehner~\cite{We}, Navascu\'es, Pironio and Ac\'in~\cite{NPA} have developed a hierarchy of semidefinite programs characterizing $\mathcal{Q}_c(\Gamma)$ for any $\Gamma$. (See also~\cite{DLTW} for overlapping work.) We will now argue that this result naturally belongs into our framework; one indication for this is already the proof of~\cite[Thm.~8]{NPA} contains a rediscovery of the GNS representation.

Within our formalism, the hierarchy works roughly as follows. We write $L_n$ (``$L$evel $n$'') for the linear span of the products of up to $n$ generators $e^x_a\otimes\mathbbm{1}$ and $\mathbbm{1}\otimes e^y_b$. The observables $e^x_a\otimes e^y_b$ are in $L_2$. This gives an increasing sequence $(L_n)_{n\in\N}$ of subspaces such that $\cup_n L_n$ is dense in $C^*(\Gamma)\otimes_{\max}C^*(\Gamma)$. By proposition~\ref{qcminmaxpovm}\ref{qmmb} and the Hahn-Banach theorem, a given $P(a,b|x,y)$ lies in $\mathcal{Q}_c(\Gamma)$ if and only if there is a positive linear functional $\mu:L_2\to\C$ satisfying
$$
\mu(e^x_a\otimes e^y_b) = P(a,b|x,y).
$$
It is known that positivity of $\mu$ is equivalent to $\mu(\gamma^* \gamma)\geq 0$ for every $\gamma\in L_n$ and every $n\in\N$~\cite{Schm}. From this, we conclude that $P(a,b|x,y)$ lies in $\mathcal{Q}_c(\Gamma)$ iff there exists, for each $n\geq 1$, some linear map $\mu_n:L_{2n}\to\C$ with $\mu_n(\gamma^*\gamma)\geq 0$ for every $\gamma\in L_n$ and $\mu_n(e^x_a\otimes e^y_b)=P(a,b|x,y)$. For if this is the case, then the restrictions $\mu_{n|{L_2}}$ form a uniformly bounded sequence of functionals, which has an accumulation point $\mu$ with all the required properties.

The question of existence of such a $\mu_n$ for a fixed $n$ is a semidefinite programming problem: it is equivalent to the existence of a positive semidefinite inner product
\beqn
\label{sdh}
s_n : L_n \times L_n \to \C ,\quad (\gamma_1,\gamma_2)\mapsto s_n(\gamma_1,\gamma_2) 
\eeqn
which is required to satisfy the linear relations $s_n(\gamma_1,\gamma_2)=s_n(\gamma_1',\gamma_2')$ whenever $\gamma_1^*\gamma_2=\gamma_1'^*\gamma_2'$ and also $s_n(e^x_a\otimes\mathbbm{1},\mathbbm{1}\otimes e^y_b)=P(a,b|x,y)$.

In conclusion, $P(a,b|x,y)\in\mathcal{Q}_c(\Gamma)$ holds if and only if each semidefinite program in this infinite sequence of semidefinite programs has a solution. This result has already found many applications; see e.g.~\cite{Rand,PV,ETC} for a small selection.
\end{rem}

\begin{rem}[Why the CHSH scenario is simple]
\label{CHSH}
The formulation of Bell scenarios in terms of group $C^*$-algebras has many other advantages. One of them is that it can be easily seen why the CHSH scenario is simple in many regards; e.g.~recall that all quantum correlations can be attained via qubit systems~\cite{Mas}. The reason in the present setting is that the corresponding group $\Z_2\ast\Z_2$ is isomorphic to the semidirect product~\cite[11.63]{Rot}
\beqn
\label{sdp}
\Z_2\ast\Z_2\cong \Z\rtimes\Z_2,
\eeqn
which implies that all its irreducible representations are two-dimensional (i.e.~qubits)~\cite{Beh}. On the other hand, all other nontrivial groups $\Gamma$ of the form\eq{freeprod} contain the free group $\F_2$ as a subgroup (see lemma~\ref{ubgroups}). This prevents them from being written in a form like~\ref{sdp} and witnesses the high complexity of their representation theory.
\end{rem}

One can also formulate the main part of theorem~\ref{qcminmaxpovm} dually in terms of Bell operators and their maximal quantum values. The intuition is that when $\overline{\mathcal{Q}_{\otimes}(\Gamma)}\neq\mathcal{Q}_c(\Gamma)$, then the Hahn-Banach theorem guarantees the existence of a Bell inequality certifying this, i.e.~a Bell inequality whose maximal quantum value with the tensor product assumption is \emph{strictly smaller} than its maximal quantum value with the commutativity assumption. Note that we use the term ``Bell inequality'' in a rather loose sense not referring to the refutation of local realism, but only as a synonym for a linear combination of outcome probabilities; more formally, if $C^{a,b}_{x,y}$ is any tensor of $\R$-valued coefficients, then we consider the linear combination
\beq
\sum_{a,b,x,y}C^{a,b}_{x,y}P(a,b|x,y)
\eeq
as defining a linear functional from the conditional probability distributions $P(\cdot,\cdot|\cdot,\cdot)$ to $\R$.

\begin{prop}
\label{bominmax}
The maximal absolute value of such a functional on the set\ldots
\begin{enumerate}
\item \ldots $\overline{\mathcal{Q}_{\otimes}(\Gamma)}$ is given by
\beqn
\label{borepmin}
\left|\left|\sum_{a,b,x,y}C^{a,b}_{x,y}e^x_a\otimes e^y_b\right|\right|_{\min}
\eeqn
where $||\cdot||_{\min}$ is the norm on $C^*(\Gamma)\otimes_{\min}C^*(\Gamma)$.
\item \ldots $\mathcal{Q}_c(\Gamma)$ is given by
\beqn
\label{borepmax}
\left|\left|\sum_{a,b,x,y}C^{a,b}_{x,y}e^x_a\otimes e^y_b\right|\right|_{\max}
\eeqn
where $||\cdot||_{\max}$ is the norm on $C^*(\Gamma)\otimes_{\max}C^*(\Gamma)$.
\end{enumerate}
\end{prop}

\begin{proof}
This is clear from the previous theorem by recalling that the norm of a self-adjoint element in a $C^*$-algebra equals the maximum of the absolute value of this element under evaluation on all $C^*$-algebraic states.
\end{proof}

\section{From Kirchberg's conjecture to Tsirelson's problem}
\label{mainsec}

We can now immediately apply the material collected in appendix~\ref{appqwep} to the results established in the previous section in order to obtain our first main theorem.

\begin{thm}
\label{QWEPtoTP}
Let $\Gamma$ be any bipartite Bell scenario. If the QWEP conjecture is true, then $\overline{\mathcal{Q}_{\otimes}(\Gamma)}=\mathcal{Q}_c(\Gamma)$.
\end{thm}

\begin{proof}
By theorem~\ref{FGamma}, the QWEP conjecture implies that the minimal tensor product and the maximal tensor product appearing in theorem~\ref{qcminmaxpovm} coincide.
\end{proof}

The ramifications of this result are twofold. Firstly and obviously, it is a conditional solution of Tsirelson's problem, given that one has some cofindence in the validity of the QWEP conjecture. If one is not willing to attribute much likelihood to the validity of QWEP, this may seem like a rather useless statement. However from this point of view, theorem~\ref{QWEPtoTP}, together with its upcoming siblings~\ref{QWEPtoSTTP} and~\ref{QWEPtoSDTP}, may provide a possible means of disproving the QWEP conjecture by finding a counterexample to Tsirelson's problem. Moreover, we believe that our approach in terms of maximal group $C^*$-algebras is an appropriate playground for trying to do so: the observables of each party are defined in terms of unitary representations of groups, and this opens the door for the application of results and methods from the theory of infinite-dimensional unitary representations of discrete groups. A particularly striking case is the one where both Alice and Bob can choose between one $2$-outcome and one $3$-outcome measurement. The corresponding group is the \emph{modular group} $\Gamma=\Z_2\ast\Z_3\cong PSL_2(\Z)$~\cite[11.64(ii)]{Rot}, a central and widely studied object in the theory of modular forms~\cite{123MF}.

The rest of the main part of this paper is devoted to deriving some variations on the main theme given by theorem~\ref{QWEPtoTP}.

\begin{rem}[The multipartite case]
\label{multipart}
For more than two parties involved in a Bell scenario, there are many possible and nontrivial combinations of the tensor product assumption with the commutativity assumption. For example when Alice, Bob and Charlie share a tripartite state, one possible assumption might be that Alice operates on a Hilbert space $\H_A$, Bob operates on a Hilbert space $\H_B$, while Charlie's observables live on the joint Hilbert space $\H_A\otimes\H_B$ and commute with those of Alice and Bob. Alternatively, one could assume that Charlie's observables only live on Bob's $\H_B$, where they commute with Bob's operators, while Alice operates on the separate Hilbert space $\H_A$, so that again $\H_A\otimes\H_B$ describes the total system. In principle, these two options might give rise to diferent sets of tripartite quantum correlations---even so if the QWEP conjecture is true!

In terms of $C^*$-algebra tensor products, these two possibilities correspond to looking at the tensor products
\beqn
\label{triparta}
\left(C^*(\Gamma)\otimes_{\mathrm{min}} C^*(\Gamma)\right)\otimes_{\mathrm{max}} C^*(\Gamma)
\eeqn
and 
\beqn
\label{tripartb}
C^*(\Gamma)\otimes_{\mathrm{min}}\left(C^*(\Gamma)\otimes_{\mathrm{max}} C^*(\Gamma)\right),
\eeqn
respectively, such that a statement analogous to proposition~\ref{qcminmaxpovm} holds. These two tensor products might in principle be different; assuming that the QWEP conjecture is true and using the commutativity of $\otimes_{\mathrm{min}}$ as well as the obvious analogue of theorem~\ref{FGamma}, asking whether~(\ref{triparta}) equals~(\ref{tripartb}) for $\Gamma\neq\Z_2\ast\F_2$ can be formulated as asking whether
\beq
C^*(\F_2\times\F_2)\otimes_{\mathrm{max}} C^*(\F_2)\stackrel{?}{=}C^*(\F_2\times\F_2)\otimes_{\mathrm{min}} C^*(\F_2).
\eeq
And although this problem seems related to the QWEP conjecture itself, according to~\cite[p.~15]{Oza} no formal connection has been found.
\end{rem}

\paragraph{Spatiotemporal correlations.} We now describe a stronger version of Tsirelson's problem, an affirmative answer to which would also follow from a proof of the QWEP conjecture. The initial motivation for considering something more than just spatial quantum correlations $P(x,y|a,b)$ is that the expressions\eq{borepmin} and\eq{borepmax} only probe the norms $||\cdot||_{\min}$ and $||\cdot||_{\max}$ on a tiny part of the group $C^*$-algebra, since each tensor factor $e^x_a$ is only a generator of the free product\eq{Cfreeprod} (as opposed to a product of more than one generator), and it is unclear whether any potential difference between $||\cdot||_{\min}$ and $||\cdot||_{\max}$ would manifest itself on this small subspace. In principle it may be conceivable that the norms $||\cdot||_{\min}$ and $||\cdot||_{\max}$ differ only on some ``higher'' part of the algebras; in general, one has to consider arbitrary elements of $\C[\Gamma]\otimes_{\mathrm{alg}}\C[\Gamma]$. An arbitary element of $\Gamma$ is a product of generators, so that an arbitrary element of $\C[\Gamma]\otimes_{\mathrm{alg}}\C[\Gamma]$ is a linear combination of tensor products of products of generators. More formally, it is not difficult to see that an arbitrary element of $\C[\Gamma]\otimes_{\mathrm{alg}}\C[\Gamma]$ can be written in the form
\beq
\sum_{\{a_i\},\{b_i\},\{x_i\},\{y_i\}}C^{a_1,\ldots,a_t;b_1,\ldots,b_t}_{x_1,\ldots,x_t;y_1,\ldots,y_t}\: e^{x_1}_{a_1}\cdots e^{x_t}_{a_t}\otimes e^{y_1}_{b_1}\cdots e^{y_t}_{b_t}
\eeq
for some $t\in \N$ specifying the highest relevant degree and some tensor of complex coefficients $C^{a_1,\ldots,a_t;b_1,\ldots,b_t}_{x_1,\ldots,x_t;y_1,\ldots,y_t}$. Using the shorthand $\vec{x}\equiv (x_1,\ldots,x_t)$, and similarly for $\vec{y}$, $\vec{a}$ and $\vec{b}$, we write this expression more succinctly as
\beqn
\label{genelem}
\sum_{\vec{a},\vec{b},\vec{x},\vec{y}}C^{\vec{a};\vec{b}}_{\vec{x};\vec{y}}\: e^{\vec{x}}_{\vec{a}}\otimes e^{\vec{y}}_{\vec{b}},
\eeqn
where it is understood that $e^{\vec{x}}_{\vec{a}}$, and likewise $e^{\vec{y}}_{\vec{b}}$, now is a product of generators.

By~\ref{FGamma}, QWEP is true if and only if $||\cdot||_{\min}$ and $||\cdot||_{\max}$ coincide on all elements of the form\eq{genelem}. Hence if we would be able to find a physical interpretation of the norm of an element like\eq{genelem} (or for a sufficiently big class of such elements), we would then obtain an extended version of Tsirelson's problem fully equivalent to QWEP.

One possible approach in this direction is based on \emph{spatiotemporal correlations}. Let us again suppose that Alice and Bob have access to $k$ measurements each, where each measurement has $m$ possible outcomes. Additionally, we now restrict these measurements to be projective, so that the $A^x_a$ and $B^y_b$ are now projection operators on their respective Hilbert spaces. 

Following~\cite{Pop}, the main idea now is to consider the case where Alice and Bob can apply these measurements not only once, but any number of times in temporal succession. In other words, they do not discard their system after they have measured, but they simply measure it again and again, each time using their free will to choose a new measurement setting. We assume furthermore that the quantum system employed by Alice and Bob does not have any nontrivial dynamics on the relevant timescales (besides the projection onto a measurement eigenstate due to the application of a measurement).

So when Alice chooses a sequence of measurement settings $\vec{x}\equiv(x_1,\ldots,x_t)$, her outcome probabilities will be governed by the \emph{history operators}
\beq
A^{\vec{x}}_{\vec{a}}\equiv A^{x_t}_{a_t}\cdots A^{x_1}_{a_1}.
\eeq
For each sequence of settings $\vec{x}$, this defines a POVM $\left(A^{\vec{x}\:\ast}_{\vec{a}} A^{\vec{x}}_{\vec{a}}\right)_{\vec{a}}$ indexed by the possible outcome sequences $\vec{a}\in[m]^t$. If Bob similarly has projective measurements which he applies not only once, but $t$ times in temporal succession with measurement settings $\vec{y}\equiv(y_1,\ldots,y_t)$, the joint outcome probabilities on an initial state defined by the vector $\psi\in\H_A\otimes\H_B$ are given by
\beqn
\label{sttensor}
P(\vec{a},\vec{b}|\vec{x},\vec{y})=\langle\psi,\left(A^{x_1}_{a_1}\ldots A^{x_t}_{a_t}\ldots A^{x_1}_{a_1}\otimes B^{y_1}_{b_1}\ldots B^{y_t}_{b_t}\ldots B^{y_1}_{b_1}\right)\psi\rangle.
\eeqn
The obvious question now is, which conditional probability distributions $P(\vec{a},\vec{b}|\vec{x},\vec{y})$ can occur in this way for some unit vector $\psi$ and projective measurements $A^x_a$, $B^y_b$? See~\cite{2D} for an answer to this question in the single-party case with $\Gamma=\Z_2\ast\Z_2$.

The set of all $P(\vec{a},\vec{b}|\vec{x},\vec{y})$ which can be written in the form\eq{sttensor} is a set $\mathcal{ST}_{\otimes}^t(\Gamma)\subset\R^{m^{2t}k^{2t}}$. Using the same proof as for lemma~\ref{Qconvex}, it can be shown that $\mathcal{ST}_{\otimes}^t(\Gamma)$ is convex.

The formula\eq{sttensor} has been derived under the tensor product assumption. With the commutativity assumption, there is only a single Hilbert space $\H$ for Alice and Bob, on which their projection operators $A^x_a$ and $B^y_b$ live. Of course, the same formula as\eq{sttensor} holds, except that the tensor product sign has to be omitted. The set of all such conditional probability distribution $P(\vec{a},\vec{b}|\vec{x},\vec{y})$ is another subset of $\R^{m^{2t}k^{2t}}$ which we will denote by $\mathcal{ST}_c^t(\Gamma)$. Note that we keep $t$ fixed throughout.

\begin{prop}
\label{stqcminmax}
In the bipartite Bell scenario $\Gamma$, a given spatiotemporal probability distribution $P(\vec{a},\vec{b}|\vec{x},\vec{y})$ is \ldots
\begin{enumerate}
\item \ldots in $\overline{\mathcal{ST}_\otimes^t(\Gamma)}$ if and only if there is a $C^*$-algebraic state $\rho\in\mathscr{S}(C^*(\Gamma)\otimes_{\min} C^*(\Gamma))$ such that
\beqn
\label{stqrepmin}
P(\vec{a},\vec{b}|\vec{x},\vec{y})=\rho(e^{x_1}_{a_1}\ldots e^{x_t}_{a_t}\ldots e^{x_1}_{a_1}\otimes e^{y_1}_{b_1}\ldots e^{y_t}_{b_t}\ldots e^{y_1}_{b_1}).
\eeqn
\item \ldots in $\mathcal{ST}_c^t(\Gamma)$ if and only if there is a $C^*$-algebraic state $\rho\in\mathscr{S}(C^*(\Gamma)\otimes_{\max} C^*(\Gamma))$ such that
\beqn
\label{stqrepmax}
P(\vec{a},\vec{b}|\vec{x},\vec{y})=\rho(e^{x_1}_{a_1}\ldots e^{x_t}_{a_t}\ldots e^{x_1}_{a_1}\otimes e^{y_1}_{b_1}\ldots e^{y_t}_{b_t}\ldots e^{y_1}_{b_1}).
\eeqn
Furthermore, $\mathcal{ST}_c^t(\Gamma)$ is a closed and convex set.
\end{enumerate}
\end{prop}

\begin{proof}
Except for the fact that working with projective measurements replaces the ucp maps by $*$-homomorphisms, this is totally analogous to the proof of proposition~\ref{qcminmaxpovm}.
\end{proof}

Just like in section~\ref{Tsipovm}, we can also find a dual version of this as an analogue to proposition~\ref{bominmax}. By now it should be obvious how to do this, so we will not embark on the details here. What is important to note is that the algebra elements which the state gets evaluated on in equations\eq{stqrepmin} and\eq{stqrepmax} is a tensor product of \emph{palindromes} in the generators. Hence the linear subspace of $\C[G]\otimes_{\mathrm{alg}}\C[G]$ on which one can probe the norms $||\cdot||_{\min}$ and $||\cdot||_{\max}$ by looking at spatiotemporal correlations is exactly the linear subspace spanned by the tensor products of palindromes in the generators $e^x_a$.

Again since the commutativity assumption subsumes the tensor product assumption, we certainly have $\overline{\mathcal{ST}_\otimes^t(\Gamma)}\subseteq\mathcal{ST}_c^t(\Gamma)$, so that the obvious question arises whether these two sets coincide for every choice of the number of measurements $t\in\N$,

\beqn
\label{STTP}
\boxed{\mathbf{STTP}(\Gamma):\quad\overline{\mathcal{ST}_\otimes^t(\Gamma)}\stackrel{?}{=}\mathcal{ST}_c^t(\Gamma)\quad\forall t}
\eeqn

\begin{cor}
\label{QWEPtoSTTP}
Let $\Gamma$ be any bipartite Bell scenario. If the QWEP conjecture is true, then $\overline{\mathcal{ST}_\otimes^t(\Gamma)}=\mathcal{ST}_c^t(\Gamma)$ holds for all $t$.
\end{cor}

\begin{proof}
This is clear by theorem~\ref{FGamma} and proposition~\ref{stqcminmax}.
\end{proof}

Unfortunately, it is presently unclear whether the subspace spanned by the tensor products of palindromes is big enough to allow for a converse implication. The propositions~\ref{statecomp} and~\ref{unbounded} on the comparison between $\otimes_{\min}$ and $\otimes_{\max}$ may be relevant for a potential proof of this. Another interesting observation is that the palindromes in the generators correspond exactly to the diagonal entries in the hierarchy of semidefinite programs characterizing quantum correlations~\cite{NPA}.

Finally, we would like to demonstrate that spatiotemporal quantum correlations can be interesting in that they may reveal nonlocality in a system where spatial correlations alone do not. While an example of this has also been given in~\cite{Pop}, ours is significantly simpler.

\begin{ex}
\label{Wstate}
Suppose that Alice and Bob share a three-qubit W state with two qubits owned by Alice and one qubit by Bob. Hence, we consider the Hilbert spaces
\beq
\mathcal{H}_A\equiv\C^2\otimes\C^2,\qquad \mathcal{H}_B\equiv\C^2
\eeq
together with, in the obvious notation, the unit vector
\beq
W\equiv\frac{1}{\sqrt{3}}e_{00}\otimes e_1+\frac{1}{\sqrt{3}}\left(e_{01}+e_{10}\right)\otimes e_0.
\eeq
Further, we suppose that Alice can choose between four projective measurements. We write these in terms of $\pm$-valued observables instead of projection operators,
\beq
A_1\equiv \sigma_z\otimes\mathbbm{1},\qquad A_2\equiv \sigma_x\otimes\mathbbm{1},\qquad A_3\equiv\mathbbm{1}\otimes\sigma_z,\qquad A_4\equiv\mathbbm{1}\otimes\sigma_x,
\eeq
while Bob has access to the two $\pm 1$-valued measurements
\beq
B_1\equiv\frac{\sigma_z-\sigma_x}{\sqrt{2}},\qquad B_2\equiv\frac{\sigma_z+\sigma_x}{\sqrt{2}}.
\eeq
We now claim that when both parties are allowed to conduct only a single measurement, the ensuing correlations $P(a_1,b|x_1,y)$ admit a local hidden variable model. Indeed, in this case Alice's measurements $A_3$ and $A_4$ are redundant, since they give the same statistics as $A_1$ and $A_2$, respectively. Hence in this case, Alice and Bob effectively operate in a CHSH scenario, and therefore it is enough to check the $8$ variants of the CHSH inequality on the reduced state on Alice's first qubit and Bob's qubit. A direct calculation shows that no CHSH violation arises.

However, when Alice is allowed to conduct two sequential measurements, while Bob stays with measuring only once, the corresponding conditional probability distribution
\beq
P(a_1,a_2,b|x_1,x_2,y)
\eeq
does not admit a local hidden variable model. In order to see this, let us consider only the case where Alice first measures $A_1$, while subsequently measuring $A_3$ or $A_4$; the idea here is that Alice's first measurement is supposed to steer the remaining two qubits into the maximally entangled vector $\frac{1}{\sqrt{2}}(e_0\otimes e_1+e_1\otimes e_0)$, which is successful with a probability of $2/3$. This fails with a probability of $1/3$, after which the resulting state is a product state, so that Alice's second measurement is useless and does not even have to be conducted. Hence, there are only \emph{three} sensibly different outcomes for Alice: ``$+$'' (respectively ``$-$'') for the case that the initial measurement of $A_1$ succeeds and the second measurement gives $+1$ (respectively $-1$); and ``$\emptyset$'' for the case that the initial measurement of $A_1$ fails. A calculation of all joint outcome probabilities results in the following table, writh $\gamma_{\pm}\equiv\frac{1}{6}\left(1\pm\frac{1}{\sqrt{2}}\right)$,
\beqn
\label{Wstateprob}
		\begin{tabular}{cc|cc|cc}
		\\
		{} && \multicolumn{2}{c|}{$B_1$} & \multicolumn{2}{|c}{$B_2$}\\
		{} && \uncircledentry{$+$} & \uncircledentry{$-$} & \uncircledentry{$+$} & \uncircledentry{$-$}\\

		\hline \multirow{3}{*}{$A_1,A_3$} & \uncircledentry{$+$} & \uncircledentry{$\gamma_-$} & \uncircledentry{$\gamma_+$} & \circledentry{$\gamma_-$} & \uncircledentry{$\gamma_+$} \\
		 & \uncircledentry{$-$} & \uncircledentry{$\gamma_+$} & \circledentry{$\gamma_-$} & \uncircledentry{$\gamma_+$} & \uncircledentry{$\gamma_-$} \\
		 & \uncircledentry{$\emptyset$} & \uncircledentry{$\gamma_+$} & \circledentry{$\gamma_-$} & \uncircledentry{$\gamma_+$} & \uncircledentry{$\gamma_-$} \\

		\hline \multirow{3}{*}{$A_1,A_4$} & \uncircledentry{$+$} & \circledentry{$\gamma_-$} & \uncircledentry{$\gamma_+$} & \squaredentry{$\gamma_+$} & \uncircledentry{$\gamma_-$} \\
		 & \uncircledentry{$-$} & \uncircledentry{$\gamma_+$} & \uncircledentry{$\gamma_-$} & \uncircledentry{$\gamma_-$} & \uncircledentry{$\gamma_+$} \\
		 & \uncircledentry{$\emptyset$} & \uncircledentry{$\gamma_+$} & \uncircledentry{$\gamma_-$} & \uncircledentry{$\gamma_+$} & \uncircledentry{$\gamma_-$} \\

		{}
\end{tabular}
\eeqn
Now it is simple to see that the inequality
\beq
P(+,+|A_1,A_4;B_2)\leq P(+,+|A_1,A_4;B_1)+P(+,+|A_1,A_3;B_2)+P(-,-|A_1,A_3;B_1)+P(-,\emptyset|A_1,A_3;B_1),
\eeq
which corresponds to the marked entries in the above table, must hold in any hidden variable model. (It is essentially the quantitative version of Hardy's nonlocality~\cite{Braun}.) However, since $\gamma_+>4\gamma_-$, this inequality does not hold for the table\eq{Wstateprob}, and hence\eq{Wstateprob} is nonlocal.

Summarizing this example, we have shown that this system of an initial state together with projective measurements reveals its nonlocality only after applying some measurements in temporal succession.
\end{ex}

Clearly this example is contrived in the sense that the ``temporal'' aspect of these correlations is not genuinely temporal, since Alice's two successive measurements (first $A_1$, then $A_3$ or $A_4$) operate on different qubits. It would be interesting to find more authentic examples; see~\cite{Pop}. Also, so far we have not been looking for examples where $t$ sequential measurements on either side reveal nonlocality, although $t-1$ sequential measurements on either side do not, for any larger value of $t$.

It might also be worthwhile to consider the case where the measuremets of Alice and Bob are allowed to use generalized measurements instead of projective measurements only. Because then one has to consider Kraus operators, which are usually not self-adjoint, one will not be able to derive a correspondence between a system of measurements and a ucp map from a maximal group $C^*$-algebra. In other words, extending proposition~\ref{POVMchar} to the case of generalized measurements described in terms of Kraus operators may not be possible within our approach.

\section{From Tsirelson's problem to Kirchberg's conjecture}
\label{reverseimp}

In this section, we will define another version of Tsirelson's problem and show that a positive answer to this version would imply a positive answer to the QWEP conjecture. Similar to the duality between theorem~\ref{qcminmaxpovm} and theorem~\ref{bominmax}, we again offer a ``primal'' and a ``dual'' version, starting with the former. Thie relevant concept here is the notion of \emph{steering}~\cite{Schr35}, i.e.~the phenomenon that a measurement on one part of a bipartite quantum system collapses the state of the other part; the state to which the other part collapses to is called the~\emph{steered state}. For a more contemporary account of steering in the context of entanglement and nonlocality, see~\cite{WJD}.

\paragraph{Steering data.} Before explaining the relevant concepts for the bipartite case, we start with the simpler setting of a single party, say Alice. Let us presume that Alice shares a bipartite state $\rho$ with a neutral party, which we call The Verifier. The Verifier operates with a quantum system of fixed dimension $d$, while Alice's system size is arbitrary and may be infinite-dimensional. As before, Alice has access to $k$ POVMs with $m$ outcomes each described in terms of positive operators $A^x_a$. Getting an outcome $a$ upon measuring $x$ will collapse the state of The Verifier's system to the state
\beqn
\label{steering}
\alpha^x_a(\cdot)\equiv\rho(\cdot\otimes A^x_b).
\eeqn
where we think of the latter $\alpha$ as standing for Alice, and we have intentionally omitted the normalization, so that the probability of getting the outcome $x$ for measurement choice $a$ can be recovered as $P(a|x)=\alpha^x_a(\mathbbm{1})$. Such a $\alpha^x_a$ is a linear positive functional on The Verifier's observable algebra with the property that $\alpha^x_a(\mathbbm{1})\leq 1$. To such a functional we will refer as a \emph{subnormalized state}, although it is, strictly speaking, not a state.

We consider ensembles of subnormalized states $(\alpha^x_a)^{x\in [k]}_{a\in[m]}$, indexed by the possible measurement settings $x\in[k]$ and the possible measurement outcomes $a\in[m]$. Every instancve of steering\eq{steering} defines such an ensemble. How should one think of such an ensemble? We have already noted that $\alpha^x_a(\mathbbm{1})=P(a|x)$. In particular, when The Verifier's system is trivial in the sense that it has dimension $d=1$, then the subnormalized steered state equals Alice's outcome probability, $\alpha^x_a=P(a|x)$. For this reason, we think of steering data as a quantum analogue or noncommutative analogue of the conditional probability distribution $P(x|a)$; mathematically speaking, it corresponds to an element of the $d$th matrix level of an operator system~\cite{Paul}. Although operator systems are secretly lurking behind our formalism, we will restrain from using this terminology and formulate our results in pedestrian terms.

The conditional probability distribution $P(a|x)$ has the properties
\begin{enumerate}
\item positivity: $P(a|x)\geq 0$ for all $x,a$.
\item normalization: $\sum_a P(a|x)=1$ for all $x$.
\item no-signaling: $\sum_a P(a|x)$ is independent of $x$. (Implied by normalization.)
\end{enumerate}
As can be readily verified, the analogous properties of an ensemble $(\alpha^x_a)^{x\in [k]}_{a\in[m]}$ arising from\eq{steering} are
\begin{enumerate}
\item positivity: $\alpha^x_a\geq 0$ for all $x,a$.
\item normalization: $\sum_a\alpha^x_a(\mathbbm{1})=1$ for all $x$.
\item no-signaling: $\sum_a\alpha^x_a$ is independent of $x$. (Not implied by normalization if $d\geq 2$.)
\end{enumerate}
If an ensemble $(\alpha^x_a)^{x\in [k]}_{a\in[m]}$ satisfies these conditions, we call it \emph{steering data}. The following result shows that any kind of steering data arises from actual steering in the form\eq{steering}; similar considerations have already made by Schr\"odinger~\cite{Schr36}, who considered the case where the joint state is fixed and assumed to be pure.

\begin{prop}[{\cite[3.3]{HJW}}]
\label{msd}
For any steering data $(\alpha^x_a)^{x\in [k]}_{a\in[m]}$ on $M_d(\C)$, there is a bipartite state $\rho$ on $M_d(\C)\otimes M_d(\C)$, shared between The Verifier and Alice, together with $k$ POVMs of Alice given in terms of positive operators $A^x_a\in M_d(\C)$, such that this steering data is of the form\eq{steering}.
\end{prop}

\begin{proof}
We closely follow the reasoning of~\cite{HJW}.

For the sake of illustration, we start by considering the case that The Verifier's reduced state $\rho_V\equiv\sum_a\alpha^x_a\in\mathscr{S}(M_d(\C))$ is the totally mixed state $\rho_V=\frac{1}{d}\tr$. In this case, we take the total state $\rho$ to be the maximally entangled state
\beq
\rho:M_d(\C)\otimes M_d(\C)\lra\C,\qquad X\otimes Y\mapsto\frac{1}{d}\tr(X^tY),
\eeq
where $X^t$ denotes the transpose of $X$. We choose Alice's observables $A^x_a$ as the unique matrices in $M_d(\C)$ satisfying
\beqn
\label{defnsteerobs}
\frac{1}{d}\tr(A^x_a\,Y^t)=\alpha^x_a(Y)\quad\forall \,Y\in M_d(\C).
\eeqn
Since $\tr(A^x_aY^t)=\tr((A^x_a)^tY)$, one can interpret this as saying that $A^x_a$ is the (transpose of the) density matrix associated to the subnormalized state $\alpha^x_a$. In particular, the $A^x_a$ are automatically going to be positive.

Now $\sum_aA^x_a=\mathbbm{1}$ holds by\eq{defnsteerobs} and the assumption that $\rho_V$ is totally mixed. By\eq{steering}, we then obtain for the steered state
\beq
\rho(X\otimes A^x_a)=\frac{1}{d}\tr(X^tA^x_a)\stackrel{(\ref{defnsteerobs})}{=}\alpha^x_a(X)\quad\forall X\in M_d(\C),
\eeq
as desired.

In the general case when $\rho_V\equiv\sum_a\alpha^x_a$ is arbitrary, we start by defining the positive matrix $A\in M_d(\C)$ to be the (transpose of the) density matrix associated to $\rho_V$,
\beq
\frac{1}{d}\tr(AX^t)=\rho_V(X)\quad\forall X\in M_d(\C).
\eeq
If $P_A$ is the range projection of $A$, then we take Alice's state space to be $\mathcal{H}=P_A\C^d$. For the initial state shared between The Verifier and Alice, we take
\beq
\rho:M_d(\C)\otimes \B(\H)\lra\C,\qquad X\otimes Y\mapsto\frac{1}{d}\tr(X^tA^{\frac{1}{2}}YA^{\frac{1}{2}}),
\eeq
Getting to Alice's observables, we first choose the positive $\widehat{A}^x_a$ to be the unique matrices satisyfing
\beq
\frac{1}{d}\tr(\widehat{A}^x_a\,Y^t)=\alpha^x_a(Y)\quad\forall \,Y\in M_d(\C),
\eeq
which we think of as operators $\widehat{A}^x_a\in\B(\H)$, since their range is a subspace of $\H=P_A\C^d$. The $\widehat{A}^x_a$ cannot yet serve as the components of Alice's POVMs, since the sum $\sum_a\widehat{A}^x_a=A$ is not as required for a POVM. However upon setting
\beq
A^x_a\equiv A^{-\frac{1}{2}}\widehat{A}^x_a A^{-\frac{1}{2}} \in\B(\H),
\eeq
which makes sense since $A$ is invertible on $\H$, we have $\sum_a A^x_a=A^{-\frac{1}{2}}\sum_a\widehat{A}^x_a A^{-\frac{1}{2}}=\mathbbm{1}_{\H}$, so that the $A^x_a$ form a POVM. Finally,
\beq
\rho(X\otimes A^x_a)=\frac{1}{d}\tr(X^tA^{\frac{1}{2}}A^x_aA^{\frac{1}{2}})=\frac{1}{d}\tr(X^t\widehat{A}^x_a)=\alpha^x_a(X),
\eeq
again as desired.
\end{proof}

\paragraph{Bipartite steering data.} Now Bob enters the picture. We now have to consider tripartite states $\rho$ shared between Alice, Bob, and The Verifier. Again we will restrict The Verifier to a $d$-dimemsional qudit, while Alice and Bob may operate with any finite-dimensional or infinite-dimensional quantum systems. And again, Alice has access to $k$ POVMs, with $m$ outcomes each, defined in terms the positive operators $A^x_a$, and similarly Bob has access to $k$ POVMs defined in terms of the positive operators $B^y_b$. And as in the case of steering data, we consider what happens to The Verifier's system after a measurement conducted by Alice or by Bob. Similar to\eq{steering}, a measurement by Alice or by Bob steers The Verifier's system to one of the subnormalized states
\beqn
\label{bisteerdatat}
\alpha^x_a(\,\cdot\,)\equiv\rho(\,\cdot\otimes A^x_a\otimes\mathbbm{1}),\qquad \beta^y_b(\,\cdot\,)\equiv\rho(\,\cdot\otimes\mathbbm{1}\otimes B^y_b),
\eeqn
Now $(\alpha^x_a)^{x\in [k]}_{a\in[m]}$ and $(\beta^y_b)^{y\in [k]}_{b\in[m]}$ are \emph{each} a set of steering data, satisfying the no-signaling condition 
\beqn
\label{nosigsd}
\sum_a\alpha^x_a=\sum_b\beta^y_b\quad\forall\, x,y\in [k].
\eeqn
We will refer to such a collection $(\alpha^x_a,\beta^y_b)$ of subnormalized states as \emph{bipartite steering data} on $M_d(\C)$. We will write $\mathcal{SD}_\otimes^d(\Gamma)$ for the set of all bipartite steering data on $M_d(\C)$ in the Bell scenario $\Gamma$, again defined by the number of possible measurement settings and possible measurement outcomes to be $k$ and $m$, respectively. For $d=1$, bipartite steering data is simply given by the marginal probabilities $P(a|x)$ and $P(b|y)$; Alice and Bob do not measure jointly. Just like in the single-party case the steering data $\rho^x_a$ was a quantum generalization of the conditional probability distribution $P(a|x)$, the bipartite steering data of the two-party case generalizes the marginal probability distributions $P(a|x)$ and $P(b|y)$.

Note that in the definition\eq{bisteerdatat}, it has been implicitly assumed that Alice's observables act on a Hilbert space different from those of Bob's observables, and the total system is given in terms of tensor products of Hilbert spaces. Again we consider as an alternative hypothesis the assumption that Alice's observables simply commute with Bob's observables\eq{commu}, all of them living on the same Hilbert space $\mathcal{H}$. Then upon a mesurement by Alice or by Bob, the unnormalized steered states are given by
\beqn
\label{bisteerdatac}
\alpha^x_a(\,\cdot\,)\equiv\rho(\,\cdot\otimes A^x_a),\qquad \beta^y_b(\,\cdot\,)\equiv\rho(\,\cdot\otimes B^y_b),
\eeqn
where the first factor refers to The Verifier's system. Again, the collection $(\alpha^x_a,\beta^y_b)$ satisfies the no-signaling requirement~(\ref{nosigsd}), and therefore qualifies as bipartite steering data. We denote the set of bipartite steering data which can be written in the form\eq{bisteerdatac} by $\mathcal{SD}_c^d(\Gamma)$.

\begin{ex}
For every nontrivial Bell scenario $\Gamma$, there exists bipartite steering data $\left(\alpha^x_a,\beta^y_b\right)$ which is not in $\mathcal{SD}_c^d(\Gamma)$. For example, assume that all subnormalized states $\alpha^x_1,\ldots,\alpha^x_m$ for some fixed $x\in[k]$, are pure in the sense that they cannot be decomposed into a nontrivial positive linear combination of other subnormalized states. Then we obtain that for every $a\in[m]$ and $y\in[k]$ the decomposition
\beqn
\label{steerex}
\alpha^x_a(\,\cdot\,)=\sum_b\rho(\,\cdot\otimes A^x_a B^y_b)
\eeqn
has the property that all its contributing subnormalized states $\rho(\,\cdot\,\otimes A^x_a B^y_b)$ are necessarily scalar multiples of $\alpha^x_a$. The missing scalar factor can be determined by plugging in $\mathbbm{1}$, which evaluates\eq{steerex} to $P(a|x)=\sum_bP(a,b|x,y)$. In terms of the outcome probabilities $P(a,b|x,y)$, we hence obtain from the purity assumption
\beq
\rho(\,\cdot\otimes A^x_a B^y_b)=\alpha^x_a(\,\cdot\,)\frac{P(a,b|x,y)}{P(a|x)}\quad\forall\, y\in[k],a,b\in[m].
\eeq
This in turn implies that
\beq
\beta^y_b(\,\cdot\,)=\sum_a\rho(\,\cdot\otimes A^x_a B^y_b)=\sum_a\alpha^x_a(\,\cdot\,)\frac{P(a,b|x,y)}{P(a|x)}\quad\forall y\in[k],b\in[m].
\eeq
So in particular, if the given bipartite steering data is supposed to lie in $\mathcal{SD}_c^d(\Gamma)$, then purity of all $\alpha^x_1,\ldots,\alpha^x_m$ for some $x\in[k]$ implies that every $\beta^y_b$ is a stochastic mixture of the $\alpha^x_a$. This may be considered an incarnation of entanglement monogamy: in any non-degenerate case, purity of the $\alpha^x_1,\ldots,\alpha^x_m$ necessitates strong entanglement between Alice and The Verifier, which in turn excludes the possibility of entanglement between Bob and The Verifier. If Alice can steer The Verifier's system to an ensemble of pure states, then Bob will not be able to steer The Verifier's system to a different ensemble of pure states.
\end{ex}

Hence, in contrast to proposition~\ref{msd}, deciding when given bipartite steering data $(\alpha^x_a,\beta^y_b)$ can come from a tripartite quantum state is nontrivial. 

We now formulate the analogue to proposition~\ref{qcminmaxpovm} for bipartite steering data. Again we note that it is unclear whether the set $\mathcal{SD}_\otimes^d(\Gamma)$ is closed, and that any potential distinction between $\mathcal{SD}_\otimes^d(\Gamma)$ and its topological closure $\overline{\mathcal{SD}_\otimes^d(\Gamma)}$ would not be physically relevant.

\begin{prop}
\label{bsdt}
Bipartite steering data $(\alpha^x_a,\beta^y_b)$ lies in \ldots
\begin{enumerate}
\item\label{bsda} \ldots $\overline{\mathcal{SD}_\otimes^d(\Gamma)}$ if and only if there is a $C^*$-algebraic state $\rho\in\mathscr{S}(M_d(C^*(\Gamma)\otimes_{\min}C^*(\Gamma)))$ such that
\beqn
\label{srepmin}
\alpha^x_a(\,\cdot\,)=\rho(\,\cdot\otimes e^x_a\otimes\mathbbm{1}),\qquad \beta^y_b(\,\cdot\,)=\rho(\,\cdot\otimes\mathbbm{1}\otimes e^y_b).
\eeqn
The set $\mathcal{SD}_\otimes^d(\Gamma)$, and therefore also $\overline{\mathcal{SD}_\otimes^d(\Gamma)}$, is convex.
\item\label{bsdb} \ldots $\mathcal{SD}_c^d(\Gamma)$ if and only if there is a $C^*$-algebraic state $\rho\in\mathscr{S}(M_d(C^*(\Gamma)\otimes_{\max}C^*(\Gamma)))$ such that
\beqn
\label{srepmax}
\alpha^x_a(\,\cdot\,)=\rho(\,\cdot\otimes e^x_a\otimes\mathbbm{1}),\qquad \beta^y_b(\,\cdot\,)=\rho(\,\cdot\otimes\mathbbm{1}\otimes e^y_b).
\eeqn
The set $\mathcal{SD}_c^d(\Gamma)$ is closed and convex.
\end{enumerate}
\end{prop}

\begin{proof}
\begin{enumerate}
\item[\ref{bsda}] This is mostly analogous to the proof of proposition~\ref{qcminmaxpovm}\ref{qmma}, but nevertheless we spell out most of the details. Suppose that the given steering data is quantum with the tensor product assumption, so that there exist Hilbert spaces $\H_A$ and $\H_B$, a unit vector $\psi$ (again, density matrices can be purified) defined over the total tripartite Hilbert space
\beq
\psi\in\C^d\otimes\H_A\otimes\H_B
\eeq
and observables $A^x_a$ and $B^y_b$ such that
\beq
\alpha^x_a(\cdot)=\langle\psi,(\cdot\otimes A^x_a\otimes\mathbbm{1})\psi\rangle,\qquad \beta^y_b(\cdot)=\langle\psi,(\cdot\otimes \mathbbm{1}\otimes B^y_b)\psi\rangle
\eeq
holds. Then again as in the proof of proposition~\ref{qcminmaxpovm}, Alice's observables define a ucp map $\Phi_A:C^*(\Gamma)\ra\B(\H_A)$, and likewise for Bob in terms of $\Phi_B:C^*(\Gamma)\ra\B(\H_B)$. But then also
\beq
\mathrm{id}\otimes_{\min}\Phi_A\otimes_{\min}\Phi_B:M_d\otimes_{\min} C^*(\Gamma)\otimes_{\min} C^*(\Gamma)\lra \B(\C^d\otimes\H_A\otimes\H_B)
\eeq
is a ucp map, so that
\beq
\rho(\,\cdot\otimes\cdot\otimes\cdot)\equiv \langle\psi,(\,\cdot\otimes\Phi_A(\,\cdot\,)\otimes\Phi_B(\,\cdot\,))\,\psi\rangle
\eeq
defines a state in $\mathscr{S}(M_d(C^*(\Gamma)\otimes_{\min}C^*(\Gamma)))$. By construction, this state satisfies~\ref{srepmin}.

For the converse implication, suppose that the bipartite steering data $(\alpha^x_a,\beta^y_b)$ is given in terms of a state $\rho\in\mathscr{S}(M_d(C^*(\Gamma)\otimes_{\min}C^*(\Gamma)))$ satisfying\eq{srepmin}. Choose some faithful representation $C^*(\Gamma)\subseteq\B(\H)$. This induces a faithful representation
\beq
M_d(C^*(\Gamma)\otimes_{\min}C^*(\Gamma))\subseteq\B(\C^d\otimes\H\otimes\H).
\eeq
Concerning observables, we again make the obvious choice $A^x_a\equiv e^x_a$ and $B^y_b\equiv e^y_b$. We will show the existence of a mixed state on $\C^n\otimes\H\otimes\H$ with associated bipartite steering data $(\widetilde{\rho}^x_a,\widetilde{\rho}^y_b)$ which approximates the given\eq{srepmin} up to $\eps$,
\beq
||\widetilde{\rho}^x_a-\alpha^x_a||<\eps\quad\forall x,a,\qquad ||\widetilde{\rho}^y_b-\beta^y_b||<\eps\quad\forall y,b.
\eeq
This set of conditions is guaranteed to be satisfied if, for an appropriate set of hermitians $V_1,\ldots,V_n\in M_d(\C)$ 
\beq
|\widetilde{\rho}^x_a(V_i)-\alpha^x_a(V_i)|<\eps\quad\forall\, i,x,a,\qquad |\widetilde{\rho}^y_b(V_i)-\beta^y_b(V_i)|<\eps\quad\forall\, i,y,b,.
\eeq
Then the assertion follows again from the density of vector states (proposition~\ref{vectorstates}) by approximating the action of $\rho$ on the operators $V_i\otimes A^x_a\otimes\mathbbm{1}$ and $V_i\otimes\mathbbm{1}\otimes B^y_b$.

We now turn to convexity of $\mathcal{SD}^d_{\otimes}(\Gamma)$; the proof here is essentially identical to the one of~(\ref{Qconvex}), except that we now have to take care of an additional tensor factor. Given two sets of bipartite steering data $(\alpha^x_{a,1},\beta^y_{b,1})$ and $(\alpha^x_{a,2},\beta^y_{b,2})$ coming from Hilbert spaces $\H_{A,1}$, $\H_{B,1}$ and $\H_{A,2}$, $\H_{B,2}$, respectively, with observables $A^x_{a,i}$ and $B^y_{b,i}$ and unit vectors $\psi_i\in\C^n\otimes\H_{A,i}\otimes\H_{B,i}$, we consider the direct sum Hilbert spaces $\H_{A,1}\oplus\H_{A,2}$ and $\H_{B,1}\oplus\H_{B,2}$ on which the direct sum observables 
\beq
A^x_a\equiv A^x_{a,1}\oplus A^x_{a,2},\qquad B^y_b\equiv B^y_{b,1}\oplus B^y_{b,2}
\eeq
act. The associated tripartite Hilbert space can be decomposed as
\begin{align}
\begin{split}
\label{Hdecompsd}
\C^d\otimes(\H_{A,1}\oplus\H_{A,2})\otimes(\H_{B,1}\oplus\H_{B,2})\:\:=&\phantom{\oplus}\:\:(\C^d\otimes\H_{A,1}\otimes\H_{B,1})\\
&\oplus(\C^d\otimes\H_{A,1}\otimes\H_{B,2})\\&\oplus(\C^d\otimes\H_{A,2}\otimes\H_{B,1})\\&\oplus(\C^d\otimes\H_{A,2}\otimes\H_{B,2})
\end{split}
\end{align}
so that one can consider on this total Hilbert space the unit vector 
\beq
\psi\equiv\sqrt{\lambda}\,\psi_1\oplus 0\oplus 0\oplus\sqrt{1-\lambda}\,\psi_2.
\eeq
By construction, this gives rise to the convex combination bipartite steering data
\beq
(\lambda\alpha^x_{a,1}+(1-\lambda)\alpha^x_{a,2}\,,\:\lambda\beta^y_{b,1}+(1-\lambda)\beta^y_{b,2}).
\eeq

\item[\ref{bsdb}] Suppose that the given bipartite steering data is quantum with the commutativity assumption, so that there is a Hilbert space $\H$, a unit vector $\psi\in\H$ (assuming purity without loss of generality) and observables $A^x_a$ and $B^y_b$ such that
\beq
\alpha^x_a(\cdot)\equiv\langle\psi,\left(\,\cdot\otimes A^x_a\right)\psi\rangle,\qquad \beta^y_b(\cdot)\equiv\langle\psi,\left(\,\cdot\otimes B^y_b\right)\psi\rangle.
\eeq
Then again, Alice's POVMs define a ucp map $\Phi_A:C^*(\Gamma)\ra\B(\H)$ such that\eq{POVMsucps} holds, and likewise for Bob in terms of $\Phi_B:C^*(\Gamma)\ra\B(\H)$. Now again, the assertion follows from corollary~\ref{maxproducp} on maximal tensor products of ucp maps.

For the converse, we assume that the state $\rho\in\mathscr{S}(M_d(C^*(\Gamma)\otimes_{\max}C^*(\Gamma)))$ satisfying\eq{srepmax} is given. Then the GNS construction associated to $\rho$ defines a representation of $M_d(C^*(\Gamma)\otimes_{\max}C^*(\Gamma))$ on a Hilbert space $\widehat{\H}$ which has a distinguished cyclic vector $\psi\in\widehat{\H}$. By noting that the standard basis matrices $e_{ij}\in M_d(\C)$ generate $M_d(\C)$ and satisfy the relations
\beq
e_{ij}e_{i'j'}=\delta_{ji'}e_{ij'},
\eeq
while their action on $\widehat{\H}$ commutes with the action of $C^*(\Gamma)\otimes_{\max}C^*(\Gamma)$, it follows that $\widehat{\H}$ canonically decomposes as, for some Hilbert space $\H$,
\beq
\widehat{\H}=\underbrace{\H\oplus\ldots\oplus\H}_{d\textrm{ summands}}\cong\C^d\otimes\H
\eeq
 where $M_d(\C)$ only acts on the first tensor factor in the standard representation, while $C^*(\Gamma)\otimes_{\max}C^*(\Gamma)$ acts on the second tensor factor. By choosing the observables $A^x_a\equiv e^x_a\otimes\mathbbm{1}\in\B(\H)$ and $B^y_b\equiv \mathbbm{1}\otimes e^y_b\in\B(\H)$ together with the distinguished cyclic unit vector $\psi\in\widehat{\H}$ as the tripartite initial state $\rho$, the given bipartite steering data is of the form\eq{bisteerdatac}.

Given this, closedness and convexity again follow from compactness and convexity of the state space $\mathscr{S}(M_d(C^*(\Gamma)\otimes_{\max}C^*(\Gamma)))$ equipped with the weak $*$-topology. 
\end{enumerate}
\end{proof}

Of course now the obvious question arises, namely the \emph{steering data variant} of Tsirelson's problem:
\beqn
\label{SDTP}
\boxed{\mathbf{SDTP}(\Gamma):\quad\overline{\mathcal{SD}_\otimes^d(\Gamma)}\stackrel{?}=\mathcal{SD}_c^d(\Gamma)\quad\forall d}
\eeqn
Let us emphasize again that any potential difference between $\overline{\mathcal{SD}_\otimes^d(\Gamma)}$ and $\mathcal{SD}_c^d(\Gamma)$ would manifest itself only when the joint Hilbert space is infinite-dimensional~\cite{SW}, and even if the difference exists it is unclear whether physical implementations witnessing this difference would be realistic.

Now we proceed towards stating the dual version of proposition~\ref{bsdt}. This dual version can conveniently be interpreted in terms of nonlocal games of the following form: Alice and Bob share some tripartite entanglement with The Verifier. At each round of the game, The Verifier randomly selects either Alice or Bob, each with probability $1/2$. Suppose that he has selected Alice. Then he sends to her a number $x\in[k]$, choosing uniformly at random. Upon receiving $x$, Alice conducts the measurement corresponding to the measurement setting $x$ and obtains an outcome $a\in[m]$. She then sends this outcome $a$ back to The Verifier. As the last step of the round of the game, The Verifier conducts, on his $d$-dimensional quantum system, the measurement $V^x_a$, which is any hermitian matrix $V^x_a\in M_d(\C)$. Finally, he records the outcome of his measurement as the score of the round. Similarly, if The Verifier has initially selected Bob, then he sends to him a uniformly chosen random number $y\in[k]$, Bob conducts measurement $y$ and obtains an outcome $b$, which he communicates back to The Verifier, who then measures some hermitian matrix $W^y_b\in M_d(\C)$ and records this outcome as the score of the round.

By this description, a particular such \emph{bipartite steering game} in the Bell scenario $\Gamma$ is defined by the collection of measurements conducted by The Verifier, which we call the \emph{dual bipartite steering data} $(V^x_a,W^y_b)$. Given some tripartite state $\rho$ with bipartite steering data $(\alpha^x_a,\beta^y_b)$, a simple calculation shows that the average score (or \emph{value}) of the game is given by
\beq
\omega(V^x_a,W^y_b,\rho)=\frac{1}{2k}\sum_{x,a}\alpha^x_a(V^x_a)+\frac{1}{2k}\sum_{y,b}\beta^y_b(W^y_b)
\eeq
We can now take the \emph{tensor product value} $\omega_\otimes$ of the game to be given by the supremum of all values of the game achievable by bipartite steering data with the tensor product assumption, $(\alpha^x_a,\beta^y_b)\in\mathcal{SD}^d_\otimes(\Gamma)$,
\beqn
\label{valuetp}
\omega_\otimes(V^x_a,W^y_b)\equiv\frac{1}{2k}\sup_{(\alpha^x_a,\beta^y_b)\in\mathcal{SD}^d_\otimes(\Gamma)}\left(\sum_{x,a}\alpha^x_a(V^x_a)+\sum_{y,b}\beta^y_b(W^y_b)\right)
\eeqn
Similarly, we define the \emph{commutative value} $\omega_c$ of the game to be given by the maximum of all values of the game achievable by bipartite steering data with the commutativity assumption,
\beqn
\label{valuec}
\omega_c(V^x_a,W^y_b)\equiv\frac{1}{2k}\max_{(\alpha^x_a,\beta^y_b)\in\mathcal{SD}^d_c(\Gamma)}\left(\sum_{x,a}\alpha^x_a(V^x_a)+\sum_{y,b}\beta^y_b(W^y_b)\right)
\eeqn
By the closedness statement of proposition~\ref{bsdt}\ref{bsdb} together with boundedness of $\mathcal{SD}^d_c(\Gamma)$, it follows that $\mathcal{SD}^d_c(\Gamma)$ is actually compact. That is why the maximum in\eq{valuec} is known to be a maximum, and not just a supremum.

\begin{thm}
\label{steeringgame}
The maximal value of the bipartite steering game $(V^x_a,W^y_b)$ is given by \ldots
\begin{enumerate}
\item \ldots with the tensor product assumption,
\beq
\omega_\otimes(V^x_a,W^y_b)=\frac{1}{2k}\left|\left|\sum_{a,x}V^x_a\otimes e^x_a\otimes\mathbbm{1}+\sum_{b,y}W^y_b\otimes\mathbbm{1}\otimes e^y_b\right|\right|_{M_d(C^*(\Gamma)\otimes_{\min}C^*(\Gamma))}.
\eeq
\item \ldots with the commutativity assumption,
\beq
\omega_c(V^x_a,W^y_b)=\frac{1}{2k}\left|\left|\sum_{a,x}V^x_a\otimes e^x_a\otimes\mathbbm{1}+\sum_{b,y}W^y_b\otimes\mathbbm{1}\otimes e^y_b\right|\right|_{M_d(C^*(\Gamma)\otimes_{\max}C^*(\Gamma))}.
\eeq
\end{enumerate}
\end{thm}

\begin{proof}
Just like in the proof of theorem~\ref{bominmax}, this is clear from the previous theorem by recalling that the norm of a self-adjoint element in a $C^*$-algebra equals the maximum of the absolute value of this element under all states.
\end{proof}

\begin{thm}
\label{QWEPtoSDTP}
The following conjectural statements are equivalent for every Bell scenario $\Gamma\neq\Z_2\ast\Z_2$:
\begin{enumerate}
\item\label{qbsda} The QWEP conjecture\eq{qwep} holds,
\item\label{qbsdb} $C^*(\Gamma)\otimes_{\min} C^*(\Gamma)=C^*(\Gamma)\otimes_{\max} C^*(\Gamma)$,
\item\label{qbsdc} $\overline{\mathcal{SD}_\otimes^d(\Gamma)}=\mathcal{SD}_c^d(\Gamma)\quad\forall d$.
\item\label{qbsdd} For every dimension $d\in\N$ and every bipartite steering game $(V^x_a,W^y_b)$, its tensor product value and its commutative value coincide:
\beq
\omega_\otimes(V^x_a,W^y_b)=\omega_c(V^x_a,W^y_b).
\eeq
\end{enumerate}
\end{thm}

\begin{proof}
The equivalence of~\ref{qbsda} and~\ref{qbsdb} is theorem~\ref{FGamma}. Given that~\ref{qbsdb} holds, then~\ref{qbsdc} follows immediately from proposition~\ref{bsdt}. The implication from~\ref{qbsdc} to~\ref{qbsdd} follows from the definitions\eq{valuetp},\eq{valuec}.

Hence, it remains to prove that~\ref{qbsdd} implies~\ref{qbsdb}. To this end, we consider the linear subspace $S\subseteq C^*(\Gamma)\otimes_{\mathrm{alg}}C^*(\Gamma)$ defined as the linear span of the generators,
\beq
S\equiv\mathrm{lin}_{\C}\left\{\mathbbm{1}\otimes\mathbbm{1},e^x_a\otimes\mathbbm{1},\mathbbm{1}\otimes e^y_b\right\}.
\eeq
We write $S_{\min}$ for the normed space defined by restricting the minimal tensor norm $||\cdot||_{\min}$ from $C^*(\Gamma)\otimes_{\min}C^*(\Gamma)$ to $S$. More generally, we take $M_d(S_{\min})$ to be the normed space defined by restricting the norm from $M_d(C^*(\Gamma)\otimes_{\min}C^*(\Gamma))$ to $M_d(S)$. Similarly, we write $S_{\max}$ and $M_d(S_{\max})$ for the normed spaces obtained by restricting the maximal tensor norm $||\cdot||_{\max}$ from $C^*(\Gamma)\otimes_{\max}C^*(\Gamma)$ and $M_d(C^*(\Gamma)\otimes_{\max}C^*(\Gamma))$ to $S$ and $M_d(S)$, respectively.

Since the unitaries considered in remark~\ref{fourier} are contained in $S$ and generate $C^*(\Gamma)$, proposition~\ref{pisierobs} shows that~\ref{qbsdb} follows as soon as the norms on the self-adjoint parts of $M_d(S_{\min})$ and $M_d(S_{\max})$ coincide for all $d$. But now since an arbitrary element of the self-adjoint part of $M_d(S)$ is of the form
\beq
\sum_{a,x}V^x_a\otimes e^x_a\otimes\mathbbm{1}+\sum_{b,y}W^y_b\otimes\mathbbm{1}\otimes e^y_b,
\eeq
this actually follows from the assumption~\ref{qbsdd} together with theorem~\ref{steeringgame}.
\end{proof}

So, remarkably, for $\Gamma\neq\Z_2\ast\Z_2$ the conjecture $\mathbf{SDTP}(\Gamma)$\eq{SDTP} is equivalent to the QWEP conjecture, even though in our bipartite steering scenarios, Alice and Bob do not even measure jointly! Moreover, here it is sufficient to consider a \emph{single} scenario $\Gamma$. In particular, we therefore know that a positive answer to $\mathbf{SDTP}(\Gamma)$ for \emph{some} Bell scenario $\Gamma$ (except CHSH) implies a positive answer to the original Tsirelson's problem\eq{TP} for \emph{all} $\Gamma$. In other words, within in the framework of bipartite steering it becomes possible to choose our favorite Bell scenario $\Gamma$ (besides CHSH) in which to attack the QWEP Conjecture and all other versions of Tsirelson's problem at the same time.

\newpage
\appendix

\section{Complete positivity}
\label{ucp}

This section introduces unital completely positive maps and discusses some of their properties. All of the following sections of the appendix build on this material. When $A$ is a $C^*$-algebra, then $M_n(A)$ stands for the $C^*$-algebra of $n\times n$-matrices with entries in $A$. All $C^*$-algebras and all $*$-homomorphisms are unital, even when not explicitly mentioned, except for the ideal $J$ in the proof of proposition~\ref{statecomp}. As in the main text, the topology on a state space $\mathscr{S}(A)$ of a $C^*$-algebra $A$ is always the weak $*$-topology.

\begin{defn}
\label{ucpdefn}
Let $A$ and $B$ be unital $C^*$-algebras. A linear map $\Phi:A\ra B$ is called \emph{unital completely positive} (\textbf{ucp map}) if it satisfies
\begin{enumerate}
\item\label{uni} unitality: $\Phi(\mathbbm{1}_A)=\mathbbm{1}_B$.
\item\label{cp} complete positivity: for any $n\in\N$, the entrywise operating map
\beq
\Phi_n\::\:M_n(A)\lra M_n(B),\qquad\left(\begin{array}{ccc}a_{11}&\ldots&a_{1n}\\\vdots&&\vdots\\a_{n1}&\ldots&a_{nn}\end{array}\right)\mapsto \left(\begin{array}{ccc}\Phi(a_{11})&\ldots&\Phi(a_{1n})\\\vdots&&\vdots\\\Phi(a_{n1})&\ldots&\Phi(a_{nn})\end{array}\right)
\eeq
is positive. i.e.~for any $a\in M_n(A)$, it holds that
\beqn
\label{ncp}
a\geq 0\:\:\mathrm{for}\:\: a\in M_n(A)\Longrightarrow\:\Phi_n(a)\geq 0\:\:\mathrm{in}\:\:M_n(B). 
\eeqn
\end{enumerate}
\end{defn}

\begin{prop}[{e.g.~\cite[2.1]{Oza}}]
\label{ucpcon}
With this notation, a unital map $\Phi:A\ra B$ is completely positive if and only if for all $n\in\N$ and all $a\in M_n(a)$ with $a=a^*$,
\beqn
\label{contract}
\Phi(a)^*=\Phi(a)\qquad\textrm{and}\qquad||\Phi_n(a)||\leq ||a||.
\eeqn
\end{prop}

\begin{proof}
If $\Phi_n$ maps positive matrices to positive matrices, we get that $\Phi(a)^*=\Phi(a)$ by writing $a$ as a difference of two positive elements. Furthermore, $||a||\cdot\mathbbm{1}-a\geq 0$ for self-adjoint $a$ gives, by linearity and unitality,
$$
\Phi_n(||a||\cdot\mathbbm{1}-a) = ||a||\cdot\mathbbm{1} - \Phi_n(a) \geq 0
$$
which shows $\Phi_n$ to be norm-nonincreasing, so that both parts of\eq{contract} have been shown. Conversely, assume that\eq{contract} holds. If $a\in M_n(A)$ is positive, then $||a-||a||\cdot\mathbbm{1}||\leq ||a||$, and by assumption
\beq
||\Phi_n(a-||a||\cdot\mathbbm{1})||\leq ||a||,
\eeq
showing that $||\Phi_n(a)-||a||\cdot\mathbbm{1}||\leq||a||$ by unitality, so that the spectrum of the self-adjoint element $\Phi_n(a)$ lies in the interval $[0,2||a||]$, so that $\Phi_n(a)\geq 0$.
\end{proof}

\begin{lem}
\label{commucp}
When $A$ is a commutative unital $C^*$-algebra, then a linear map $\Phi:A\ra B$ is ucp as soon as it satisfies requirements~\ref{ucpdefn}\ref{uni} and~\ref{ucpdefn}\ref{cp} for $n=1$.
\end{lem}

\begin{proof}
Write $A\cong\mathscr{C}(X)$ for a compact Hausdorff space $X$, so that $M_n(A)\cong\mathscr{C}(X,M_n(\C))$, and note that an element of $\mathscr{C}(X,M_n(\C))$ is positive if and only if it is positive upon evaluation on all $x\in X$. For an alternative proof, see~\cite[3.11]{Paul}.
\end{proof}

Now let $\widehat{\H}$ be a Hilbert space and $\H\subseteq\widehat{\H}$ a closed subspace, $P_{\mathcal{H}}\in\B(\widehat{\H})$ the orthogonal projection onto this subspace, and $\pi:A\ra\B(\widehat{\H})$ some $*$-homomorphism. Then it is straightforward to show that the \emph{compression} of $\pi$ to $\mathcal{H}$, defined by the assignment
\beqn
\label{compression}
a\mapsto P_{\mathcal{H}}\pi(a)P_{\mathcal{H}},
\eeqn
is a ucp map from $A$ to $\B(\mathcal{H})$. The notorious Stinespring dilation theorem~\cite{Stine} states that every ucp map $A\to\B(\H)$ is of this form for some embedding $\H\subseteq\widehat{\H}$:

\begin{thm}[Stinespring dilation theorem]
\label{spthm}
Let $\Phi:A\ra\B(\H)$ be a ucp map. Then there is a Hilbert space $\widehat{\H}$ together with an isometric embedding $\H\subseteq\widehat{\H}$ and a $*$-homomorphism $\pi:A\ra\B(\widehat{\H})$ such that
\beqn
\label{sprep}
\Phi(a)=P_{\H}\pi(a)P_{\H}\quad\forall a\in A.
\eeqn
\end{thm}

We record the well-known proof here because a similar construction will be used later in the proof of theorem~\ref{spthm}.

\begin{proof}
The Hilbert space $\widehat{\H}$ will be constructed from the tensor product $A\otimes_{\C}\H$, a tensor product of complex vector spaces, in several steps. On $A\otimes_{\C}\H$, we define an inner product as
\beqn
\label{spsp}
\langle a\otimes\xi,a'\otimes\xi'\rangle\equiv \langle \xi,\Phi(a^*a')\xi'\rangle_\H
\eeqn
and extending by sesquilinearity. That $\Phi$ is ucp is now crucial for checking that this inner product is positive semi-definite: the scalar product of any tensor $\sum_{i=1}^n a_i\otimes\xi_i$ with itself evaluates to
\beqn
\label{stinespring}
\left\langle\sum_{i=1}^n a_i\otimes\xi_i,\sum_{i=1}^n a_i\otimes\xi_i\right\rangle=\sum_{i,j=1}^n\langle\xi_i,\Phi(a_i^*a_j)\xi_j\rangle_\H .
\eeqn
The expression on the right-hand side can be seen as being of the form $\langle\xi,\Phi_n(a^*a)\xi\rangle_{\oplus^n\H}$, with
\beq
\xi\equiv\left(\xi_1,\ldots,\xi_n\right)\in\oplus^n\H
\eeq
and $a\in M_n(A)$ being the matrix
\beq
a\equiv\left(\begin{array}{ccc}a_1 & \cdots & a_n\\ 0 & \cdots & 0\\ \vdots && \vdots\\ 0 & \cdots & 0 \\\end{array}\right).
\eeq
Therefore complete positivity of $\Phi$ yields $\Phi_n(a^*a)\geq 0$, which now implies that the scalar product\eq{spsp} is positive semi-definite. 

A standard calculation using the Cauchy-Schwarz inequality tells us that the null set
\beq
\mathcal{N}\equiv\left\{x\in A\otimes_{\C}\H\:|\langle x,x\rangle=0\:\right\}
\eeq
is a linear subspace of $A\otimes_{\C}\H$. Hence the quotient $(A\otimes_{\C}\H)/\mathcal{N}$ carries an induced inner product, which is positive definite by construction. The completion of this quotient, with respect to the norm induced by the inner product, is therefore a Hilbert space. This will be the desired Hilbert space $\widehat{\H}$. The assignment
\beq
\xi\mapsto \mathbbm{1}\otimes\xi+\mathcal{N}
\eeq
embeds $\H$ isometrically as a subspace of $\widehat{\H}$. By completeness of $\H$, this subspace is closed.

Furthermore, every element $a'\in A$ naturally acts on $\widehat{\H}$ via
\beqn
\label{stspact}
a'(a\otimes\xi+\mathcal{N})\equiv a'a\otimes\xi+\mathcal{N}
\eeqn
and extending by linearity. This is well-defined regarding the quotiening with respect to $\mathcal{N}$ since whenever\eq{stinespring} vanishes, then so does the expression
\beq
\left\langle\sum_{i=1}^n a'a_i\otimes\xi_i,\sum_{i=1}^n a'a_i\otimes\xi_i\right\rangle=\sum_{i,j=1}^n\langle\xi_i,\Phi(a_i^*a'^*a'a_j)\xi_j\rangle_\H\leq||a'^*a'||\langle\xi_i,\Phi(a_i^*a_j)\xi_j\rangle_\H=0
\eeq
where the estimate again uses the assumption of $\Phi$ being ucp. The same estimate shows that the action of $a'$ is bounded. Given this, it is immediate that\eq{stspact} defines a $*$-homomorphism $\pi:A\ra\B(\widehat{\H})$. 

Finally, we verify that these data satisfy the desired equation~(\ref{sprep}). For any $\xi\in\mathcal{H}$, which is of the form $\mathbbm{1}\otimes\xi+\mathcal{N}$ when considered as an element of $\widehat{\H}$,
\beq
\langle\mathbbm{1}\otimes \xi,P_{\H}\pi(a)P_{\H}(\mathbbm{1}\otimes\xi)\rangle_{\widehat{\H}}=\langle P_{\H}(\mathbbm{1}\otimes\xi),\pi(a)(\mathbbm{1}\otimes\xi)\rangle_{\widehat{\H}}=\langle\mathbbm{1}\otimes\xi,a\otimes\xi\rangle_{\widehat{\H}}\stackrel{(\ref{spsp})}{=}\langle\xi,\Phi(a)\xi\rangle_{\H}
\eeq
which therefore shows that\eq{sprep} does indeed hold.
\end{proof}

This result may be viewed as a classification theorem of the ucp maps $A\ra\B(\H)$. It is a very powerful principle which lets us prove other statements about ucp maps as simple corollaries:

\begin{prop}[{\cite[3.3]{Paul}}]
\label{CS}
If $\Phi:A\ra B$ is ucp, then
\beqn
\label{ncCS}
\Phi(a)^*\Phi(a)\leq\Phi(a^*a).
\eeqn
\end{prop}

\begin{proof}
Upon embedding $B$ into some $\B(\H)$, this is now a direct application of~\ref{spthm}. (For details, see the proof of the following proposition.) For proofs not relying on the Stinespring dilation theorem~\ref{spthm}, see~\cite[2.8]{Choi} for the original very general proof, or~\cite[3.3]{Paul} for another elegant and general argument.
\end{proof}

The inequality\eq{ncCS} may viewed as a noncommutative generalization of the fact that a random variable has positive variance, i.e.~that the square of its expectation value is bounded by the expectation value of its square.

The following proposition is a special case of Choi's theory of multiplicative domains~\cite[3.]{Choi}.

\begin{prop}
\label{multdomain}
Let $\Phi:A\ra B$ be ucp. Then the set of all $a\in A$ which saturate\eq{ncCS}, i.e.~which satisfy both the equations
\beqn
\label{mdass}
\Phi(a^*a)=\Phi(a)^*\Phi(a)\qquad\textrm{and}\qquad \Phi(aa^*)=\Phi(a)\Phi(a)^*
\eeqn
is called the \emph{multiplicative domain} of $A$. If $a\in A$ lies in the multiplicative domain and $x\in A$ is arbitrary, then
\beqn
\label{mdres}
\Phi(ax)=\Phi(a)\Phi(x)\qquad\textrm{and}\qquad\Phi(xa)=\Phi(x)\Phi(a).
\eeqn
Moreover, the multiplicative domain is a $C^*$-subalgebra of $A$, and the restriction of $\Phi$ to the multiplicative domain is a $*$-homomorphism.
\end{prop}

\begin{proof}
Upon choosing an embedding $B\subseteq \B(\H)$, we can assume without loss of generality that the codomain of $\Phi$ is $\B(\H)$. After enlarging $\H$, theorem~\ref{spthm} guarantees that $\Phi$ has the form
\beq
\Phi(a)=P\pi(a)P
\eeq
for some orthogonal projection $P\in\B(\H)$ and some representation $\pi:A\ra\B(\H)$. Hence,
\beq
\Phi(a^*a)=P\pi(a)^*\pi(a)P,\qquad \Phi(a)^*\Phi(a)=P\pi(a)^*P\pi(a)P.
\eeq
Comparing these two quantities gives
\beq
\Phi(a^*a)-\Phi(a)^*\Phi(a)=P\pi(a)^*(\mathbbm{1}-P)\pi(a)P=\left[(\mathbbm{1}-P)\pi(a)P\right]^*\left[(\mathbbm{1}-P)\pi(a)P\right],
\eeq
which in particular proves\eq{ncCS}. This expression vanishes if and only if
\beqn
\label{pap}
(\mathbbm{1}-P)\pi(a)P=0.
\eeqn
By the assumption\eq{mdass}, the same holds if we replace $a$ by $a^*$, and hence we obtain~(\ref{pap}) also for $a^*$. Taking the adjoint of this gives the equation
\beqn
\label{pap2}
P\pi(a)(\mathbbm{1}-P)=0.
\eeqn
Now we know that $\pi(a)$ commutes with $P$, since
\beq
P\pi(a)\stackrel{(\ref{pap2})}{=}P\pi(a)P\stackrel{(\ref{pap})}{=}\pi(a)P.
\eeq
So in terms of the direct sum decomposition $\mathcal{H}=P\H\oplus(\mathbbm{1}-P)\H$ into closed orthogonal subspaces, this means that $a$ lies in the multiplicative domain of $\Phi$ if and only if the operator $\pi(a)$ is block-diagonal. This implies in particular that the multiplicative domain is a $C^*$-subalgebra of $A$ and that the restriction of $\Phi$ to this $C^*$-subalgebra is a $*$-homomorphism.

Checking that this implies\eq{mdres} is now simple:
\beq
\Phi(ax)=P\pi(ax)P=PP\pi(a)\pi(x)P=P\pi(a)P\pi(x)P=\Phi(a)\Phi(x).
\eeq
This ends the proof.
\end{proof}

We record two more statements for further use and refer to the literature for their proofs.

\begin{thm}[{Arveson extension theorem~\cite[1.2.3]{Arv}}]
\label{arvext}
Let $S\subseteq A$ be a self-adjoint subspace containing $\mathbbm{1}_A$. Then any ucp map $\Phi:S\ra\B(\H)$ can be extended to a ucp map $\widehat{\Phi}:A\ra\B(\H)$,
\beq
\xymatrix{S\ar[rr]^\Phi\ar@{^{(}->}[dr] && \B(\H) \\
& A\ar@{-->}[ur]_{\exists\:\widehat{\Phi}}}
\eeq
\end{thm}

\begin{prop}[{\cite{OS}}]
\label{extfreeprod}
For ucp maps $\Phi:A\ra \B(\H)$ and $\chi:B\ra \B(\H)$, there is a ucp map $A\ast_1 B\ra \B(\H)$ which extends $\Phi$ and $\chi$:
\beq
\xymatrix{A\ar@{_{(}->}[d]\ar[rrd]^\Phi \\
A\ast_1 B\ar@{-->}[rr]^(.4){\exists} && \B(\H)\\
B\ar@{^{(}->}[u]\ar[rru]_\chi}
\eeq
\end{prop}

Alternatively, we also could have used a stronger result of Boca~\cite{Boca}, who has shown in particular that ucp maps $A\to C$ and $B\to C$ with some $C^*$-algebra $C$ as codomain extend to a ucp map $A\ast_1 B\to C$.

\section{Tensor products of $C^*$-algebras}
\label{tensor}

\paragraph{The algebraic tensor product.} Let $A\otimes_\C B$ denote the vector space tensor product of the two unital $C^*$-algebras $A$ and $B$. On elementary tensors, one can define a multiplication by
\beq
(a\otimes b)(a'\otimes b')\equiv aa'\otimes bb',
\eeq
and this extends to a multiplication on all of $A\otimes_\C B$ by bilinearity. This multiplication is associative and has a unit element given by $\mathbbm{1}_A\otimes \mathbbm{1}_B$, hence it turns $A\otimes_{\C}B$ into an algebra over $\C$. There is an involution on this algebra defined as the antilinear extension of
\beqn
\label{multtens}
(a\otimes b)^*\equiv a^*\otimes b^*.
\eeqn
Hence this construction gives $A\otimes_\C B$ the structure of a $*$-algebra. $A\otimes_\C B$ together with this $*$-algebra structure is also known as the \emph{algebraic tensor product} of $A$ and $B$ and denoted by $A\otimes_{\mathrm{alg}}B$. We will freely confuse $A$ and $B$ with their isomorphic images under the $*$-homomorphisms
\beq
A\lra A\otimes_{\mathrm{alg}}B,\quad a\mapsto a\otimes\mathbbm{1};\qquad B\lra A\otimes_{\mathrm{alg}}B,\quad b\mapsto \mathbbm{1}\otimes b.
\eeq
With this, the $*$-algebra $A\otimes_{\mathrm{alg}}B$ has the following universal property:

\begin{prop}
\label{algtens}
Any $*$-homomorphism 
\beq
\pi:A\otimes_{\mathrm{alg}}B\lra C,
\eeq
for some $*$-algebra $C$, restricts to $*$-homomorphisms $\pi_A:A\ra C$ and $\pi_B:B\ra C$ with commuting ranges. Conversely, given any such $\pi_A$ and $\pi_B$ with commuting ranges, there is a unique such $\pi$ extending these.
\end{prop}

\begin{proof}
The first assertion is clear since
\beqn
\label{commtens}
(a\otimes\mathbbm{1})\cdot(\mathbbm{1}\otimes b)=a\otimes b=(\mathbbm{1}\otimes b)\cdot(a\otimes\mathbbm{1})
\eeqn
holds in $A\otimes_{\mathrm{alg}}B$, so that the ranges commute. For the second assertion, the requirement that $\pi$ extends $\pi_A$ and $\pi_B$ means precisely that
\beq
\pi(a\otimes \mathbbm{1})=\pi_A(a),\qquad \pi(\mathbbm{1}\otimes b)=\pi_B(b).
\eeq
By\eq{commtens} and multiplicativity of $\pi$, this forces
\beq
\pi(a\otimes b)=\pi_A(a)\pi_B(b),
\eeq
which in turn extends uniquely to a linear map on all of $A\otimes_{\mathrm{alg}}B$ by the universal property of $A\otimes_\C B$. Using the commutativity of the ranges of $\pi_A$ and $\pi_B$, it is simple to check that this is a $*$-homomorphism.
\end{proof}

\begin{ex}
\label{matrixA}
As a basic example of this concept, we consider $M_n(\C)\otimes_{\mathrm{alg}}A$ for any unital $C^*$-algebra $A$. The $C^*$-algebra of $n\times n$-matrices $M_n(A)$ contains an isomorphic copy of $M_n(\C)$. It also contains $A$ as a unital subalgebra, included via the embedding $a\mapsto a\mathbbm{1}_n$. Since these two copies commute, we get an induced $*$-homomorphism $M_n(\C)\otimes_{\mathrm{alg}}A\ra M_n(A)$. It is simple to construct an inverse $*$-homomorphism, thereby showing a natural isomorphism $M_n(\C)\otimes_{\mathrm{alg}}A\cong M_n(A)$. This is why we also write $M_n(\C)\otimes A$ for $M_n(A)$; in fact, $M_n(\C)\otimes A$ coincides with both the minimal tensor product $M_n(\C)\otimes_{\min}A$ and the maximal tensor product $M_n(\C)\otimes_{\max} A$, which are constructions to be introduced in the following paragraphs.
\end{ex}

\paragraph{Cross norms.} A priori, the algebraic tensor product $A\otimes_{\mathrm{alg}}B$ does not come equipped with a norm, and in particular it does not possess the structure of a $C^*$-algebra. So the obvious question is now, is there a canonical way to turn $A\otimes_{\mathrm{alg}}B$ into a $C^*$-algebra? A common procedure for turning any $*$-algebra into a $C^*$-algebra is to specify a $C^*$-norm, i.e.~a norm satisfying the conditions
\beq
||xy||\leq ||x||\cdot ||y||,\qquad ||x||^2\leq||x^*x||,
\eeq
and then taking the Banach space completion of the $*$-algebra with respect to this norm. Since the multiplication and the involution extend to the completion by continuity, such a completion automatically carries the structure of a $C^*$-algebra.

In fact, it can be shown~\cite{KR2},~\cite[p.~216]{Tak} that any $C^*$-norm on $A\otimes_{\mathrm{alg}}B$ is a \emph{cross norm}, which means that
\beq
||a\otimes b||=||a||\cdot ||b||
\eeq
holds for all $a\in A$ and $b\in B$.

\begin{lem}
\label{linin}
Let $a_1,\ldots,a_n\in\B(\H)$ be linearly independent. Then there are unit vectors $\xi_1,\ldots,\xi_n\in\H$ such that the $n\times n$-matrix $(\langle\xi_i,a_j\xi_i\rangle)_{i,j=1}^n$ has full rank $n$.
\end{lem}

\begin{proof}
For each $\xi\in\H$, we consider $(\langle\xi,a_j\xi\rangle)_{j=1}^n$ as a vector in $\C^n$. The claim of the lemma is that the linear span of all these vectors, as indexed by $\xi$, has dimenions $n$; for then, it is possible to choose $n$ of these vectors which span the whole space. Now if the dimension would be $n-1$ or less, then there would exist a nontrivial linear relation $\sum_j\lambda_j\langle\xi,a_j\xi\rangle=0$ which would hold for all $\xi$. This would imply $\sum_j\lambda_j a_j=0$, contradicting the linear independence of the $a_j$.
\end{proof}

\paragraph{The minimal tensor product.} Upon choosing faithful representations of $A$ and $B$ on Hilbert spaces $\mathcal{H}_A$ and $\mathcal{H}_B$, respectively, we can consider $A$ as a $C^*$-subalgebra of $\mathcal{B}(\mathcal{H}_A)$ and $B$ as a $C^*$-subalgebra of $\mathcal{B}(\mathcal{H}_B)$. 

The following fact may appear intuitively obvious, but nevertheless we will offer a proof:

\begin{lem}
The embeddings $A\subseteq\B(\H_A)$ and $B\subseteq\B(\H_B)$ extend to a subalgebra inclusion
\beq
A\otimes_{\mathrm{alg}}B\subseteq\mathcal{B}(\mathcal{H}_A\otimes\mathcal{H}_B) .
\eeq
\end{lem}

\begin{proof}
If an element $\sum_{i=1}^n a_i\otimes b_i\in A\otimes_{\mathrm{alg}}B$ is nonzero, then it can be written in a form in which both the sets $\{a_i\}$ and $\{b_i\}$ are each linearly independent. Then lemma~\ref{linin} guarantees the existence of unit vectors $\xi_1,\ldots,\xi_n\in\H_A$ and $\eta_1,\ldots,\eta_n\in\H_B$ such that the matrices $(\langle\xi_i,a_j\xi_i\rangle)_{i,j=1}^n$ and $(\langle\eta_i,b_j\eta_i\rangle)_{i,j=1}^n$ both have full rank. Then the $n\times n$-matrix with entries
\beqn
\label{prodstates}
\left\langle \xi_k\otimes\eta_l,\left(\sum_{i=1}^n a_i\otimes b_i\right)\left(\xi_k\otimes\eta_l\right)\right\rangle=\sum_{i=1}^n\langle\xi_k,a_i\xi_k\rangle\langle\eta_l,b_i\eta_l\rangle
\eeqn
is the product of the matrix $(\langle\xi_k,a_i\xi_k\rangle)_{k,i=1}^n$ with the transpose of the matrix $(\langle\eta_l,b_i\eta_l\rangle)_{l,i=1}^n$. Since the latter two matrices both have full rank $n$, the former matrix also has full rank $n$, and therefore there are choices of indices $k$ and $l$ for which\eq{prodstates} does not vanish. In particular, the operator $\sum_{i=1}^na_i\otimes b_i\in\B(\H_A\otimes\H_B)$ also does not vanish.
\end{proof}

The norm closure of $A\otimes_{\mathrm{alg}}B$ in this representation is called the \emph{minimal tensor product} $A\otimes_{\min}B$. A basic example is the situation of example~\ref{matrixA}, where $B=M_n(\C)=\B(\C^n)$.

We will soon see that, surprisingly, the isomorphism class of $A\otimes_{\min}B$ does not depened on the embeddings $A\subseteq\B(\H_A)$ and $B\subseteq\B(\H_B)$. Proving this will need a little preparation; in particular, we will need the following proposition which shows that one can approximately turn an abstract $C^*$-algebraic state into a concrete mixed state on Hilbert space, i.e.~into a density matrix. (Thanks to the GNS construction, this is known to be exactly possible when the embedding is allowed to vary; so the point of the proposition is that the embedding is fixed.)

\begin{prop}[{e.g.~\cite[4.3.10]{KR2}}]
\label{vectorstates}
Let $C\subseteq\B(\H)$ be a concretely represented unital $C^*$-algebra and let $\rho\in\mathscr{S}(C)$ be some state. Then for every finite set of elements $x_1,\ldots,x_n\in C$ and every $\eps>0$, there are coefficients $\lambda_1,\ldots,\lambda_l\geq 0$ with $\sum_j\lambda_j=1$ and vectors $\xi_1,\ldots,\xi_l\in\H$ such that
\beq
\big|\rho(x_i)-\sum_j\lambda_j\langle\xi_j,x_i\xi_j\rangle\big|<\eps,\qquad i=1,\ldots,n.
\eeq
In other words: the mixed vector states $x\mapsto\sum_j\lambda_j\langle\xi_j,x\xi_j\rangle$ are dense in $\mathscr{S}(C)$ with respect to the weak $*$-topology.
\end{prop}

\begin{proof}
Since every $x_i$ can be written as a linear combination of two self-adjoint elements, we may assume without loss of generality the $x_i$'s to be self-adjoint; in general, this requires us to choose a different value for $\eps$.

Now consider the real linear space $S\equiv\mathrm{lin}_{\R}\{\mathbbm{1},x_1,\ldots,x_n\}$. It contains the closed convex cone $S_+\subseteq S$ of those elements which are positive in the $C^*$-algebra $C$. Every state $\rho\in\mathscr{S}(C)$ restricts to an element of the dual cone $S^*_+\subseteq S^*$. Conversely, every element of the dual cone $S^*_+\subseteq S^*$ can be extended to a (positive scalar multiple of a) state in $\mathscr{S}(C)$ by the Hahn-Banach theorem. Now every vector state $x\mapsto\langle\xi,x\xi\rangle$ certainly also lies in $S^*_+$; hence the set of unnormalized mixed states
\beq
x\mapsto\sum_{j=1}^l\lambda_j\langle\xi_j,x\xi_j\rangle,\qquad \lambda_1,\ldots,\lambda_l\geq 0\;,
\eeq
is a convex subcone of $S^*_+$. If this subcone were not dense in $S^*_+$, then the Hahn-Banach theorem would show the existence of some element of $S^{**}\cong S$ separating this proper subcone of $S^*_+$ from $S^*_+$ itself; this element $x$ would be a real linear combination of the $x_i$'s which satisfies $\langle\xi,x\xi\rangle\geq 0$ for all $\xi\in\H$, although there is some state $\rho\in\mathscr{S}(C)$ with $\rho(x)<0$. This is a contradiction.
\end{proof}

We can now apply this result to prove a concrete formula for the norm $||\cdot||_{\min}$ on $A\otimes_{\min}B$. We write $x=\sum_{i=1}^n a_i\otimes b_i$ for the elements of $A\otimes_{\mathrm{alg}}B$.

\begin{prop}
\begin{enumerate}
\item\label{mta} If $\rho_A\in\mathscr{S}(A)$ and $\rho_B\in\mathscr{S}(B)$ are any states, then
\beq
\left|(\rho_A\otimes\rho_B)(x)\right|=\left|\sum_{i=1}^n\rho_A(a_i)\rho_B(b_i)\right|\leq ||x||_{\min}
\eeq
for all $x=\sum_{i=1}^n a_i\otimes b_i\in A\otimes_{\min}B$. In other words, $\rho_A\otimes\rho_B\in\mathscr{S}(A\otimes_{\min}B)$.
\item\label{mtb} The norm of $x$ in $A\otimes_{\min}B$ is given by~\cite{Tur}
\beqn
\label{minnorm}
||x||^2_{\min}=||x^*x||_{\min}=\sup_{\rho_A,\rho_B,v}\frac{|\left(\rho_A\otimes\rho_B\right)\left(v^*x^*xv\right)|}{\left(\rho_A\otimes\rho_B\right)(v^*v)}
\eeqn
where the supremum runs over all states $\rho_A\in\mathscr{S}(A)$, $\rho_B\in\mathscr{S}(B)$ and elements $v\in A\otimes_{\mathrm{alg}}B$ having the property that $(\rho_A\otimes\rho_B)(v^*v)\neq 0$.
\end{enumerate}
\end{prop}

\begin{proof}
\begin{enumerate}
\item[\ref{mta}]\label{Parta}
We first check this when $\rho_A(a)=\langle\xi,a\xi\rangle$ for some $\xi\in\H_A$ and $\rho_B(b)=\langle\eta,b\eta\rangle$ for some $\eta\in\H_B$. Then,
\beq
(\rho_A\otimes\rho_B)(x)=\sum_i\langle\xi,a_i\xi\rangle\langle\eta,b_i\eta\rangle=\left\langle\xi\otimes\eta,\left(\sum_ia_i\otimes b_i\right)(\xi\otimes\eta)\right\rangle,
\eeq
so that in this case the statement follows from the definition of $||x||_{\min}$. Since it holds in this case, it also holds when $\rho_A$ is a convex combination of the form $\rho_A(a)=\sum_j\lambda_j\langle\xi_j,a\xi_j\rangle$, and likewise for $\rho_B$. For general $\rho_A$ and $\rho_B$, the assertion follows from proposition~\ref{vectorstates} by approximating them by such convex combinations on the given algebra elements $a_1,\ldots,a_n$ and $b_1,\ldots,b_n$.
\item[\ref{mtb}] Given $\rho_A\in\mathscr{S}(A)$, $\rho_B\in\mathscr{S}(B)$ and $v\in A\otimes_{\mathrm{alg}}B$ with $(\rho_A\otimes\rho_B)(v^*v)\neq 0$, we get the ``$\geq$'' part of\eq{minnorm} by noting that the assignment
\beq
y\mapsto\frac{\left(\rho_A\otimes\rho_B\right)\left(v^*yv\right)}{\left(\rho_A\otimes\rho_B\right)(v^*v)}
\eeq
is a state on $A\otimes_{\min}B$; this can be shown as in part~\ref{Parta} by first considering vector states and then applying proposition~\ref{vectorstates}.

For the converse direction, we sketch the existence of $\rho_A$, $\rho_B$ and $v$ for every $\eps>0$ such that the right-hand side of\eq{minnorm} is $\geq ||x^*x||_{\min}-\eps$. First of all, the vectors of the form $\sum_{j=1}^l\xi_j\otimes\eta_j$, i.e.~the finite linear combinations of elementary tensors, are dense in $\H_A\otimes\H_B$; hence for any $\delta_1>0$, there is some unit vector $\sum_{j=1}^l\xi_j\otimes\eta_j$ such that
\beq
\left\langle\left(\sum_{j=1}^l\xi_j\otimes\eta_j\right),x^*x\,\left(\sum_{j=1}^l\xi_j\otimes\eta_j\right)\right\rangle>||x^*x||_{\min}-\delta_1.
\eeq
We now consider the case that there are unit vectors $\xi\in\H_A$ and $\eta\in\H_B$ such that $A\xi$ is dense in $\H_A$ and $B\eta$ is dense in $\H_B$; since both embeddings split into a direct sum of cyclic representations~\cite[I.9.17]{Tak}, this is actually no loss of generality (see~\cite[IV.4.9.ii]{Tak} for details). By this assumption, for any $\delta_2>0$ there are elements $r_j\in A$ and $s_j\in B$ such that
\beq
||\xi_j-r_j\xi||_{\H_{A}}<\delta_2,\qquad ||\eta_j-s_j\eta||_{\H_B}<\delta_2,\qquad j=1,\ldots,l.
\eeq
Then upon setting $v\equiv\sum_{j=1}^lr_j\otimes s_j\in A\otimes_{\mathrm{alg}}B$, we find that the vector $v(\xi\otimes\eta)=\sum_{j=1}^l r_j\xi\otimes s_j\eta$ is a good approximation to $\sum_{j=1}^l\xi_j\otimes\eta_j$; in particular, the constants $\delta_1$ and $\delta_2$ can be chosen so small that
\beq 
\left\langle v(\xi\otimes\eta),x^*xv(\xi\otimes\eta)\right\rangle>\left(||x^*x||_{\min}-\eps\right)\left\langle v(\xi\otimes\eta),v(\xi\otimes\eta)\right\rangle
\eeq
Therefore choosing $\rho_A(\cdot)\equiv\langle\xi,\cdot\xi\rangle$ and $\rho_B(\cdot)\equiv\langle\eta,\cdot\eta\rangle$ together with the already defined $v$ makes the right-hand side of\eq{minnorm} greater than $||x^*x||_{\min}-\eps$, as claimed.
\end{enumerate}
\end{proof}

Note that the right-hand side of\eq{minnorm} does not depend on the chosen embeddings $A\subseteq\B(\H_A)$ and $B\subseteq\B(\H_B)$. Therefore as an immediate corollary, we obtain that the isomorphism class of the minimal tensor product $A\otimes_{\min}B$, including  the minimal tensor norm $||\cdot||_{\min}$, is independent of the chosen representations. A posteriori, this justifies the notations ``$A\otimes_{\min}B$'' and ``$||\cdot||_{\min}$'' which do not explicitly mention the embeddings.

\begin{cor}
If $A_1$, $A_2$, $B_2$ and $B_2$ are unital $C^*$-algebras, and $\pi_A:A_1\ra A_2$, $\pi_B:B_1\ra B_2$ are $*$-homomorphisms, then
\beq
\left(\pi_A\otimes_{\min}\pi_B\right)(a\otimes b)\equiv \pi_A(a)\otimes\pi_B(b)
\eeq
defines a $*$-homomorphism $\pi_A\otimes_{\min}\pi_B: A_1\otimes_{\min}B_1\lra A_2\otimes_{\min}B_2$.
\end{cor}

\begin{proof}
By the universal property of proposition~\ref{algtens}, this is clear on the level of the algebraic tensor product $\otimes_{\mathrm{alg}}$ in place of $\otimes_{\min}$. What remains to be shown is continuity of the resulting $*$-homorphism $\pi_A\otimes_{\mathrm{alg}}\pi_B$ with respect to $||\cdot||_{\min}$. We will use the result that the $||\cdot||_{\min}$-norm on $A_i\otimes_{\mathrm{alg}} B_j$ does not depend on the choice of embeddings.

Since $\pi_A\otimes_{\mathrm{alg}}\pi_B$ factors as
\beq
\pi_A\otimes_{\mathrm{alg}}\pi_B=(\id_{A_1}\otimes_{\mathrm{alg}}\pi_B)\circ(\pi_A\otimes_{\mathrm{alg}}\id_{B_2}),
\eeq
it is enough to consider only the case $B_1=B_2=B$ and $\pi_B=\id_B$. Then when $A_1\subseteq\B(\H_1)$ and $A_2\subseteq\B(\H_2)$ are any embeddings, then so is
\beq
A_1\lra \B(\H_1\oplus\H_2),\qquad a\mapsto a\oplus\pi_A(a).
\eeq
With this, the $*$-homomorphism $\pi_A$ simply becomes the compression onto the subspace $\H_2\subseteq\H_1\oplus\H_2$. Similarly, we can identify $\pi_A\otimes\id_B$ as the compression onto the subspace $\H_2\otimes\H_B\subseteq \left(\H_1\oplus\H_2\right)\otimes\H_B$, as in the diagram
\beq
\xymatrix{ A_1\otimes_{\mathrm{alg}}B \ar[rr]^{\pi_A\otimes_{\mathrm{alg}}\id_B} \ar@{^{(}->}[d] && A_2\otimes_{\mathrm{alg}}B \ar@{^{(}->}[d] \\
\B\left((\H_1\oplus\H_2)\otimes\H_B\right) \ar[rr]^(.56){\textrm{compression}} && \B\left(\H_2\otimes\H_B\right) }
\eeq
This implies $||\pi_A\otimes_{\mathrm{alg}}\id_B||_{\min}\leq 1$.
\end{proof}

It has further been shown by Takesaki~\cite{Takpap} that the cross norm $||\cdot||_{\min}$ is in fact the smallest possible $C^*$-norm on the $*$-algebra $A\otimes_{\mathrm{alg}}B$. Since $||a\otimes b||\leq||a||\cdot||b||$ needs to hold for any $C^*$-norm on $A\otimes_{\mathrm{alg}}B$, and by Takesaki's result the minimal $C^*$-norm already saturates this inequality, this implies the already mentioned result that every $C^*$-norm on $A\otimes_{\mathrm{alg}}B$ is a cross norm.

\begin{prop}
\label{minproducp}
If $\Phi_A:A\ra\B(\H_A)$ and $\Phi_B:B\ra\B(\H_B)$ are ucp maps, then the map defined by
\beq
\Phi_A\otimes_{\min}\Phi_B:A\otimes_{\min}B\lra\B(\H_A\otimes\H_B),\qquad a\otimes b\mapsto \Phi_A(a)\otimes\Phi_B(b)
\eeq
is also ucp.
\end{prop}

\begin{proof}
Applying the Stinespring dilation theorem~\ref{spthm} yields Hilbert spaces $\widehat{\H}_A\supseteq\H_A$ and $\widehat{\H}_B\supseteq\H_B$ together with $*$-homomorphisms $\pi_A$ and $\pi_B$ forming commutative diagrams
\beq
\xymatrix{ & \B(\widehat{\H}_A) \ar@{->>}[dd] &&& \B(\widehat{\H}_B) \ar@{->>}[dd] \\
A \ar[ur]^{\pi_A} \ar[dr]_{\Phi_A} &&& B \ar[ur]^{\pi_B} \ar[dr]_{\Phi_B} \\
 & \B(\H_A) &&& \B(\H_B) }
\eeq
where the vertical maps stand for compressions to the respective subspace. The previous corollary then defines the tensor product $*$-homomorphism $A\otimes_{\min}B\lra\B(\widehat{\H}_A)\otimes_{\min}\B(\widehat{\H}_B)\subseteq\B(\widehat{\H}_A\otimes\widehat{\H}_B)$. This map can be composed with the compression to the subspace $\H_A\otimes\H_B\subseteq\widehat{\H}_A\otimes\widehat{\H}_B$ to give the desired $\Phi_A\otimes_{\min}\Phi_B$, which maps $a\otimes b\mapsto  \Phi_A(a)\otimes\Phi_B(b)$ as desired.
\end{proof}

This finishes our overview of the minimal tensor product, including all those results relevant for the main text.

\paragraph{The maximal tensor product.}

There is another canonical norm on $A\otimes_{\mathrm{alg}}B$ defined as follows:
\beqn
\label{maxdef}
||x||_{\max}=\sup_{\pi}||\pi(x)||,\quad x\in A\otimes_{\mathrm{alg}}B,
\eeqn
where the supremum runs over all representations $\pi:A\otimes_{\mathrm{alg}}B\ra\mathcal{B}(\mathcal{H})$; by the universal property~(\ref{algtens}), such a representation is equivalent to two representations $\pi_A:A\ra\B(\H)$ and $\pi_B:B\ra\B(\H)$ with commuting ranges. 

The formula (\ref{maxdef}) can immediately be shown to define a $C^*$-norm. The completion of $A\otimes_{\mathrm{alg}}B$ with respect to this norm is a $C^*$-algebra, taken to be the \emph{maximal tensor product} $A\otimes_{\mathrm{max}}B$. Essentially by definition, it satisfies the following universal property, which transfers~\ref{algtens} from the world of $*$-algebras to the world of $C^*$-algebras:

\begin{prop}
\label{maxtens}
Any $*$-homomorphism
\beq
\pi:A\otimes_{\max}B\lra C,
\eeq
for some $C^*$-algebra $C$, restricts to $*$-homomorphisms $\pi_A:A\ra C$ and $\pi_B:B\ra C$ with commuting ranges. Conversely, given any such $\pi_A$ and $\pi_B$ with commuting ranges, there is a unique $\pi$ extending these.
\end{prop}

\begin{proof}
By proposition~\ref{algtens} and the definition in equation\eq{maxdef}.
\end{proof}

Before proceeding with properties of the maximal tensor product, we prove a tensor product version of the Stinespring dilation theorem~\ref{spthm}. We begin with a lemma.

\begin{lem}[{compare~\cite[1.2]{Schurproducts}}]
\label{entrywise}
If $\mathcal{A},\mathcal{B}\in M_n(\B(\H))$ are matrices of operators, such that both $\mathcal{A}\geq 0$ and $\mathcal{B}\geq 0$ are positive and all entries pairwise commute, i.e.~$\mathcal{A}_{ij}\mathcal{B}_{kl}=\mathcal{B}_{kl}\mathcal{A}_{ij}$ for all indices $i$, $j$, $k$, $l$, then the entrywise product
\beq
\left(\mathcal{A}\odot\mathcal{B}\right)_{ij}\equiv\mathcal{A}_{ij}\mathcal{B}_{ij}\:\in M_n(\B(\H))
\eeq
is also positive, i.e.~$\mathcal{A}\odot\mathcal{B}\geq 0$ in $M_n(\B(\H))$.
\end{lem}

\begin{proof}
By the assumption of commutativity, we conclude that also all the entries of the square roots of $\mathcal{A}$ and $\mathcal{B}$ do pairwise commute, i.e.~$\sqrt{\mathcal{A}}_{\,ij}\sqrt{\mathcal{B}}_{\,kl}=\sqrt{\mathcal{B}}_{\,kl}\sqrt{\mathcal{A}}_{\,ij}$, since these entries are limits of polynomials in the $\mathcal{A}_{ij}$ and the $\mathcal{B}_{kl}$, respectively. Using this, the scalar product $\langle\xi,(\mathcal{A}\odot\mathcal{B})\xi\rangle$ on any $\xi\in\oplus^n\H$ can be evaluated to
\begin{align*}
\langle\xi,(\mathcal{A}\odot\mathcal{B})\xi\rangle_{\oplus^n\H}=&\sum_{i,j}\langle\xi_i,\mathcal{A}_{ij}\mathcal{B}_{ij}\xi_j\rangle_{\H}=\sum_{i,j,k,l}\langle\xi_i,\sqrt{\mathcal{A}}_{\,ik}\sqrt{\mathcal{A}}_{\,kj}\sqrt{\mathcal{B}}_{\,il}\sqrt{\mathcal{B}}_{\,lj}\xi_j\rangle_{\H}\\
=&\sum_{i,j,k,l}\langle\xi_i,\sqrt{\mathcal{A}}_{\,ik}\sqrt{\mathcal{B}}_{\,il}\sqrt{\mathcal{B}}_{\,lj}\sqrt{\mathcal{A}}_{\,kj}\xi_j\rangle_{\H}=\sum_{i,j,k,l}\langle \sqrt{\mathcal{B}}_{\,li}\sqrt{\mathcal{A}}_{\,ki}\xi_i,\sqrt{\mathcal{B}}_{\,lj}\sqrt{\mathcal{A}}_{\,kj}\xi_j\rangle_{\H}
,
\end{align*}
where the last step relies on $\sqrt{\mathcal{A}}_{\,ki}^{\,*}=\sqrt{\mathcal{A}}_{\,ik}$, and similarly $\sqrt{\mathcal{B}}_{\,ki}^{\,*}=\sqrt{\mathcal{B}}_{\,ik}$. Writing $\mathcal{C}_{lk,j}\equiv \sqrt{\mathcal{B}}_{\,lj}\sqrt{\mathcal{A}}_{\,kj}\in M_{n^2,n}(\B(\H))$, this expression reads as
\beq
\langle\xi,(\mathcal{A}\odot\mathcal{B})\xi\rangle_{\oplus^n\H}=\langle\mathcal{C}\xi,\mathcal{C}\xi\rangle_{\oplus^{n^2}\H}\geq 0
\eeq
and hence gives the assertion.
\end{proof}

\begin{thm}[double Stinespring theorem]
\label{spthm2}
Let $\Phi_A:A\ra\B(\H)$ and $\Phi_B:B\ra\B(\H)$ be ucp maps with commuting ranges. Then there is a Hilbert space $\widehat{\H}$ together with an isometric embedding $\H\subseteq\widehat{\H}$ and a $*$-homomorphism $\pi:A\otimes_{\mathrm{max}}B\ra\B(\widehat{\H})$ such that
\beqn
\label{sprep2}
\Phi_A(a)\Phi_B(b)=P_{\H}\pi(a\otimes b)P_{\H}\quad\forall a\in A,\:b\in B.
\eeqn
\end{thm}

\begin{proof}
This is totally analogous to the proof of theorem~\ref{spthm}. The Hilbert space $\widehat{\H}$ will be constructed from the tensor product $A\otimes_{\C}B\otimes_{\C}\H$, a tensor product of complex vector spaces, in several steps. On $A\otimes_{\C}B\otimes_{\C}\H$, we define an inner product as
\beqn
\label{spsp2}
\langle a\otimes b\otimes\xi,a'\otimes b'\otimes\xi'\rangle\equiv \langle \xi,\Phi_A(a^*a')\Phi_B(b^*b')\xi'\rangle_\H
\eeqn
and extending by sesquilinearity. That $\Phi_A$ and $\Phi_B$ are ucp is now crucial for checking that this inner product is positive semi-definite: the scalar product of any tensor $\sum_{i=1}^n a_i\otimes b_i\otimes\xi_i$ with itself evaluates to
\beqn
\label{stinespring2}
\left\langle\sum_{i=1}^n a_i\otimes b_i\otimes\xi_i,\sum_{i=1}^n a_i\otimes b_i\otimes\xi_i\right\rangle=\sum_{i,j=1}^n\langle\xi_i,\Phi_A(a_i^*a_j)\Phi_B(b_i^*b_j)\xi_j\rangle_\H.
\eeqn
As in the proof of~\ref{spthm}, the operator matrix $\mathcal{A}_{ij}\equiv\Phi_A(a_i^*a_j)\in M_n(A)$ is positive, and likewise is $\mathcal{B}_{ij}\equiv\Phi_B(b_i^*b_j)\in M_n(B)$. By lemma~\ref{entrywise} and the assumption of commuting ranges, this then shows non-negativity of\eq{stinespring2}, so that\eq{spsp2} is positive semi-definite. Again, a standard calculation using the Cauchy-Schwarz inequality tells us that the null set
\beq
\mathcal{N}\equiv\left\{x\in A\otimes_{\C}B\otimes_{\C}\H\:|\langle x,x\rangle=0\:\right\}
\eeq
is a linear subspace of $A\otimes_{\C}B\otimes_{\C}\H$. Hence the quotient $(A\otimes_{\C}B\otimes_{\C}\H)/\mathcal{N}$ carries an induced inner product, which is positive definite by construction. The completion of this quotient, with regard to the norm induced by the inner product, is therefore a Hilbert space. This will be the desired Hilbert space $\widehat{\H}$. The assignment
\beq
\xi\mapsto \mathbbm{1}\otimes\mathbbm{1}\otimes\xi+\mathcal{N}
\eeq
embeds $\H$ isometrically as a subspace of $\widehat{\H}$. By completeness of $\H$, this subspace is closed.

Furthermore, every element of $A\otimes_{\mathrm{alg}}B$ naturally acts on $\widehat{\H}$; we define this action on elementary tensors $a'\otimes b'\in A\otimes_{\mathrm{alg}}B$ as
\beqn
\label{stspact2}
(a'\otimes b')(a\otimes b\otimes\xi+\mathcal{N})\equiv a'a\otimes b'b\otimes\xi+\mathcal{N}
\eeqn
and extending by linearity. This is well-defined regarding the quotiening with respect to $\mathcal{N}$ since whenever\eq{stinespring2} vanishes, then so does the expression
\begin{align*}
\left\langle\sum_{i=1}^n a'a_i\otimes b'b_i\otimes\xi_i,\sum_{i=1}^n a'a_i\otimes b'b_i\otimes\xi_i\right\rangle&=\sum_{i,j=1}^n\langle\xi_i,\Phi_A(a_i^*a'^*a'a_j)\Phi_B(b_i^*b'^*b'b_j)\xi_j\rangle_\H\\[.5cm]
&\leq||a'^*a'||\sum_{i,j=1}^n\langle\xi_i,\Phi_A(a_i^*a_j)\Phi_B(b_ib'^*b'^*b_j)\xi_j\rangle_\H\\[.5cm]
&\leq||a'^*a'||\cdot||b'^*b'||\sum_{i,j=1}^n\langle\xi_i,\Phi_A(a_i^*a_j)\Phi_B(b_i^*b_j)\xi_j\rangle_\H=0
\end{align*}
where the two estimates again use the assumption of $\Phi_A$ and $\Phi_B$ being ucp as well as the commutativity of their ranges needed for the applications of lemma~\ref{entrywise}. The same estimate shows that the action of $a'\otimes b'$ is bounded. It is straightforward to check that\eq{stspact2} defines a $*$-homomorphism $\pi:A\otimes_{\mathrm{alg}}B\ra\B(\widehat{\H})$. 

Finally, we verify that these data satisfy the desired equation\eq{sprep2}. For any $\xi\in\mathcal{H}$, which is of the form $\mathbbm{1}\otimes\mathbbm{1}\otimes\xi+\mathcal{N}$ when considered as an element of $\widehat{\H}$,
\begin{align*}
\langle\mathbbm{1}\otimes\mathbbm{1}\otimes \xi,P_{\H}\pi(a\otimes b)P_{\H}(\mathbbm{1}\otimes\mathbbm{1}\otimes\xi)\rangle_{\widehat{\H}}&=\langle P_{\H}(\mathbbm{1}\otimes\mathbbm{1}\otimes\xi),\pi(a\otimes b)(\mathbbm{1}\otimes\mathbbm{1}\otimes\xi)\rangle_{\widehat{\H}}\\[.3cm]
&=\langle\mathbbm{1}\otimes\mathbbm{1}\otimes\xi,a\otimes b\otimes\xi\rangle_{\widehat{\H}}=\langle\xi,\Phi_A(a)\Phi_B(b)\xi\rangle_{\H}
\end{align*}
which therefore shows that\eq{sprep2} does indeed hold.
\end{proof}

As a direct corollary, we obtain a well-known result~\cite[IV.4.23(ii)]{Tak} analogous to proposition~\ref{minproducp}. For another alternative proof of this, see~\cite[12.8]{Paul}.

\begin{cor}
\label{maxproducp}
If $\Phi_A:A\ra\B(\H)$ and $\Phi_B:B\ra\B(\H)$ are ucp maps with commuting ranges, then the map
\beq
\Phi_A\otimes_{\max}\Phi_B:A\otimes_{\max}B\lra\B(\H),\qquad a\otimes b\mapsto \Phi_A(a)\Phi_B(b)
\eeq
is well-defined and ucp.
\end{cor}

\begin{proof}
Upon remembering that a compression like\eq{compression} is always ucp, this is a direct consequence of the double Stinespring dilation theorem~\ref{spthm2} and the universal property of proposition~\ref{maxtens}.
\end{proof}

\paragraph{Comparison between $\otimes_{\min}$ and $\otimes_{\max}$.} In general, the maximal and the minimal tensor product of two $C^*$-algebras are different: in general, the norm of an ``entangled'' operator $\sum_{i=1}^n a_i\otimes b_i$ (with $n>1$) depends on the choice of cross norm on $A\otimes_{\mathrm{alg}}B$, with $||\cdot||_{\min}$ and $||\cdot||_{\max}$ being the two extreme cases.

Since $||\cdot||_{\min}\leq||\cdot||_{\max}$, or alternatively by the universal property of $A\otimes_{\max}B$ from proposition~\ref{maxtens}, there is a natural comparison map
\beqn
\label{natcomp}
A\otimes_{\max}B\twoheadrightarrow A\otimes_{\min}B
\eeqn
which is a $*$-homomorphism. The comparison map is necessarily surjective since the image of a $*$-ho\-mo\-mor\-phism is a sub-$C^*$-algebra of the codomain, and in this case the dense subalgebra $A\otimes_{\mathrm{alg}}B\subseteq A\otimes_{\max}B$ lies in the image of the $*$-homomorphism. Therefore, the image of the $*$-homomorphism is the whole codomain $A\otimes_{\max}B$.

The following example where $A\otimes_{\min}B\neq A\otimes_{\max}B$ is a modified version of~\cite[11.3.14]{KR2}, the original proof of which appeared in~\cite{Takpap}.

\begin{ex}
\label{minmax}
Let $G=\mathbb{F}_2$ be the free group on two generators $g_1,g_2$. The Hilbert space $\ell^2(G)$ is defined to be the vector space of all square-summable $\C$-valued functions on $G$ equipped with the usual componentwise scalar product. It has an orthonormal basis given in terms of the set $\{\delta_g,\: g\in G\}$ of indicator functions
\beq
\delta_g(h)=\left\{\begin{array}{cl}1& \textrm{for }\:h=g,\\ 0 & \textrm{for }\:h\neq g.\end{array}\right.
\eeq
Now every group element $x\in G$ acts unitarily on $\ell^2(G)$ via left multiplicaton $L_x$ and via right multiplication $R_x$,
\begin{align}
\begin{split}
\label{LR}
&L_x:\delta_g\mapsto\delta_{xg}\\
&R_x:\delta_g\mapsto\delta_{gx^{-1}}
\end{split}
\end{align}
where both operators are defined in terms of their action on basis elements. Since $L_{xy}=L_xL_y$ and $R_{xy}=R_xR_y$, we have two unitary representations of $G$ on $\ell^2(G)$, known as the \emph{left regular representation} and \emph{right regular representation}, respectively. The left regular representation $\{L_x,\: x\in G\}$ generates a $C^*$-algebra, called the \emph{reduced group $C^*$-algebra} $C^*_r(G)$. Similarly, the right regular representation $\{R_x,\: x\in G\}$ generates a $C^*$-algebra which is isomorphic to $C^*_r(G)$; this isomorphism follows from the fact that the left and right regular representation are unitarily equivalent via the unitary map defined by $\delta_g\mapsto\delta_{g^{-1}}$. Until here, this has been the standard construction of $C^*_r(G)$.

Now we can note that these two copies of $C^*_r(G)$ commute by construction,
\beq
L_xR_{y}=R_{y}L_x\qquad\forall\: x\in G,\:y\in G.
\eeq
Therefore, these two copies define a representation $\pi$ of $C_r^*(G)\otimes_{\max} C_r^*(G)$. Due to the way it has been constructed, this representation is also known as the \emph{biregular representation}.

Consider now the self-adjoint operator
\beq
\Delta\equiv L_{g_1}\otimes L_{g_1}+L_{g_1}^*\otimes L_{g_1}^*+L_{g_2}\otimes L_{g_2}+L_{g_2}^*\otimes L_{g_2}^*
\eeq
as an element of $C_r^*(G)\otimes_{\mathrm{alg}} C_r^*(G)$. In the biregular representation, this element has the form
\beq
\pi(\Delta)=L_{g_1}R_{g_1}+L_{g_1}^*R_{g_1}^*+L_{g_2}R_{g_2}+L_{g_2}^*R_{g_2}^*
\eeq
The definition\eq{LR} implies that $\pi(\Delta)\delta_{e}=4\delta_{e}$, where $e\in G$ is the neutral element, and hence $||\pi(\Delta)||\geq 4$. But since $\Delta$ is a sum of $4$ unitaries, this lower bound has to be tight. Since the biregular representation is a representation of $C_r^*(G)\otimes_{\max}C_r^*(G)$, we can conclude that
\beqn
\label{Amax}
||\Delta||_{\max}=4\qquad\textrm{in}\qquad C_r^*(G)\otimes_{\max}C_r^*(G).
\eeqn

On the other hand, we can also consider $\Delta$ as an element of $C_r^*(G)\otimes_{\min}C_r^*(G)=C_r^*(G\times G)$. (This equality holds since two tensor copies of the right regular representation of $G$ give the right regular representation of $G\times G$.) Under the diagonal homomorphism $G\ra G\times G$, which induces an embedding $C_r(G)\ra C_r^*(G\times G)$, the element $\Delta$ is the image of
\beq
\widehat{\Delta}=L_{g_1}+L_{g_1}^*+L_{g_2}+L_{g_2}^*
\eeq
and therefore $||\Delta||_{\min}=||\widehat{\Delta}||$. But now a direct calculation by the methods of~\cite{PP} shows that $||\widehat{\Delta}||=\sqrt{12}<4$, which gives the conclusion upon comparison with\eq{Amax}. (As a curiosity, note that $\widehat{\Delta}/4$ is the generator of the random walk on the Cayley graph of $\F_2$. Related to this, $\mathbbm{1}-\widehat{\Delta}/4$ is the Cayley graph's Laplace operator, which explains our funny notation.)

More generally, one can show that a similar phenomenon occurs for any group $G$ which is not amenable; see~\cite{Lance},~\cite{AD} for details.
\end{ex}


Deciding whether $A\otimes_{\min}B$ and $A\otimes_{\max}B$ agree for given $A$ and $B$ is often a very difficult problem. There is a useful criterion for equality due to Pisier~\cite[p.~6]{Pis2}. This criterion starts from a given set of unitaries $\{u_i\}\subseteq A$ which generates $A$ as a $C^*$-algebra, and a set of unitaries $\{v_j\}\subseteq B$ which generates $B$ as a $C^*$-algebra. Pisier then considers the vector space
\beq
S\equiv\mathrm{lin}_\C\left\{\mathbbm{1}\otimes\mathbbm{1},u_i\otimes\mathbbm{1},u_i^*\otimes\mathbbm{1},\mathbbm{1}\otimes v_j,\mathbbm{1}\otimes v_j^*\right\}\subseteq A\otimes_{\mathrm{alg}}B.
\eeq
The space $M_n(S)$ of $n\times n$-matrices over $S$ inherits a norm $||\cdot||_{\min}$ from its embedding into $M_n(A\otimes_{\min}B)$ and another norm $||\cdot||_{\max}$ from its embedding into $M_n(A\otimes_{\max}B)$. It is clear that the latter dominates the former,
\beq
||x||_{\min}\leq ||x||_{\max}\quad \forall x\in M_n(S).
\eeq
Possibly surprisingly, if these two norms coincide for all $n\in\N$, then the two tensor products coincide:

\begin{prop}[{\cite[p.~6]{Pis2}}]
\label{pisierobs}
In this situation, $A\otimes_{\max}B=A\otimes_{\min}B$ if and only if $||x||_{\min}=||x||_{\max}$ for all $x\in M_n(S)$ with $x^*=x$ and all $n\in\N$.
\end{prop}

\begin{proof}
The ``only if'' part is clear, while the ``if'' part is nontrivial. By proposition~\ref{ucpcon} and Arveson's extension theorem~\ref{arvext}, the assumption guarantees the existence of a ucp map
\beq
\Psi:A\otimes_{\min}B\lra A\otimes_{\max}B
\eeq
which is the identity on $S$. By this property, all the unitaries $u_i\otimes\mathbbm{1}$ and $\mathbbm{1}\otimes v_j$ lie in the multiplicative domain of $\Psi$ in the sense of proposition~\ref{multdomain}. But since these unitaries generate $A\otimes_{\min}B$ as a $C^*$-algebra, proposition~\ref{multdomain} shows that $\Psi$ is in fact a $*$-homomorphism. By construction, it is inverse to\eq{natcomp}, so that\eq{natcomp} is actually an isomorphism.
\end{proof}

The potential difference between $A\otimes_{\min} B$ and $A\otimes_{\max} B$ gets nicely reflected in the state spaces of these two $C^*$-algebras:

\begin{prop}
\label{statecomp}
The natural projection
\beqn
\label{natcomp2}
A\otimes_{\max}B\twoheadrightarrow A\otimes_{\min}B
\eeqn
induces an inclusion of state spaces
\beqn
\label{stateface}
\mathscr{S}(A\otimes_{\min}B)\hookrightarrow \mathscr{S}(A\otimes_{\max}B).
\eeqn
\begin{enumerate}
\item\label{sca} Under this inclusion, the convex set $\mathscr{S}(A\otimes_{\min}B)$ becomes a closed face of $\mathscr{S}(A\otimes_{\max}B)$.
\item\label{scb} $A\otimes_{\max}B\neq A\otimes_{\min}B$ if and only if there is a state $\rho\in\mathscr{S}(A\otimes_{\max}B)\setminus\mathscr{S}(A\otimes_{\min}B)$.
\item\label{scc} Let $G$ be a finite group acting by $*$-(anti-)automorphisms on $A\otimes_{\mathrm{alg}}B$. Then $A\otimes_{\max}B\neq A\otimes_{\min}B$ if and only if there is a $G$-invariant state $\rho\in\mathscr{S}(A\otimes_{\max}B)\setminus\mathscr{S}(A\otimes_{\min}B)$.
\end{enumerate}
\end{prop}

\begin{proof}
\begin{enumerate}
\item[\ref{sca}]
Clearly if two states differ on $A\otimes_{\min}B$, they also differ on $A\otimes_{\max}B$ by surjectivity of\eq{natcomp2}, so it is rather trivial that\eq{stateface} is injective. Since the state space of a $C^*$-algebra is compact in the weak $*$-topology, continuity of\eq{stateface} shows that its image is closed. It remains to check that its image is a face. To this end, note that $\mathscr{S}(A\otimes_{\min}B)$ consists of exactly those elements of $\mathscr{S}(A\otimes_{\max}B)$ which vanish on the kernel of\eq{natcomp2}; call this kernel $J$. Then suppose that
\beq
\rho=\lambda \rho_0+(1-\lambda)\rho_1
\eeq
is a state $\rho\in\mathscr{S}(A\otimes_{\min}B)$ which is a convex combination of states $\rho_0,\rho_1\in\mathscr{S}(A\otimes_{\max}B)$ with $\lambda\in(0,1)$. Let $j\in J$ be a positive element. Then $\rho(j)=0$, and hence positivity show that $\rho_0(j)=0=\rho_1(j)$ as well. But now since all elements of $J$ are linear combinations of positives, it follows that $\rho_0$ and $\rho_1$ vanish on all of $J$, so that both $\rho_0$ and $\rho_1$ lie in $\mathscr{S}(A\otimes_{\min}B)$.
\item[\ref{scb}] While the ``if'' direction is trivial, the ``only if'' direction is also immediate by choosing any state $\rho\in\mathscr{S}(A\otimes_{\max}B)$ which does not identically vanish on the ideal $J$.
\item[\ref{scc}] Only the ``only if'' direction is nontrivial. Starting with some state $\rho_0\in\mathscr{S}(A\otimes_{\max}B)\setminus\mathscr{S}(A\otimes_{\min}B)$, our goal is to turn it into a $G$-invariant state. Since $G$ acts by (anti-)automorphisms on both $A\otimes_{\min}B$ and $A\otimes_{\max}B$, the translated state $\rho_0(g(\cdot))$ is also not in $\mathscr{S}(A\otimes_{\min}B)$ for any $g\in G$. Then since $\mathscr{S}(A\otimes_{\min}B)$ is a face of $\mathscr{S}(A\otimes_{\max}B)$ by part (a), the averaged state
\beq
\rho(\cdot)\equiv\frac{1}{|G|}\sum_{g\in G}\rho_0(g(\cdot))
\eeq
is also not in $\mathscr{S}(A\otimes_{\min}B)$. It is $G$-invariant by construction.
\end{enumerate}
\end{proof}

Another criterion for the comparison between $A\otimes_{\min}B$ and $A\otimes_{\max}B$ is the following:

\begin{prop}
\label{unbounded}
$A\otimes_{\max}B=A\otimes_{\min}B$ if and only if there is some $\lambda\in\R$ with
\beq
||x||_{\max}\leq\lambda\cdot||x||_{\min}\quad\forall x\in A\otimes_{\mathrm{alg}}B.
\eeq
\end{prop}

\begin{proof}
Again, the ``only if'' part is trivial since one can just choose $\lambda=1$. For the ``if'' part, note that the assumption implies the equivalence of norms
\beq
||x||_{\min}\leq ||x||_{\max}\leq\lambda\cdot||x||_{\min},
\eeq
which means that the completions of $A\otimes_{\mathrm{alg}}B$ with respect to these norms actually coincide, as vector spaces. Hence the comparison map~\ref{natcomp2} has trivial kernel, and is therefore a $*$-isomorphism. (In particular, this implies that one can take $\lambda=1$.)
\end{proof}

In other words, if $A\otimes_{\min}B$ and $A\otimes_{\max}B$ are different, then they are so different that the set of quotients $\frac{||x||_{\max}}{||x||_{\min}}$ with $x\in A\otimes_{\mathrm{alg}}B$ is unbounded.

\section{Maximal group $C^*$-algebras}
\label{maxgroup}

We have seen that there are two canonical ways to define a tensor product of $C^*$-algebras: a concrete version $\otimes_{\min}$, where $A\otimes_{\min}B$ directly comes equipped with a faithful representation induced by respective faithful representations for $A$ and $B$; and an abstract version $\otimes_{\max}$, where $A\otimes_{\max}B$ can be defined in terms of a universal property.

We have encountered reduced group $C^*$-algebras in example~\ref{minmax}. These are defined directly in terms of a faithful representation. Which raises the question, does the dichotomy between concrete and abstract also pertain to the theory of group $C^*$-algebras? In other words, is there an abstract version of group $C^*$-algebras? The answer turns out to be affirmative, and this section introduces the corresponding concept of \emph{maximal group $C^*$-algebra}. 

In the following, we write $\mathcal{U}(A)$ for the group of unitary elements of a unital $C^*$-algebra $A$. All group we consider will automatically assumed to be discrete, i.e.~without any topology.

\begin{prop}
\label{maxgroupunivprop}
Let $G$ be any group. Then there is a $C^*$-algebra $C^*(G)$ together with a unitary representation $\eta_G:G\ra\mathcal{U}(C^*(G))$ having the following universal property: for any unitary representation $\pi:G\ra\mathcal{U}(\mathcal{A})$ on a unital $C^*$-algebra $A$, there is a unique $*$-homomorphism $\widehat{\pi}:C^*(G)\ra A$ such that the diagram
\beq
\xymatrix{G\ar[rr]^\pi\ar[d]_{\eta_G} && \mathcal{U}(\mathcal{A}) \ar@{^{(}->}[d] \\
 C^*(G)\ar@{-->}[rr]_{\exists ! \:\widehat{\pi}} && A }
\eeq
commutes.
\end{prop}

\begin{proof}
We start to construct $C^*(G)$ by taking the ordinary group algebra $\C[G]$, which is defined as the vector space of all formal linear combinations of the elements of $G$ with multiplication induced from the multiplication on $G$ in the obvious way. Like in example~\ref{minmax}, we denote the trivial formal linear combination $1\cdot g$ for an element $g\in G$ also by $\delta_g$. In this notation, the multiplication in $\C[G]$ is defined by bilinearity and the equation
\beq
\delta_g\cdot\delta_h\equiv\delta_{gh}
\eeq
There is an involution $*$ on $\C[G]$ defined to be the antilinear extension of the assignment
\beq
\delta_g\mapsto \delta_g^*\equiv\delta_{g^{-1}}.
\eeq
In other words, $*$ inverts all group elements and acts by complex conjugation on the coefficients of any formal linear combination:
\beq
\left(\sum_g x_g\delta_g\right)^*=\sum_g \overline{x}_g\delta_{g^{-1}}.
\eeq
Then by definition, we have that $\delta_e$ is the unit of the $*$-algebra $\C[G]$, and the $\delta_g$ are all unitary:
\beq
\delta_g^*\delta_g=\delta_e=\delta_g\delta_g^*.
\eeq
Hence there is a unitary representation
\beq
\widetilde{\eta}_G:G\ra\mathcal{U}(\C[G]),\qquad g\mapsto\delta_g.
\eeq
and by construction, this essentially satisfies the desired universal property. The only problematic issue is that $\C[G]$ is not a $C^*$-algebra, but only a $*$-algebra---we have not even defined a norm yet! So given some element $\sum_g x_g\delta_g\in\C[G]$, what should its norm possibly be so that the desired universal property holds after completion with respect to the norm? Since every $*$-homomorphism like $\widehat{\pi}$ is automatically norm-nonincreasing, we need to have
\beq
\left|\left|\sum_g x_g\delta_g\right|\right|_{\C[G]}\geq\left|\left|\sum_g x_g\pi(g)\right|\right|_A\quad\forall\pi.
\eeq
This should motivate the definition
\beq
\left|\left|\sum_g x_g\delta_g\right|\right|_{\C[G]}\equiv\sup_{A,\,\pi:G\ra\mathcal{U}(A)}\:\left|\left|\sum_g x_g\pi(g)\right|\right|_A.
\eeq
It is not difficult to see that this defines a $C^*$-norm. Indeed, first of all the left-hand side is finite since $||\sum x_g\pi(g)||\leq\sum |x_g|$ by unitarity. Submultiplicativity of this norm and the $C^*$-condition are immediate. The completion of $\C[G]$ is therefore a $C^*$-algebra, the \emph{maximal group $C^*$-algebra} $C^*(G)$.

Now when $\pi:G\ra\mathcal{U}(\mathcal{H})$ is any unitary representation, we get an induced $*$-algebra representation
\beq
\C[G]\ra\mathcal{B}(\mathcal{H}),\qquad \sum_g x_g\delta_g\mapsto\sum_g x_g\pi(g).
\eeq
Then by the definition of $||\cdot||_{\C[G]}$, this representation is continuous with respect to this norm, and thererfore extends uniquely to the completion $C^*(G)$ to yield a $*$-homomorphism $\widehat{\pi}:C^*(G)\ra\mathcal{U}(\mathcal{H})$. The diagram above commutes since $\widehat{\pi}(\delta_g)=\pi(g)$.
\end{proof}

Since $C^*(G)$ has been defined as the completion of $\C[G]$ with respect to the maximal $C^*$-norm on $\C[G]$, it is known as the \emph{maximal group $C^*$-algebra} of $G$. Note that the results of example~\ref{minmax} do not apply here, since~\ref{minmax} is about \emph{reduced} group $C^*$-algebras, a significantly different case. In general, $C^*_r(G)$ and $C^*(G)$ are different; they coincide if and only if $G$ is amenable~\cite[7.3.9]{Pedersen}.

Another thing to note is that the universal property of the maximal group $C^*$-algebra implies that the operation of taking the maximal group $C^*$-algebra is functorial: given groups $G$ and $H$ together with a group homomorphism $f:G\ra H$, we have that $H$ is a subgroup of the unitary group $\mathcal{U}(C^*(H))$ by definition of $C^*(H)$. Hence, $f$ also defines a group homomorphism $G\ra\mathcal{U}(C^*(H))$. But by the universal property of $C^*(G)$, this induces a $*$-homomorphism $C^*(f):C^*(G)\ra C^*(H)$.

The definition of $C^*(G)$ has been quite abstract. Is it possible to make this more concrete? For example, what is an element of $C^*(G)$, actually? This question turns out to be rather subtle. If we ask this for the case of the ordinary group algebra $\C[G]$, then the answer is rather simple: elements of $\C[G]$ are precisely the formal linear combinations $x=\sum_{g\in G}x_g\delta_g$ with finite support. In other words, an element $x\in\C[G]$ is determined by its coefficients $x_g$. A straightforward calculation shows that the coefficient $x_g$ in turn is determined by $x$ via the formula
\beq
x_g=\langle\delta_g,\pi_L(x)\delta_e\rangle
\eeq
where $\pi_L:\C[G]\ra\B(\ell^2(G))$ is the left regular representation as defined in example~\ref{minmax}. This formula can also be used to define coefficients $x_g$ for any element $x\in C^*(G)$, since by the universal property of $C^*(G)$ the left regular representation extends to $\pi_L:C^*(G)\ra\B(\ell^2(G))$. As outlined in example~\ref{minmax}, the image of this representation is precisely the reduced group $C^*$-algebra $C^*_r(G)$. So if the associated projection map
\beq
\xymatrix{C^*(G)\ar@{->>}[r] & C^*_r(G)\subseteq\B(\ell^2(G))}
\eeq
is not injective, then there are nontrivial elements $x\in C^*(G)$ with trivial coefficients, $x_g=0\:\:\forall g\in G$. In conclusion: to an element of $C^*(G)$ one can assign coefficients on the elements of $G$, but these coefficients will determine the element only if $C^*(G)=C^*_r(G)$.

By what we already know, it is not difficult to see that there are groups $G$ with $C^*(G)\neq C^*_r(G)$. If $C^*(\F_2)\neq C^*_r(\F_2)$, we are done; so let us assume that $C^*(\F_2)=C^*_r(\F_2)$ were true. Then we would deduce from the results of example~\ref{minmax} and the upcoming lemma~\ref{groupmax},
\begin{align*}
C^*(\F_2\times\F_2)=C^*(\F_2)\otimes_{\max}C^*(\F_2)&=C^*_r(\F_2)\otimes_{\max}C^*_r(\F_2)\qquad\qquad\textrm{(hypothetically)}\\[.3cm]
&\neq C^*_r(\F_2)\otimes_{\min}C^*_r(\F_2)=C^*_r(\F_2\times\F_2)
\end{align*}
and we would have that $C^*(\F_2\times\F_2)\neq C^*_r(\F_2\times\F_2)$, so that $G=\F_2\times\F_2$ would be our example. In fact, it is known that $C^*(\F_2)\neq C^*_r(\F_2)$~\cite{Lance}, so this reasoning was purely hypothetical.

\begin{lem}
\label{groupmax}
\beqn
C^*(G_1\times G_2)=C^*(G_1)\otimes_{\max} C^*(G_2)
\eeqn
\end{lem}

\begin{proof}
The left-hand side has the universal property of extending any unitary representation $G_1\times G_2\lra A$ to a $*$-homomorphism $C^*(G_1\times G_2)\lra A$. On the other hand, the right-hand side has the universal property of extending any two commuting unitary representations $G_1\lra A$ and $G_2\lra A$ to a $*$-homomorphism $C^*(G_1)\otimes_{\max} C^*(G_2)\lra A$. Since a representation $G_1\times G_2\lra A$ is the same thing as a pair of commuting representations $G_1\lra A$ and $G_2\lra A$, the assertion follows.
\end{proof}

\paragraph{Induced representations.}

When $G$ is a discrete group and $H\subseteq G$ is any subgroup, then there is a natural and simple technique for constructing a unitary representation of $G$ from a unitary representation of $H$, known by the name \emph{induction of representations}~\cite[ch.~6]{Fol}. We will introduce this technique in the proof of the following proposition. 

\begin{prop}
\label{groupinc}
Let $G$ be a discrete group and $H\subseteq G$ any subgroup. Then the natural $*$-homomorphism $C^*(H)\ra C^*(G)$ is an inclusion, i.e.~$C^*(H)$ is a $C^*$-subalgebra of $C^*(G)$.
\end{prop}

\begin{proof}
First of all, this is indeed a nontrivial statement, since although there clearly is an inclusion of complex group algebras $\C[H]\subseteq\C[G]$, it is not a priori clear why the completion of $\C[H]$ to $C^*(H)$ should not be ``bigger'' than the completion of $\C[H]$ with respect to the norm induced by the inclusion into $\C[G]$. Since the first completion is with respect to the norm arising by maximizing over all unitary representations of $H$, while the second completion is with respect to the norm arising by maximizing over all unitary representations of $G$, it is enough to show that every unitary representation of $H$ can be extended to a unitary representation of $G$, and then the two norms on $\C[H]$ coincide. We follow~\cite[ch.~6]{Fol} in outlining how to achieve this such extensions by constructing induced representations. Let $\pi:H\ra\B(\H)$ be any unitary representation. Then consider the vector space of functions
\beqn
\label{indrep}
\H_0\equiv\left\{f:G\ra\H\:\bigg|\:\: \mathrm{supp}(f)_{G/H}\:\textrm{ is finite and}\quad f(gh)=\pi(h)^{-1}f(g)\quad\forall g\in G, h\in H\right\}.
\eeqn
Here, $\mathrm{supp}(f)_{G/H}$ stands for the set of left cosets in $G/H$ on which $f$ does not identically vanish. $\H_0$ carries a natural inner product given by
\beq
\langle f_1,f_2\rangle_{\H_0}\equiv\sum_{[g]\in G/H}\langle\: f_1(g),\: f_2(g)\:\rangle_{\H},
\eeq
where the sum is over all cosets $[g]\in G/H$ with $g$ being some representative of the coset $[g]$. This definition is independent of the choice of representatives $g\in [g]$ by the covariance condition $f(gh)=\pi(h)^{-1}f(g)$ together with unitarity of $\pi$. This inner product is obviously positive definite, so that $\H_0$ becomes a pre-Hilbert space.

Furthermore, $\H_0$ carries a natural unitary representation of $G$ given by the left translation operation on which any $g'\in G$ acts as
\beqn
\label{Gact}
(g'f)(g)\equiv f(g'^{-1}g),
\eeqn
since if the function $f$ satisfies the conditions in\eq{indrep}, then so does the function $g'f$. To each $\xi\in\H$ we can naturally associate an $f_\xi\in\H_0$ by setting
\beq
f_\xi(g)\equiv\left\{\begin{array}{cl}\pi(g)^{-1}\xi & \textrm{ if }g\in H\\ 0 & \textrm{ if }g\notin H  \end{array}\right.
\eeq 
which clearly fulfills the conditions in\eq{indrep}, is linear in $\xi$, and satisfies $\langle f_\xi,f_\xi\rangle_{\H_0}=\langle \xi,\xi\rangle_\H$. Hence we have a natural inclusion of $\H$ into $\H_0$. We claim that this inclusion is compatible with the action of $H$ on $\H_0$ via $\pi$ and on $\H$ via\eq{Gact}. Indeed,
\beq
(hf_\xi)(g)=f_\xi(h^{-1}g)=\left\{\begin{array}{cl}\pi(g)^{-1}\pi(h)\xi & \textrm{ if }g\in H\\ 0 & \textrm{ if }g\notin H  \end{array}\right.,\qquad f_{\pi(h)\xi}(g)=\left\{\begin{array}{cl}\pi(g)^{-1}\pi(h)\xi & \textrm{ if }g\in H\\ 0 & \textrm{ if }g\notin H  \end{array}\right.,
\eeq
which are the same element of $\H_0$. Hence the action of $G$ on $\H_0$ naturally extends the action of $H$ on $\H$. Clearly, the same still holds upon completing the pre-Hilbert space $\H_0$ to an actual Hilbert space.
\end{proof}

\begin{cor}
Suppose that $G_1$ and $G_2$ are groups and let $H_1\subseteq G_1$ and $H_2\subseteq G_2$ be any subgroups. Then,
\beq
C^*(G_1)\otimes_{\min}C^*(G_2)=C^*(G_1)\otimes_{\max}C^*(G_2)
\eeq
implies that
\beqn
\label{subgroupiso}
C^*(H_1)\otimes_{\min}C^*(H_2)=C^*(H_1)\otimes_{\max}C^*(H_2).
\eeqn
\end{cor}

\begin{proof}
We apply lemma\eq{groupmax} and consider the diagram
\beq
\xymatrix{ C^*(H_1\times H_2)\ar@{->>}[rr]\ar@{^{(}->}[dd] && C^*(H_1)\otimes_{\min} C^*(H_2)\ar[dd] \\\\
C^*(G_1\times G_2)\ar@{=}[rr] && C^*(G_1)\otimes_{\min} C^*(G_2) }
\eeq
where the left vertical arrow is injective by proposition~\ref{groupinc}. Then by commutativity of the diagram, the upper horizontal arrow has to be injective as well. But since it is the canonical surjection from the maximal to the minimal tensor product\eq{natcomp}, it is therefore bijective, and the assertion\eq{subgroupiso} follows.
\end{proof}

\section{Formulations of Kirchberg's conjecture}
\label{appqwep}

One formulation of Kirchberg's QWEP conjecture is as follows:
\beqn
\label{qwep}
\boxed{\mathbf{QWEP}:\quad C^*(\F_2)\otimes_{\min}C^*(\F_2)\stackrel{?}{=}C^*(\F_2)\otimes_{\max}C^*(\F_2)}
\eeqn
Kirchberg has conjectured this in~\cite{Kir}, provided a long list of statements equivalent to this one, and also proved its equivalence to Connes' embedding problem, a notorious open question in the theory of von Neumann algebras (see~\cite{Con} for the original paper, or~\cite{Cap} for a recent review). In the main text, we need several (relatively simple) reformulations of\eq{qwep}, which will be developed in this section. 

In the following, $\Gamma$ denotes any group of the form
\beqn
\label{Gamma}
\Gamma=\underbrace{\Z_m\ast\ldots\ast\Z_m}_{k\textrm{ factors}}\qquad\textrm{with}\quad k,m\geq 2;\qquad\Gamma\neq\Z_2\ast\Z_2.
\eeqn
We begin with some lemmas. The first two are well-known group-theoretical statements included for the sake of completeness. They may provide a glimpse of how the methods of modern combinatorial and geometric group theory may be of high relevance to the QWEP conjecture.

The first lemma, commonly dubbed ``ping-pong lemma'' due to the way in which the group $G$ ``plays ping-pong'' with the elements of the set $X$, exists in many different variant. We use the one most convenient for our purpose.

\begin{lem}[Ping-pong lemma]
\label{ppl}
Let $l\geq 2$ and $G$ be a group with generators $g_1,\ldots,g_l$ acting on a nonempty set $X$. Suppose that there are disjoint subsets
\beq
X_1^-,\ldots,X_l^-,X_1^+,\ldots X_l^+\subseteq X
\eeq
such that
\begin{align}
\begin{split}
\label{ppc}
& g_i\left(X\setminus X_i^-\right)\subseteq X_i^+ \\
& g_i\left(X\setminus X_i^+\right)\subseteq X_i^-
\end{split}
\end{align}
for all $i=1,\ldots,l$. Then $G$ is free with generators $g_1,\ldots,g_l$.
\end{lem}

\begin{proof}
This follows for example from the version in~\cite[Prop.~12.2]{LS} by induction on $l$; the required premise $g_i^2\neq e$ holds since, by\eq{ppc}, $g_i(X_j^+)\subseteq X_i^+\:\forall i,j$, and all $X_i^-$ and $X_i^+$ are nonempty.
\end{proof}

The ping-pong lemma is of great usefulness in proving that certain subgroups of groups are free. We now directly apply this result in order to derive the second basic lemma.

\begin{lem}
\begin{enumerate}
\item\label{sga} Any $\Gamma$ as in\eq{Gamma} has a subgroup isomorphic to $\F_2$.
\item\label{sgb} For any $n\in\N$, the group $\F_2$ has a subgroup isomorphic to $\F_n$.
\end{enumerate}
\label{ubgroups}
\end{lem}

\begin{proof}
The proof of this lemma covers three different cases which work in very similar ways.
\begin{enumerate}
\item[\ref{sga}] We distinguish the case $m=2$ from the case $m\geq 3$, starting with the latter. It is sufficient to consider the case $k=2$, so that $\Gamma=\Z_m\ast\Z_m$. Writing $a$ and $b$ for the cyclic generators of $\Gamma$, we consider the four classes $X_1^-,X_1^+,X_2^-,X_2^+$ of elements of $\Gamma$ which are defined as containing all those elements which can be written as reduced words of the respective forms
\beqn
b^{-1}x,\quad ax,\quad a^{-1}x,\quad bx,
\eeqn
where $x=e$ is allowed. Then with the definition
\beq
g_1\equiv ab,\qquad g_2\equiv ba,
\eeq
the elements $g_1$ and $g_2$ generate a subgroup which acts on $\Gamma$ by left multiplication. Since the premises of lemma~\ref{ppl} are satisfied, we conclude that the subgroup generated by $g_1$ and $g_2$ is a free subgroup $\F_2\subseteq\Gamma$.

Now for the case $m=2$. It is sufficient to consider $k=3$, so that $\Gamma=\Z_2\ast\Z_2\ast\Z_2$. Then elements of $\Gamma$ are words with letters from the alphabet $\{a,b,c\}$, and a word is reduced if and only if it does not contain the same letter in two neighboring positions. Similar to the previous case, we take $X_1^-,X_1^+,X_2^-,X_2^+$ to be defined to contain, respectively, all group elements given by reduced words of the forms
\beq
cbx,\quad abx,\quad bcx,\quad acx.
\eeq
The group elements
\beq
g_1\equiv abc,\qquad g_2\equiv acb
\eeq
generate a subgroup which acts on $\Gamma$ by left multiplication. Again, the hypotheses of lemma~\ref{ppl} are satisfied, and we conclude that this subgroup is a free subgroup $\F_2\subseteq\Gamma$. 

\item[\ref{sgb}] Taking $a$ and $b$ to be the generators of the group, we can similarly consider $X_1^-,\ldots X_n^-,X_1^+,\ldots,X_n^+$ to be the classes of elements given by reduced words of the forms
$$
b^{1} x,\quad\ldots,\quad b^{n} x,\quad a^1 x,\quad\ldots,\quad a^n x,
$$
respectively. Then lemma~\ref{ppl} applies to the subgroup generated by the elements
\beq
g_i\equiv a^ib^{-i},\quad i=1,\ldots,n ,
\eeq
and we conclude that this subgroup is a free subgroup $\F_n\subseteq\F_2$.
\end{enumerate}
\end{proof}

Actually, taking $n\ra\infty$ in the proof of the second part shows that $\F_2$ contains $\F_\infty$, the free group on a countable number of generators. And then by the first part, so does any $\Gamma$ of the form\eq{Gamma}.

We now return to group $C^*$-algebras. The next lemma states essentially that $C^*(\Gamma)$ has the \emph{lifting property}~\cite[3.9]{Oza}, although we will not have any need for using this terminology. While this result also follows from known theorems~\cite[3.20]{Oza}, we would like to offer an independent proof which might provide some additional insight into this particular case.

We take $p:\Z\ra\Z_m$ to be the canonical projection, so that its $k$-fold free product $p^{*k}:\F_k\ra\Gamma$ is a surjective group homomorphism mapping the canonical generators of $\F_k$ to the canonical generators of $\Gamma$. 

\begin{lem}
There is a ucp map $\Psi:C^*(\Gamma)\ra C^*(\F_k)$ such that the diagram
\beqn
\label{desdia}
\xymatrix{ & C^*(\F_k)\ar@{->>}[d]^{C^*(p^{*k})} \\
C^*(\Gamma)\ar@{-->}[ur]^{\Psi}\ar@{=}[r] & C^*(\Gamma) }
\eeqn
commutes.
\end{lem}

\begin{proof}
We work in the Fourier transformed picture (remark~\ref{fourier} of the main text) by using the isomorphisms
\begin{align*}
C^*(\F_k) & \cong \mathscr{C}(S^1)\ast_1\ldots\ast_1\mathscr{C}(S^1)\\
C^*(\Gamma) & \cong \C^m\ast_1\ldots\ast_1\C^m.
\end{align*}
Then the canonical surjection
\beq
C^*(p^{*k}):C^*(\F_k)\lra C^*(\Gamma)
\eeq
gets identified with the $k$-fold unital free product of the projection map
\beqn
\label{projSm}
C^*(p):\mathscr{C}(S^1)\lra\C^m,\qquad f\mapsto\left(f(e^{2\pi i\frac{1}{m}}),\ldots,f\left(e^{2\pi i\frac{m}{m}}\right)\right).
\eeqn
We will define $\Psi$ in terms of a suitable $*$-homomorphism $C^*(\Gamma)\ra M_m(C^*(\F_k))$ of which $\Psi$ will be the compression to the matrix entry $1,1$. Let us start with the case $k=1$ (albeit this case is trivial, since the algebras are commutative). Note first that specifying a unital $*$-homomorphism $\C^m\ra M_m(\C)$ is the same thing as mapping the standard basis vectors $e_i\in\C^m$ to projectors of rank $1$, such that these projectors are mutually orthogonal. Hence, the space of all these $*$-homomorphisms can be identified with the space of ordered orthonormal bases modulo phases, which is the symmetric space $U(m)/(S^1)^{\times m}$, also known as the \emph{flag variety}. In particular, this space is connected.

Furthermore, a $*$-homomorpism $\C^m\ra M_m(\mathscr{C}(S^1))\cong \mathscr{C}(S^1,M_m(\C))$ is then determined by a closed loop in the flag variety $U(m)/(S^1)^{\times m}$. By said connectedness, there is such a loop $\gamma$ having the property that at parameter $e^{2\pi i\frac{j}{m}}\in S^1$, the coresponding ordered orthonormal basis is the standard basis shifted by a cyclic permutation, $e_i\mapsto e_{(i-j+1)\!\!\mod m}$. Spelled out more explicitly, this means commutativity of the diagram
\beqn
\begin{split}
\label{lift}
\xymatrix{ & M_m(\mathscr{C}(S^1))\ar@{->>}[dd]^{M_m(C^*(p))} \\
 \C^m\ar@{^{(}->}[dr]_s\ar[ur]^(.45){\gamma} \\
 & M_m(\C^m) }
\end{split}
\eeqn
where the lower diagonal $*$-homomorphism is $s:e_i\mapsto\sum_j e_{jj}\otimes e_{(i-j+1)\!\!\mod m}$; in particular, the image of $s$ lies inside the subalgebra of diagonal matrices. The compression of $s$ to its matrix $1,1$-entry is $\mathrm{id}_{\C^m}$, and the compression of $\gamma$ to the matrix $1,1$-entry is a ucp map $\C^m\ra\mathscr{C}(S^1)$ which is a section for the projection map\eq{projSm}. 

We now proceed to the case of general $k$. Taking the $k$-fold unital free product of\eq{lift}, we obtain the left triangle of the diagram
\beq
\xymatrix@=10pt{ & M_m(\mathscr{C}(S^1))\ast_1\ldots\ast_1 M_m(\mathscr{C}(S^1)) \ar@{->>}[dd]^{M_m(C^*(p))^{*k}}\ar@{->>}[rr] && M_m(\mathscr{C}(S^1)\ast_1\ldots\ast_1\mathscr{C}(S^1))\ar@{->>}[dd]^{M_m(C^*(p^{*k}))} \ar@{->>}[rr]^(.53){1,1} && \mathscr{C}(S^1)\ast_1\ldots\ast_1\mathscr{C}(S^1) \ar@{->>}[dd]^{C^*(p^{*k})} \\ 
 \C^m\ast_1\ldots\ast_1\C^m\ar@{^{(}->}[dr]_(.45){s^{*k}}\ar[ur]^(.4){\gamma^{*k}} \\
 & M_m(\C^m)\ast_1\ldots\ast_1 M_m(\C^m) \ar@{->>}[rr] && M_m(\C^m\ast_1\ldots\ast_1\C^m) \ar@{->>}[rr]^{1,1} && \C^m\ast_1\ldots\ast_1\C^m }
\eeq
Here, the middle horizontal $*$-homomorphisms are the natural surjections projecting the unital free product to the free product with amalgamation over $M_m(\C)$, so that the middle square commutes due to the universal property of free products, and ``$1,1$'' stands for compression to the $1,1$-entry.

We now start from the $\C^m\ast_1\ldots\ast_1\C^m$ on the left and claim that the composition of $s^{*k}$ with the lower two arrows is the identity; this implies that the diagram contains\eq{desdia} as a subdiagram if one takes $\Psi$ to be the composition of $\gamma^{*k}$ with the upper two horizontal arrows, which proves the assertion.

In order to see that the composition of $s^{*k}$ with the lower horizontal arrows is the identity, we note first that its composition with the first horizontal arrow has its image inside the diagonal matrices, so that taking the compression to the $1,1$-entry gives a $*$-homomorphism. Since the restriction of this $*$-homomorphism to each factor $\C^m$ coincides with the natural inclusion, this $*$-homomorphism actually is the identity of $\C^m\ast_1\ldots\ast_1\C^m$, as has been claimed.
\end{proof}

\begin{thm}
\label{FGamma}
The following conjectural statements are equivalent:
\begin{enumerate}
\item\label{Fa} QWEP in the form\eq{qwep},
\beq
C^*(\F_2)\otimes_{\min}C^*(\F_2)=C^*(\F_2)\otimes_{\max}C^*(\F_2).
\eeq
\item\label{Fb} The equation
\beq
C^*(\Gamma)\otimes_{\min}C^*(\Gamma)=C^*(\Gamma)\otimes_{\max}C^*(\Gamma)
\eeq
holds for all $\Gamma$ of the form\eq{Gamma}.
\item\label{Fc} The equation
\beq
C^*(\Gamma)\otimes_{\min}C^*(\Gamma)=C^*(\Gamma)\otimes_{\max}C^*(\Gamma)
\eeq
holds for some $\Gamma$ of the form\eq{Gamma}.
\end{enumerate}
\end{thm}

\begin{proof}
\begin{enumerate}
\setlength\itemindent{20pt}
\item[\ref{Fa}$\Rightarrow$\ref{Fb}:] By lemma~\ref{ubgroups}\ref{sgb} and corollary~\ref{subgroupiso}, the assumption~\ref{Fa} implies that the equality
\beq
C^*(\F_k)\otimes_{\min}C^*(\F_k)=C^*(\F_k)\otimes_{\max}C^*(\F_k).
\eeq
holds as well. Now consider the diagram, using the same notation as in the preceding lemma,
\beq
\xymatrix{ C^*(\Gamma)\otimes_{\max}C^*(\Gamma) \ar@{->>}[rr] \ar[dd]_{\Psi\otimes_{\max}\Psi} && C^*(\Gamma)\otimes_{\min} C^*(\Gamma) \ar[dd]^{\Psi\otimes_{\min}\Psi} \\\\
C^*(\F_k)\otimes_{\max}C^*(\F_k) \ar@{=}[rr] && C^*(\F_k)\otimes_{\min}C^*(\F_k) }
\eeq
Then $\Psi\otimes_{\max}\Psi$ is injective, since it is a section of
\beq
C^*(p)\otimes_{\max}C^*(p):C^*(\F_k)\otimes_{\max}C^*(\F_k)\lra C^*(\Gamma)\otimes_{\max}C^*(\Gamma).
\eeq
This implies injectivity of the upper horizontal arrow by commutativity of the diagram.
\item[\ref{Fb}$\Rightarrow$\ref{Fc}:] Trivial.
\item[\ref{Fc}$\Rightarrow$\ref{Fa}:] In view of lemma~\ref{ubgroups}\ref{sga}, this is an application of corollary~\ref{subgroupiso}.
\end{enumerate}
\end{proof}

\newpage

\bibliographystyle{halpha}
\bibliography{tsirelsons_problem}

\end{document}